\documentclass[11pt]{article}
 
\usepackage{graphicx}
\usepackage{amsthm}
\usepackage{amsmath}
\usepackage{amssymb}
\usepackage{bm} 
\usepackage{bbm}
\usepackage{mathrsfs}
\usepackage{mathtools}
\usepackage{dsfont}
\usepackage{float}
\usepackage{bm}
\usepackage{natbib}
\usepackage{makecell}
\usepackage{longtable}
\usepackage{caption} 
\setcitestyle{authoryear,open={(},close={)}}

\usepackage{graphicx}
\usepackage{adjustbox}
\usepackage{booktabs}
\usepackage{multirow}
\usepackage{caption}
\usepackage[dvipsnames]{xcolor}
\usepackage{comment}
\usepackage{nicefrac}
\usepackage{enumitem}
\usepackage{apptools}
\AtAppendix{\counterwithin{lemma}{section}}
\AtAppendix{\counterwithin{corollary}{section}}
\AtAppendix{\counterwithin{proposition}{section}}
\AtAppendix{\counterwithin{theorem}{section}}
\AtAppendix{\counterwithin{figure}{section}}
\AtAppendix{\counterwithin{table}{section}}
\AtAppendix{\counterwithin{remark}{section}}
\AtAppendix{\counterwithin{equation}{section}}
\usepackage{subcaption}
\usepackage{tabularx}
 
\theoremstyle{definition}

\newtheorem{lemma}{Lemma}

\usepackage{xspace}
\usepackage[nolist,smaller]{acronym}  
\PassOptionsToPackage{hyphens}{url}\usepackage[colorlinks=true, allcolors=blue]{hyperref}

\newcommand{\acro}[1]{\textsc{#1}\xspace }

\newcommand{\LASSO}{\acro{\smaller LASSO}}
\newcommand{\GLASSO}{\acro{\smaller GLASSO}}

\newcommand{\GOLAZO}{\acro{\smaller GOLAZO}}

\newcommand{\BIC}{\acro{\smaller BIC}}
\newcommand{\EBIC}{\acro{\smaller EBIC}}

\newcommand{\MCMC}{\acro{\smaller MCMC}}
 
\newcommand{\COVID}{\acro{\smaller COVID-19}}

\newcommand{\siGGM}{\acro{\smaller siGGM}} 

\newcommand{\ESS}{\acro{\smaller ESS}}

\newcommand{\R}{\textit{R}}

\def\*#1{\bm{#1}} %\def\*#1{#1}
\DeclareMathOperator*{\argmax}{arg\,max} % Jan Hlavacek
\DeclareMathOperator*{\argmin}{arg\,min}

\def\@fnsymbol#1{\ensuremath{\ifcase#1\or *\or \dagger\or \ddagger\or
   \mathsection\or \mathparagraph\or \|\or **\or \dagger\dagger
   \or \ddagger\ddagger \else\@ctrerr\fi}}
\newcommand{\ssymbol}[1]{^{\@fnsymbol{#1}}}
 
\let\OLDthebibliography\thebibliography
\renewcommand\thebibliography[1]{
  \OLDthebibliography{#1}
  \setlength{\parskip}{0pt}
  \setlength{\itemsep}{3pt plus 0.3ex}
}
 
\def\*#1{\bm{#1}} %\def\*#1{#1}
\usepackage[left=2.3cm, right=2.3cm, top=3cm]{geometry}
\usepackage{authblk}
 
\title{\bf  {Graphical model inference with external network data}}
\author[1,2,*]{Jack Jewson}
\author[3]{Li Li}
\author[2,4]{Laura Battaglia}
\author[5]{Stephen Hansen}
\author[1,2]{David Rossell}
\author[1,2,6]{Piotr Zwiernik}
 
\affil[1]{Department of Business and Economics, Universitat Pompeu Fabra, Barcelona, Spain}
\affil[2]{Data Science Center, Barcelona School of Economics, Spain}
\affil[3]{School of Economics, Sichuan University, China}
\affil[4]{Department of Statistics, University of Oxford, UK}
\affil[5]{Department of Economics, University College London, UK}
\affil[6]{Department of Statistical Sciences, University of Toronto, Canada}
\affil[*]{Correspondence address jack.jewson@upf.edu}
 
\date{November 2023}

\begin{document}
 
%\doparttoc % Tell to minitoc to generate a toc for the parts
%\faketableofcontents % Run a fake tableofcontents command for the partocs
 
%\part{} % Start the document part
%\parttoc % Insert the document TOC
 
%\bibliographystyle{natbib}
 
\def\spacingset#1{\renewcommand{\baselinestretch}%
{#1}\small\normalsize} \spacingset{1}

\setcounter{Maxaffil}{0}
\renewcommand\Affilfont{\itshape\small}

\spacingset{1.42} % DON'T change the spacing!

\maketitle
\begin{abstract}
We consider two applications where we study how dependence structure between many  variables is linked to external network data.
We first study the interplay between social media connectedness and the co-evolution of the \COVID pandemic across USA counties.
We next study study how the dependence between stock market returns across firms relates to similarities in economic and policy indicators from text regulatory filings.
Both applications are modelled via Gaussian graphical models where one has external network data.
We develop spike-and-slab and graphical \LASSO frameworks to integrate the network data, both facilitating the interpretation of the graphical model and improving inference.
The goal is to detect when the network data relates to the graphical model and, if so, explain how.
We found that counties strongly connected on Facebook are more likely to have similar \COVID evolution (positive partial correlations), accounting for various factors driving the mean.
We also found that the association in stock market returns depends in a stronger fashion on economic than on policy indicators.
The examples show that data integration can improve interpretation, statistical accuracy, and out-of-sample prediction, in some instances using significantly sparser graphical models.
\end{abstract}

\noindent % 
{\it Keywords:}  Graphical model, Network data, Spike-and-slab, COVID-19, Stock market, Social media
%\noindent%
%{\it Keywords:}  3 to 6 keywords, that do not appear in the title
%\vfill
 
\section{Introduction} \label{sec:intro}
 
We consider two motivating applications where one seeks to learn the dependence structure (partial correlations) across many variables, and is specifically interested in assessing whether said dependence is associated to multiple external network datasets.
First, we study the dependence between \COVID infection rates across USA counties, and whether said dependence is linked to network data measuring Facebook connections between counties.
This is an important question because individuals who are connected in social networks tend to have similar backgrounds and to be exposed to similar information. Such a shared background may lead to similar attitudes towards health prevention, and hence similar infection risks. 
For example, \cite{allcott2020polarization} found that political beliefs were strongly tied to behaviour during the \COVID pandemic, more specifically that Republicans practised less social distancing.
It is hence important to study the association between social media and health outcomes.
As described in more detail below, a study by \cite{kuchler:2021} found a link between {\it marginal correlations} in infection rates between counties and the Facebook index.
We propose a probability model that can describe whether and how {\it partial correlations} depend on said index, as well as two other networks related to geographical distance and flights passenger traffic. The latter two are meant to help disentangle the effect of two counties being connected on Facebook and their being geographically close or their being major travel between them, i.e. more direct contacts.
As a preview of our findings, Figure \ref{fig:glasso_covid} (top) shows estimated (residual) partial correlations between each county pair vs. their geographical closeness and the Facebook index, Figure B.6 contains the corresponding plot for the flights  network.
Counties that are highly connected on Facebook have a higher proportion of positive partial correlations, whereas for those lowly connected most non-zero partial correlations are negative. Geographically close counties also tend to have positive partial correlations while there does not appear to be a strong relationship between the estimated partial correlations and the flight passenger network.
The bottom panels show our spike-and-slab model relating the network data to the probability that a partial correlation is non-zero, and to the mean and variance of the non-zero partial correlations.

As a second application, we study the dependence of stock market excess returns across companies, incorporating external data on similarities between companies in their exposure to economic and policy risks (as defined by \cite{bakerPolicyNewsStock2019}). Said risks were extracted from text data in mandatory regulatory filings where companies must disclose potential risks, and the idea is that if two companies disclosed similar economy- or policy-related risks then it may be more likely that they have similar stock market returns. 
The applied relevance of the problem arises from the longstanding insight in finance  that the dependence among assets informs optimal portfolios \citep{markowitzPortfolioSelection1952}.  In particular, the precision matrix determines the weights across assets that minimise a portfolio's standard deviation.  Bringing optimal portfolio theory to data requires estimating high-dimensional covariance/precision matrices, which is an important barrier to its practical application \citep{eltonEstimatingDependenceStructure1973}.  A variety of approaches have been used to tackle the problem including, recently, \GLASSO \citep{gotoImprovingMeanVariance2015},
see also \citet{senneretCovariancePrecisionMatrix2016} for an empirical review.
%Moreover, \citet{senneretCovariancePrecisionMatrix2016} compare several approaches for the empirical implementation of the theory and find \GLASSO to be competitive with other approaches in the financial econometrics literature.  
A critical observation is that we seek not only to estimate the partial correlations featuring in the precision matrix, but also to portray how they may depend on the text-based economic and policy risks, to shed light onto the joint behavior of stock market returns.
%To the best of our knowledge, such a study has never been conducted.

\begin{figure}[!ht]
\begin{center}
\includegraphics[width =0.49\linewidth]{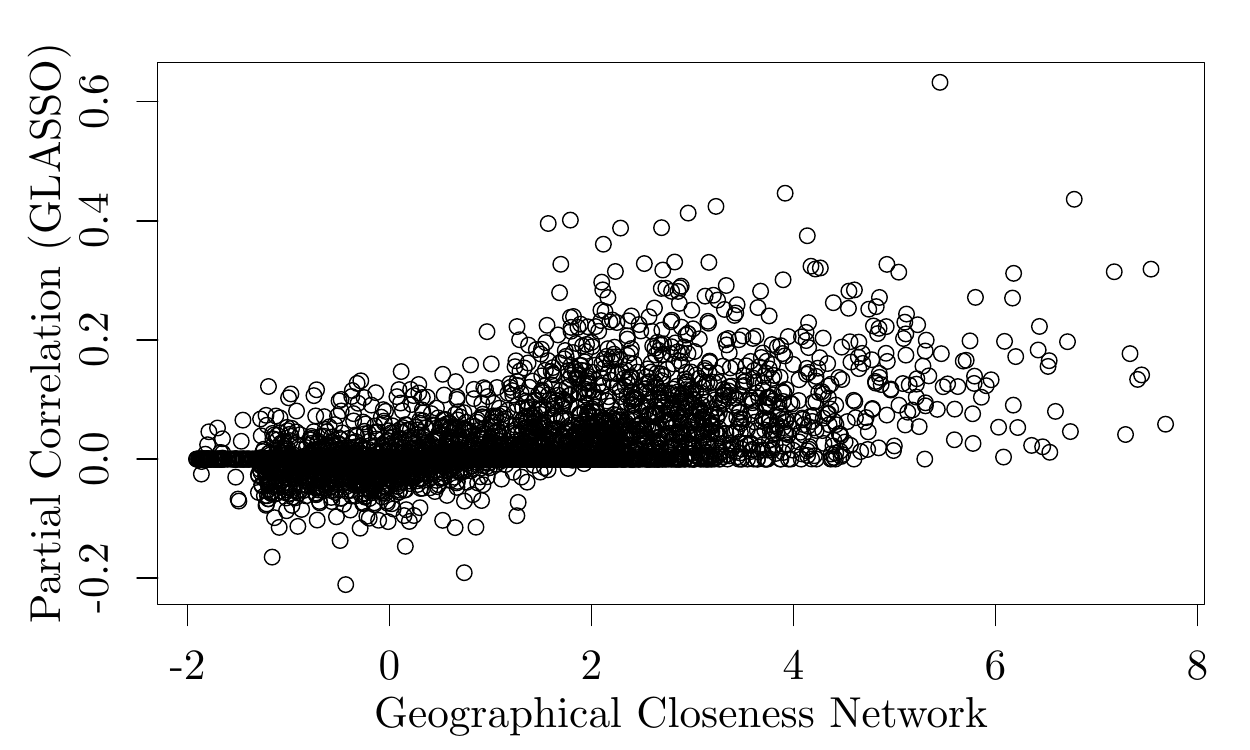}
\includegraphics[width =0.49\linewidth]{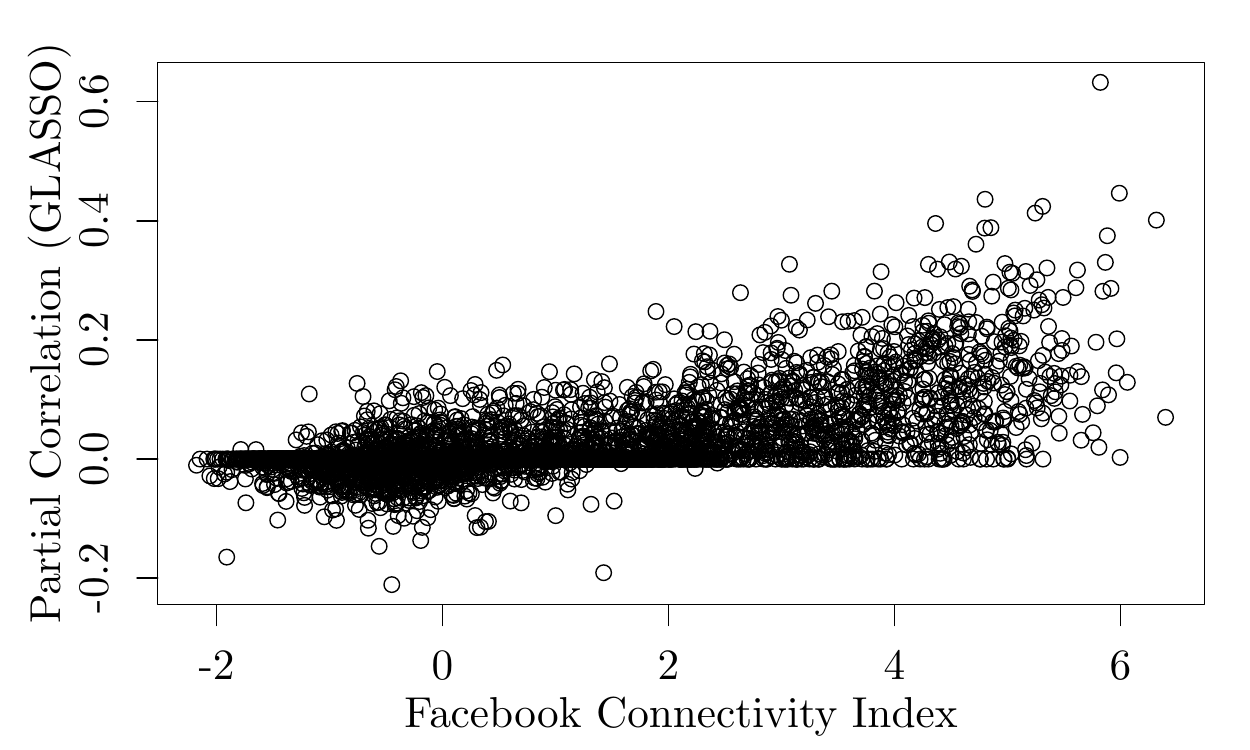}\\
\includegraphics[width =0.49\linewidth]{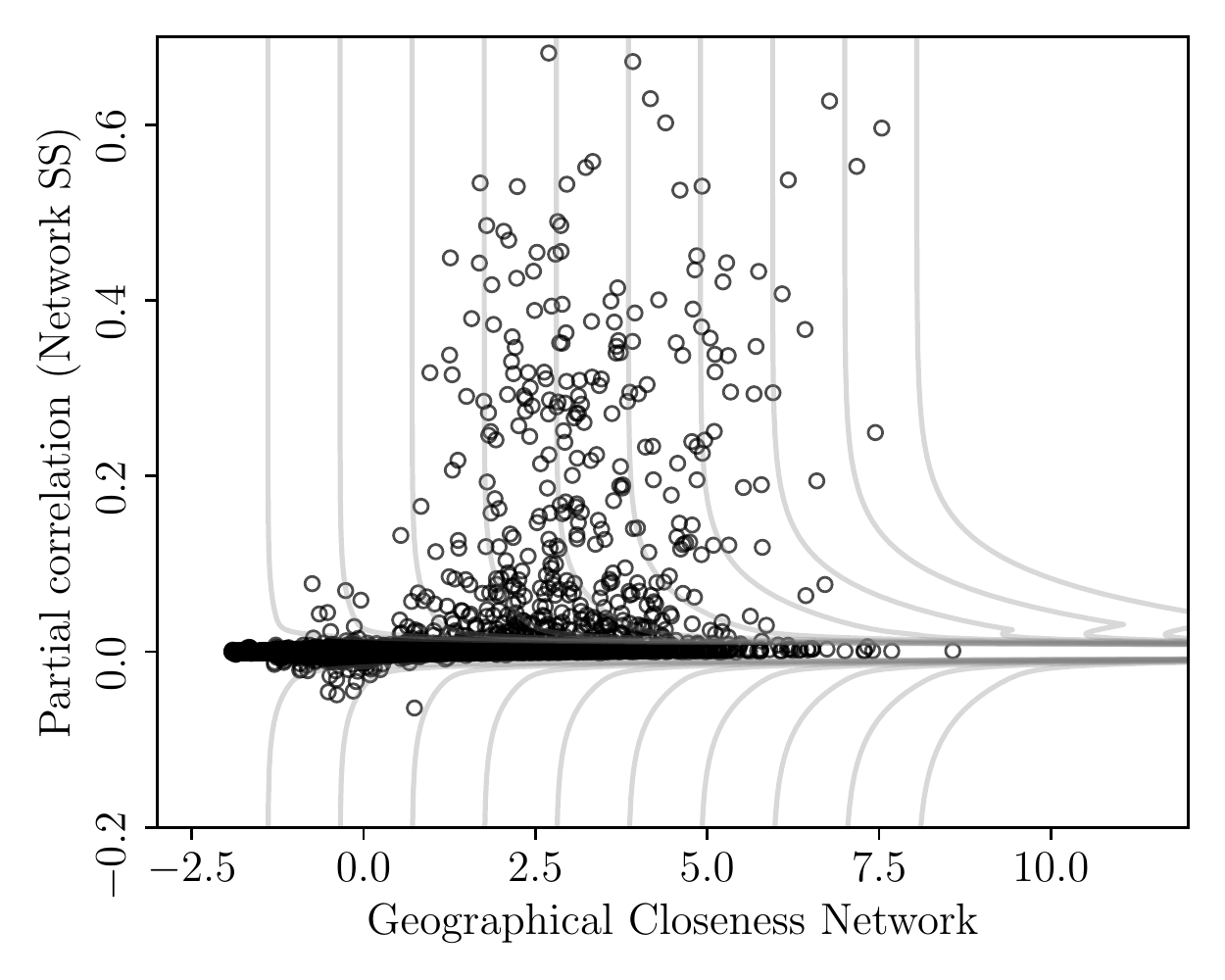}
\includegraphics[width =0.49\linewidth]{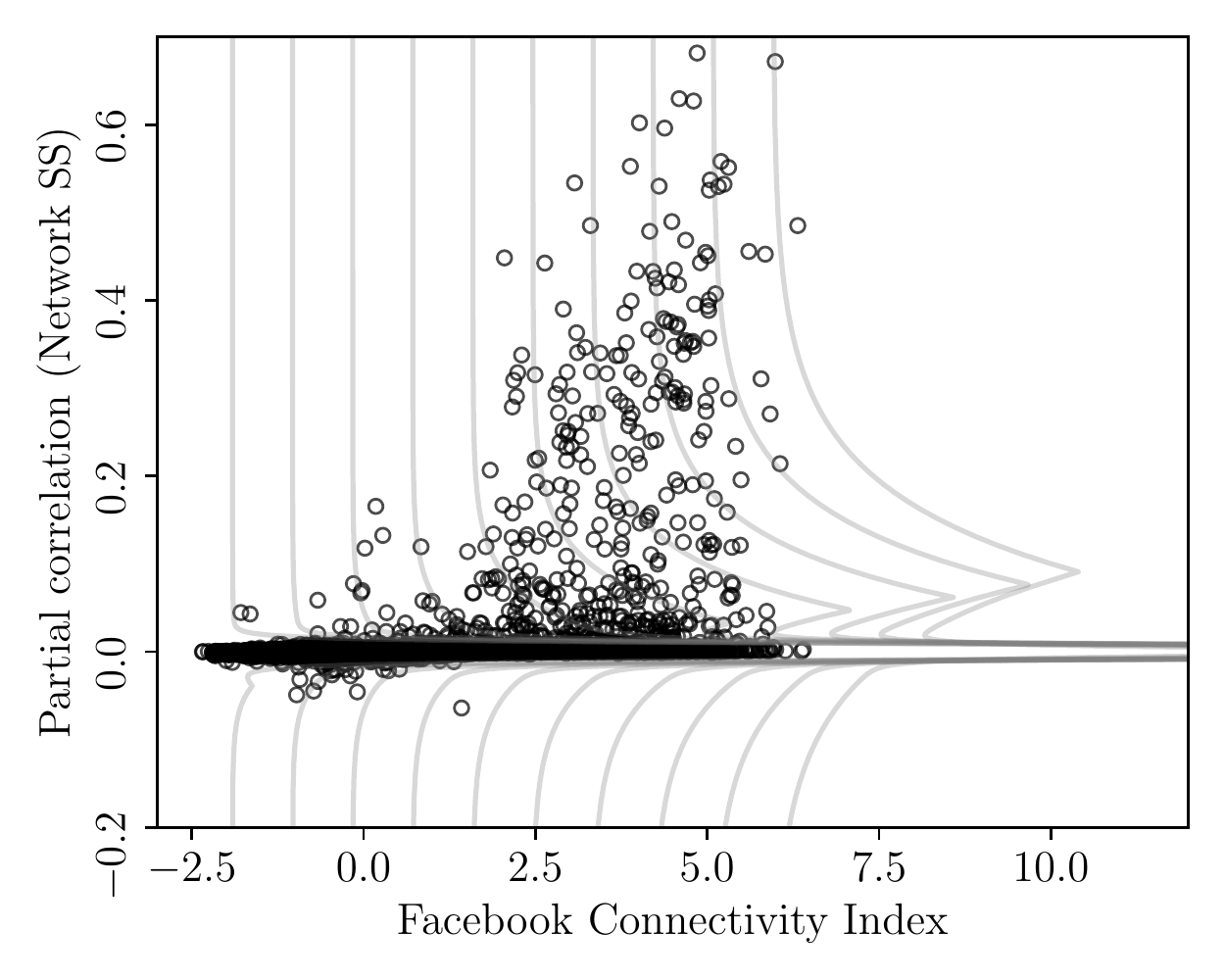}\\
%\includegraphics[width =0.49\linewidth]{plot/COVID_new/COVID_SS_prob_slab.pdf}
%This plot went in the appendix 
%\includegraphics[width =0.5\linewidth]{plot/COVID_new/partial corr vs facebook.pdf} 
%\includegraphics[width =0.5\linewidth]{plot/COVID_new/partial corr vs geodist.pdf}
%\includegraphics[width =0.49\linewidth]{plot/COVID_new/GLASSO_vs_A1_1_0.5_tikz-1.pdf}
%\includegraphics[width =0.49\linewidth]{plot/COVID_new/GLASSO_vs_A2_1_0.5_tikz-1.pdf}
\caption{Residual partial correlations in \COVID infections (adjusted for covariates) across counties vs Geographical Closeness Network defined as $ 1/\log(Geodistance)$ (left) and log-Facebook Connectivity Index (right). Top panel: partial correlations estimated with graphical LASSO, with penalization parameter set via \BIC. % \EBIC $\gamma_{\EBIC} = 0.5$. 
Bottom panel: fitted spike-and-slab distributions and fitted partial correlations estimated with network graphical spike-and-slab LASSO. %Bottom panel: prior slab probability as a function of both networks.
}
\label{fig:glasso_covid}
\end{center}
\end{figure}
 
We use Gaussian graphical models (GGMs) and extensions discussed later as a convenient framework that describes the dependence among random variables in an interpretable manner, providing a suitable basis for our applications.
There are however certain challenges that led us to develop a methodological framework that is another main contribution of this paper, and can be applied to numerous applications other than those considered here.

A first applied challenge is that the ease with which one can interpret the output of a graphical model deteriorates as the number of variables $p$ gets large, i.e. there are simply too many edges to read them one by one.
Our proposed model provides a way to regress the probability of an edge being present, as well as the mean and variance of the associated (non-zero) partial correlation, on external network data. Said regression helps understand when one can expect an edge to be present, and to have a certain sign and magnitude, as illustrated in Figure \ref{fig:glasso_covid}.
A second challenge is that in our applications the sample size $n$ is moderate relative to the $p(p+1)/2$ covariance parameters.
By integrating external network data one hopes to improve the accuracy of the inference, provided said data carries useful information regarding the graphical model. Our framework provides natural novel strategies to assess whether the network data is indeed useful.

%Despite numerous theoretical and methodological advances, an important practical limitation is that they require the estimation of an inherently large number of parameters, which can be challenging unless the sample size is large enough. The main idea behind our work is that there are numerous applications where external data provides valuable information to help guide the graphical model selection and estimation, and hence improve their accuracy. 
%We propose a frequentist and a Bayesian framework to exploit this complementary information.
%This proposes a new paradigm in high-dimensional statistics based on exploiting as much as possible the complementary information. %DAVID: this is not a new paradigm, strictly speaking, people from the data integration literature have used similar ideas in regression for example (discussed in our intro)
 
%Of particular interest to us are situations where the external data comes in the form of a network between the variables being studied.
%\stephen{I wonder if this gets into our specific application too soon, maybe mention a few relevant examples and point out the big challenge is how to incorporate network data into estimation?  And that we provide a means for doing so.  The COVID-19 thing is just to illustrate the method, not the point of the paper.}
 
To our knowledge, there are no model-based methods to incorporate multiple network-valued external data in undirected graphical models. %\jack{Is this now not true?}
There has been, however, active research on incorporating external data in regression.
For example, \cite{stingo:2010} proposed a multivariate regression of gene expression on micro-RNA, where the prior probabilities that micro-RNAs have a non-zero coefficient depend on an external biological and structural similarity score.
%\cite{ni:2019} proposed a Bayesian framework each individual is assigned a different DAG, and the DAG structure depends on the covariates of that individual, performing both edge and covariate selection.
%\cite{hoff:2012} consider a covariance regression model where each individual has a different covariance matrix, and said covariance depends on the covariates of that individual. Extending such a model to the precision matrix does not seem trivial, further our interest is in a common covariance for all individuals.
Similarly, 
\cite{stingo:2011} incorporated pathway information into regression models for gene expression,
\cite{quintana_ma:2013} proposed a Bayesian variable selection framework where prior inclusion probabilities depend on meta-covariates, %with applications in genomics, 
\cite{cassese:2014} a multivariate regression of gene expression versus copy number variations that incorporates their physical distance in the genome,
\cite{peterson2016joint} a regression framework using a network for covariate penalisation,
and \cite{chiang:2017} a brain activity vector auto-regression that incorporates external brain information. 
%\cite{guha:2020} proposed a Bayesian shrinkage prior to regress the mean of a univariate outcome on network-valued covariates, encouraging similar regression coefficients for covariates that are connected in the network.
\cite{chen_tinghuei:2021} predicted disease outcomes given single nucleotide polymorphisms, where the \LASSO regularisation parameter depends on functional annotation categories.

There has also been work incorporating network data in graphical models, primarily in neuroscience.
\cite{ng2012novel, pineda2014guiding, higgins2018integrative} considered penalised likelihood GGMs to understand co-activation across brain regions, where one has strong grounds to believe that external network data extracted from known brain structure provides useful information.
In a similar vein, \cite{bu2021integrating} use distances between brain regions to drive the regularisation of a GGM that is fit via multiple univariate regressions, and provide theoretical conditions for asymptotic learning of the GGM's structure.
The main applied difference with our setting is that we wish to assess whether the network data are informative and, if so, depict how. Another difference is that we consider multiple network datasets (e.g. Facebook, distance, flights), rather than only one. In simulations we illustrate how assessing whether the network data are useful or not can lead to significant practical improvements.
As discussed, the main methodological difference is that we develop a probabilistic spike-and-slab model to regress the GGM on the network data that helps interpret the presence of edges and the sign and magnitude of partial correlations.
%The main novelty in our work is that we incorporate external network data to model the (inverse) covariance rather than the mean, i.e. we use graphical models to study the dependence structure. 
%Also, we allow the external data to inform not only prior inclusion probabilities or overall regularisation but also the location and variance of non-zero parameters. 
This is important in our applications, e.g. the bottom panels in Figure \ref{fig:glasso_covid} depict that large Facebook connectivity is associated with positive partial correlations.

We develop two frameworks to integrate network data into GGM selection and parameter estimation. 
The first framework is a hierarchical extension of the graphical \LASSO (\GLASSO) (\cite{friedman:2008,yuan2007model}, see also \cite{wang_hao:2012} for a discussion of Bayesian counterparts).
The framework largely follows that in \cite{ng2012novel}, except that we learn critical hyper-parameters from data and assess whether each network data is actually useful or not.
%The construction resembles the \GLASSO priors of \cite{khondker:2013}, where each precision matrix entry has a different regularisation parameter, with the important difference that we allow the latter to depend on external network data.
We also develop tailored optimisation algorithms that build on the \GOLAZO algorithm of \cite{lauritzen:2020} so that the computational cost is similar to a standard \GLASSO problem, and we apply Bayesopt algorithms to speed up the search over hyper-parameter values.
%A limitation however is that, as mentioned, in our applications the external data appears to be not only informative about whether a parameter is zero but also about its sign.
Our second framework is the main contribution and uses a spike-and-slab prior, with the novel feature that the slab's probability, location and variance are regressed on the network data. 
%Said feature allows one to effectively regress the GGM onto the network data, aiding the interpretation of the GGM as discussed above (when to expect an edge, or positive/negative partial correlations).
To ensure its practical applicability we developed a software implementation in the probabilistic programming languages \texttt{Stan} \citep{carpenter:2017} and \texttt{NumPyro} \citep{bingham:2019,phan:2019}. The latter capitalises on efficient automatic differentiation and GPUs to help boost the computational speed. Similarly, the first framework is implemented in \textit{R}. %\david{We should put the R functions into an R package and upload it at GitHub. Similarly Laura's numpyro implementation, and perhaps Jack's STAN implementation.}
 
The paper proceeds as follows. 
Section \ref{sec:applications} discusses our motivating applications in more detail.
Section~\ref{sec:model} reviews the \GLASSO,
%and partial correlation (\PCGLASSO) \citep{carter2021partial} 
introduces our network-adjusted extension and its Bayesian analogue. Section~\ref{sec:computation} discusses our computational strategy for learning the graphical model and hyper-parameters that depict its association with the external network data. Section~\ref{sec:simulations} uses simulations to shed light on a natural practical question: what if the network data are useless, i.e. uninformative regarding the graphical model we seek to learn? We illustrate that one should assess whether the network data have useful information about the GGM and, if not, discard them to avoid deteriorating inference.
Section~\ref{sec:results} shows our main results for the \COVID and stock market applications, and Section \ref{sec:discussion} concludes.
%Both examples demonstrate that the network data is informative, the Facebook network being particularly informative about the structure of \COVID cases and the Economic risks network about stock-market dependence.
%that while similarity in Policy risks appears more predictive of two companies not being conditionally independent, similarity in Economics risks appears more predictive of the magnitude of non-zero partial correlations \jack{is this too strong?}. 
Code to implement all of our experiments and data pre-processing is available at \url{https://github.com/llaurabat91/graphical-models-external-networks}.

%\cite{wang_hao:2015} for a spike-and-slab prior strategy for edge selection %in undirected graphical models.

\section{Motivating applications}
\label{sec:applications}

\subsection{Dependence in \COVID infections versus Facebook, geographical and flight networks}
 
Studying the evolution of pandemics such as \COVID is of great importance for health, economic and societal reasons. There are many studies to forecast infections or to understand how they are related to various factors (e.g. health measures, temperature).
We consider a further important aspect that received less attention: understanding how the disease co-evolves across (possibly distant) geographical units, and what factors are associated to such co-evolution. 
For example, if several counties were expected to simultaneously exhibit higher-than-expected infection rates, health authorities might need to plan resources accordingly.
Further, identifying factors that are related to the co-evolution (e.g. the Facebook index) may suggest strategies to limit such coordinated growth (e.g. targeted information campaigns).
 
To study \COVID co-evolution across USA counties, we downloaded weekly infection rates from \cite{COVID:2020} %\url{https://github.com/CSSEGISandData/COVID-19/blob/master/csse_covid_19_data/csse_covid_19_time_series/time_series_covid19_confirmed_US.csv} 
for the period 22 January 2020 to 30 November 2021 (97 weeks total) for all USA counties ($>3,000$ in total).
We then iteratively clustered neighbouring counties with small population until all aggregated counties had at least 500,000 inhabitants, obtaining 332 aggregated counties in total. Full details of our clustering procedure are presented in Section B.3.  For simplicity onwards we refer to aggregated counties simply as counties.
The reason for clustering counties was two-fold. First, the weekly infection rates for smaller counties are subject to high variance, and hence less reliable than when grouping counties. Second, working with $>3,000$ counties results in a GGM with $>4,500,000$ parameters, which imposes serious computational bottlenecks.

We also obtained data on covariates that are thought to be associated with the disease's evolution, such as temperature, population density, vaccination rates and an index measuring the stringency of pandemic measures (\cite{Temperature:2020}; \cite{Population:2020}; \cite{Vaccination:2020}; \cite{Policy:2020}).
We defined the outcome of interest as the county log-infection rates, i.e. log infections relative to the county's population.
Our interest is in studying the disease co-evolution {\it after} accounting for factors driving the mean structure. To this end, we fitted a linear regression model that included temperature, vaccination rates, the stringency of pandemic measures, a weekly fixed effect term estimating the mean infections across all counties in that particular week, and a first-order auto-regressive term measuring the infection rate in the previous week. 
See Section B and the supplementary code for the data collection,  pre-processing, and residual checks assessing the linearity and normality assumptions, and that higher-order auto-regressive terms are not needed.
%\david{Li, can you please prepare an R notebook showing the code, plots, and commenting on what we're doing /seeing at each step in the code?}
 
Although the mean model explained most of the variance in infection rates (adjusted $R^2$ coefficient 0.942), certain county pairs were systematically both above or below the model predictions.
Specifically, we estimated partial correlations in the regression residuals for each county pair via graphical LASSO, and obtained numerous non-zero estimates (Figure \ref{fig:glasso_covid}, top).
%Motivated by \cite{kuchler:2021}, who found that {\it marginal} correlations between county infections rates were related to their Facebook connectivity index, we wish to study how {\it partial} correlations depend on the Facebook index. 
Said partial correlations indicate that certain county pairs tend to behave better or worse than expected (given the week's overall pandemic status and other covariates)  in a coordinated fashion.
Our primary goal is to assess whether this coordinated behavior occurs more frequently across counties that are strongly connected via social media, given by the Facebook index.
Said index defines a network of counties, measuring the strength of the connection between every pair of counties. 
We also consider two further networks, one based on geographical closeness (see Section \ref{ssec:covid_results}) and a second measuring flow of passengers between two counties by plane (see Section B).
%\color{blue}[DR. Describe where the flight data was obtained from.] \color{black}

We see partial correlations as an appealing measure of disease co-evolution. For example, suppose that infections in County A drive those of County B, which in turn drive those of County C, then all three counties would have non-zero marginal correlation. In contrast, the partial correlation between counties A and C would be zero, suggesting there is no direct link between them.
%As discussed, GGMs are a natural strategy to estimate partial correlations. Here we develop a framework to fit Gaussian graphical models that can incorporate the external information provided by the Facebook and geographical networks, and neglect said information when not needed.
%The Facebook, geographical and flight network datasets should help improve the graphical model estimation, e.g. regularise to a lesser degree the partial correlations for county pairs that are highly-connected in Facebook and \textit{vice versa}. 
%This desideratum led us to develop a network-regularised graphical LASSO framework, see Section \ref{sec:model}.
An important observation stemming from Figure \ref{fig:glasso_covid} is that counties that are highly connected on Facebook have a higher proportion of non-zero, and positive, partial correlations. A similar observation applies to geographical distances.
Hence one wishes not only to regularise to a lesser extent county pairs with a strong Facebook connection but also to describe how the average non-zero partial correlation depends on Facebook (or geographical, or flight) connectivity. This desideratum led us to develop a network-regularised spike-and-slab framework, where the slab's mean, variance and probability are regressed on the network, see Section \ref{sec:model}.

\subsection{Dependence in stock market returns versus text data}

Our goal is to study whether and how covariation in stock market excess returns (i.e. returns above/below those that were expected, see below) across firms is associated with firms' sharing similar risks.
To measure to what extent they do so, we downloaded text of the \textit{Risk Factors} (\textit{RF}) section of publicly traded firms' annual 10-K filings to the USA Securities and Exchange Commission. For each firm, we combine all filings made between 2015 and 2019, inclusive.  Said filings describe exhaustively future earnings risks faced by the firms, and there is an incentive for full disclosure because investors can take legal action when firms withhold information that if disclosed would have prevented financial losses. Firms that face similar risks may have more dependent stock returns, e.g. two firms mentioning risks to oil price rises may co-move when oil prices change. Indeed, \citet{hanleyDynamicInterpretationEmerging2019} regressed the covariance of excess returns between pairs of financial firms on a measure of \textit{RF} text overlap and showed a positive relationship in the lead-up to the global financial crisis in 2008.  More recently, \citet{davisFirmLevelRiskExposures2020} shows that firms with similar \textit{RF} texts reacted similarly to the arrival of \COVID.  Our analysis goes beyond these studies by modelling partial rather than marginal correlations.  We also allow distance in \textit{RF}-text space to influence both the probability of a connection between firms and the mean (and variance) of the partial correlations on the network.

We consider $p=366$ firms traded on US markets that satisfy the following conditions: i) membership in S\&P500 at the end of 2019; ii) closing stock price adjusted for stock splits and dividends available in the COMPUSTAT database for every trading day between 2 January 2019 to 31 December 2019 (252 trading days in total); iii) at least one 10-K filing available in 2014-2019.
For each trading day in 2019 we construct daily excess returns using the Fama-French three-factor model.  Specifically, we individually regress each firm's daily log-returns on the variables contained in the daily, three Fama/French factors file downloaded from Kenneth French's Data Library website.  The residual is the excess return.

To measure textual similarity between companies, we first construct a bag-of-words representation of each firm's 10-K filings during 2014-2019. We follow \citet{bakerPolicyNewsStock2019}, and compute firms' exposure to 16 separate \textit{economic} risks and 20 separate \textit{policy} risks.  For each risk $r$, \citet{bakerPolicyNewsStock2019} define a term set $T_r$ containing terms that reflect the exposure.  For example, the policy risk `food and drug policy' is captured by the term set $\{$prescription drug act, drug policy, food and drug administration, fda$\}$.\footnote{In common with the text-as-data literature, we refer here to terms even when a `term' is a multi-word expression.  See Appendix B of \cite{bakerPolicyNewsStock2019} for a complete description of the term sets associated with each risk.}  \citet{bakerPolicyNewsStock2019} show that intertemporal variation in economic and policy risk terms in newspaper articles closely tracks aggregate market volatility.  This motivates the idea of using variation in these terms across individual firms to better measure their co-movement across trading days. 

Let $x_{i,v}$ be the count of term $v$ in firm $i$'s 10-K filings during 2014-2019 and let $C_i \equiv \sum_v x_{i,v}$ be the total number of terms. We measure each firm's exposure to risk $r$ as $\log\left(1 + \sum_{v \in T_r} x_{i,v} / C_i\right)$, i.e. logarithm of 1 plus the proportion of words referring to risk $r$ out of the total $C_i$ words. We use the logarithm to account for the fact that a risk term not being mentioned at all versus being mentioned once is likely to be more informative than being mentioned many times compared with slightly more times. For each pair of firms, we then measure its similarity in exposure to economic risks by computing the correlation between the vector of economy-related risks. This defines a network between companies such that the network connection between companies $(j,k)$ is given by said correlation. We proceeded analogously to define a policy risk network by computing correlations between policy-related terms.
%\color{blue}[DR. I found the latter half of the previous paragraph very hard to follow. Steve, can we explain this in a simpler/summarised way? Perhaps then give a fuller description, showing some examples, in an appendix/supplement/markdown document at github.\color{black}

In summary, our data processing produced two networks between firms that measure their similarity in risk exposures based on a particular 
representation of \textit{RF} texts. We remark that one could use alternative text analysis tools, however our goal is to establish that text-based relational data can be useful to estimate dependence in stock returns. The optimal representation of text for this task is left as an open question.  Still, as we show below, separately controlling for economic and policy risks yields important insights regarding whether government policy generates return co-movement above and beyond that generated by firm fundamentals.
 
See Section C and the supplementary code for the data collection,  pre-processing, linear model fit, and residual checks assessing our model assumptions.
 
\section{Model} \label{sec:model}
 
We describe two model-fitting strategies to regress an undirected GGM on $p$ variables onto multiple external network datasets.
Section \ref{ssec:network_penalisation} discusses network \GLASSO, which we mainly use as a computationally-convenient framework to assess whether one should add/remove each network dataset. We also discuss a Bayesian interpretation useful to check that the assumed model fits the observed data.
%which as discussed builds on \cite{ng2012novel}, except that we learn critical hyper-parameters from data and assess whether each network data is actually useful or not
Section \ref{ssec:spikeslab_nglasso} is our main contribution, a spike-and-slab model to regress partial correlations on network data.
Section \ref{sec:nongaussian} discusses how to extend our framework beyond Gaussian data, as needed for the stock market application.

We set notation. Let $y_i \in \mathbb{R}^p$ be the outcome vector for individuals $i=1,\ldots,n$ (e.g. log-infection rates in $p$ counties at week $i$, or stock excess returns for $p$ companies at day $i$) and $x_i \in \mathbb{R}^d$ covariates (week indicator, temperature, percentage of fully vaccinated individuals in week $i$, etc.).
We assume that $y_i \sim \mathcal{N}_p\left(B x_i, \Theta^{-1}\right)$ independently across $i=1,\ldots,n$, where $B$ is a $p \times d$ regression coefficients matrix and $\Theta$ a $p \times p$ positive-definite precision (or inverse covariance) matrix.
To ensure that the independence assumption across $i$ is tenable, we include lagged versions of $y_i$ into the covariates $x_i$, as described in Section \ref{sec:applications} and B.
For simplicity, in our applications we start by subtracting the estimated mean $\hat{B} x_i$ from $y_i$, where $\hat{B}$ is the least-squares estimator,
and subsequently assume the outcomes to have zero mean, i.e. $y_i \sim \mathcal{N}_p(0, \Theta^{-1})$.

A convenient property of modelling $y_i \sim \mathcal{N}_p\left(0, \Theta^{-1}\right)$ is that conditional independence statements can be drawn from the graph defined by the non-zero elements of $\Theta$. Specifically, $(y_{ij},y_{ik})$ are independent given the remaining elements in $y_i$ if and only if $\Theta_{jk}=0$.
%Hence, by determining what elements in $\Theta$ are zero one learns about conditional independence. %\begin{equation}
%    \Theta_{jk} = 0 \quad\Leftrightarrow\quad Y_j \independent Y_k | Y_{\{1, \ldots, p\}\setminus\{j, k\}},\textrm{ for } j, k = 1,\ldots, p.
%\end{equation}
%As a result it is of interest when estimating $\Theta$ from data to ascertain which elements of $\Theta$ are 0.
As argued earlier, in our applications we use partial correlations as a measure of association.
%In the Gaussian model, conditional independence is equivalent to zero partial correlation.
We denote partial correlations by
\begin{align}
    \rho_{jk} := \textrm{corr}(y_{ij} , y_{ik} \mid  y_{i\{1, \ldots, p\}\setminus\{j, k\}})= -\frac{\Theta_{jk}}{\sqrt{\Theta_{jj}\Theta_{kk}}}.\label{Equ:partial_corr}
\end{align}
%We denote by $\Rho= (\rho_{jk})$ the $p \times p$ partial correlations matrix. 
 
Importantly, in our framework, one also observes external data in the form of $Q \geq 1$ networks between variables. These are $p \times p$ symmetric matrices $A^{(1)},\ldots,A^{(Q)}$, where
$a^{(q)}_{jk}$ measures strength of the connection between variables $(j,k)$.
In the \COVID application $a^{(1)}_{jk}$ is the geographical closeness between counties $(j,k)$, $a^{(2)}_{jk}$ their Facebook connection index, and $a^{(3)}_{jk}$ their flight connectivity.
In the stock application, $a^{(1)}_{jk}$ is the similarity between firms $(j,k)$ in their exposure to economic risks, and analogously $a^{(2)}_{jk}$ for policy risks.

\subsection{Network graphical LASSO}\label{ssec:network_penalisation}
%\subsection{Gaussian graphical models}{\label{sec:gaussian_graphical_model}}

Network graphical LASSO is a penalised likelihood framework to estimate $\Theta\in \mathcal{S}^p_{+}$ by maximising a Gaussian log-likelihood plus a graphical \LASSO (\GLASSO) penalty \citep{friedman:2008,yuan2007model}, where the magnitude of said penalty is regressed onto the network datasets.
Specifically, we consider
\begin{align}
\hat{\Theta}= \argmax_{\Theta\in \mathcal{S}^p_{+}}\;\;\; \log \det(\Theta) - \textrm{tr}(S\Theta)
-  \sum_{j\neq k} \lambda_{jk}|\Theta_{jk}|,
 \label{equ:NGLASSO_objective}
\end{align}  
where $\mathcal{S}^p_{+}$ is the set of non-negative definite matrices, $\textrm{tr}(\cdot)$ the matrix trace, $S$ the empirical covariance matrix of $(y_1, \ldots, y_n)$, 
\begin{equation}
    \lambda_{jk}= \lambda_{jk}(A^{(1)},\ldots,A^{(Q)})= \exp \left\{ \beta_0 + \sum_{q=1}^Q \beta_q a_{jk}^{(q)} \right\}
    \label{equ:network_regression_multiple}
\end{equation}
are regularisation parameters, and $\beta=(\beta_0,\ldots,\beta_Q) \in \mathbb{R}^{Q+1}$ are regularisation hyperparameters that play a critical role in determining the level of sparsity in $\hat{\Theta}$.
That is, each $\Theta_{jk}$ gets a potentially different penalty parameter $\lambda_{jk}$, which is a function of the network data $A^{(1)},\ldots,A^{(Q)}$. To simplify notation, we omit the dependence on $A^{(1)},\ldots,A^{(Q)}$ and simply use $\lambda_{jk}$, and let $A=(A^{(1)},\ldots,A^{(Q)})$.
For convenience we parameterise the penalties in terms of a scaled version of $A^{(q)}$ that is centered to have zero sample mean and unit sample variance, and which we denote by $\bar{A}^{(q)}$.
\GLASSO is the particular case where $\lambda_{jk}$ are constant across $(j,k$).
%A popular approach to estimate $\Theta$, and hence the partial correlations, is the Graphical \LASSO (\GLASSO) \citep{friedman:2008,yuan2007model}. \GLASSO produces an estimate of $\Theta$ that contains zeroes by maximising the Gaussian log-likelihood with a \LASSO penalty. Specifically,

\cite{ng2012novel} proposed the penalty in \eqref{equ:NGLASSO_objective}-\eqref{equ:network_regression_multiple}, the main difference being that we consider multiple networks ($Q>1$) and that we learn hyper-parameters $\beta$ from data, including the exclusion of some networks.
Two popular strategies to set hyper-parameters are using cross-validation \citep{friedman:2008} and information criteria such as the Bayesian information criterion (\BIC) \citep{schwarz:1978}. % and the Extended \BIC (\EBIC) \citep{chen:2008}.
The former is more suitable for predictive tasks than when seeking models that help explain the data-generating truth, e.g. cross-validation does not lead to consistent model selection even in simpler linear regression
where the \BIC and related information criteria are consistent, see \cite{foygel2010extended,zhang_yiyun:2010,wang_tao:2011,fan_yingying:2013}. 
We hence use the \BIC to learn $\beta$. Specifically, viewing $\hat{\Theta}(\beta)$ as a function of $\beta$, we choose $\beta$ minimising 
\begin{align}
%    \BIC(\lambda) &\;=\; -2\ell_n(\hat{\Theta}(\lambda)) + \big|\mathbf{E}(\hat{\Theta}(\lambda))\big|\cdot\log n, %+ 4|\mathbf{E}(\hat{\Theta}(\lambda))|\gamma \log p
    \hat{\beta}_{\BIC} := \argmin_{\beta\in\mathbb{R}^{Q+1}}
    \BIC(\beta) &\;=\; -2\ell_n(\hat{\Theta}(\beta)) + \big|\mathbf{E}(\hat{\Theta}(\beta))\big|\cdot\log n, %+ 4|\mathbf{E}(\hat{\Theta}(\lambda))|\gamma \log p
    \label{eq:bic_optim}
\end{align}
where $\ell_n(\hat{\Theta})$ is the Gaussian log-likelihood function and $|\mathbf{E}(\hat{\Theta}(\beta))|$ counts the number of edges in the graph associated with $\hat{\Theta}(\beta)$. 
Importantly, note that when $\beta_q=0$ then the $q^{th}$ network dataset is effectively excluded. The idea is that if a network dataset does not provide useful information about $\Theta$, then one may set $\beta_q=0$ to avoid adding unnecessary noise to $\hat{\Theta}$, see Section \ref{sec:simulations} for an illustration.
An alternative to the \BIC is the Extended \BIC (\EBIC) \citep{chen:2008}. % which contains an additional penalty to the \BIC for model complexity. 
%(see \eqref{eq:ebic_optim}), allowing it to maintain consistency in certain high-dimensional regression settings. 
As a sensitivity check, we provide results using the \EBIC to select $\beta$ in Sections A.5, B.8 and C.6. In our examples the \EBIC was overly conservative in selecting edges, which resulted in high mean-squared-error.
Finally, we note that there are alternative approaches to choosing $\beta$, see \cite{kuismin}, but they require more extensive computations that become prohibitive in our setting. 
We also note that alternatives to  \eqref{equ:NGLASSO_objective} include the adaptive graphical LASSO, SCAD and MCP \citep{fan:2009,wang_lingxiao:2016}, which were proposed to reduce bias in the estimation of large entries in $\Theta$. We focus on \eqref{equ:NGLASSO_objective} however due to its practical appeal of defining a concave problem for which one may establish efficient optimisation methods. 
 
One could of course consider alternative parameterisations to \eqref{equ:network_regression_multiple}, e.g. let $\lambda_{jk}$ depend non-parametrically on the network data. 
However, \eqref{equ:network_regression_multiple} requires fewer hyper-parameters than a non-parametric treatment and is easy to interpret: the log-regularisation depends linearly on the networks.
Further, a model-checking exercise suggested that \eqref{equ:network_regression_multiple} is a reasonable parameterisation for our two motivating applications.
Said model-checking is best understood by adopting a Bayesian interpretation.
The penalised estimator associated to \eqref{equ:network_regression_multiple} is equivalent to the posterior mode under independent Laplace priors \citep{wang_hao:2012} with scale parameter $1/\lambda_{jk}$, that is
%A Bayesian interpretation can be given to the \GLASSO/\PCGLASSO by considering the absolute value penalisation term as the log-density of a double exponential prior \citep{wang_hao:2012} with location parameter set to 0 and scale parameter $\sigma = 1/\lambda$
%\footnote{To be precise, the maximum a posteriori estimation matches that of \GLASSO when $\sigma = \frac{1}{n\lambda}$ \jack{not sure why it is not $\frac{n}{2}$ but experimentally it isn't}, where the $n$ captures the fact that the \GLASSO objective (Eq. \eqref{equ:GLASSO_objective}) takes as argument the sample covariance matrix, which is already averaged over $n$ data points, while the Bayesian likelihood sums over the data. \jack{Note: this is not true for \PCGLASSO - for Bayes, the likelihood is still $\Theta$ but for the Frequentist we use the empirical correlation}}. The density of $\rho_{ij} \sim \textrm{DE}(\mu, \sigma)$ where $\textrm{DE}(\mu, \sigma)$ is the double exponential distribution (a.k.a. Laplace distribution) with location parameter $\mu$ and scale parameter $\sigma$, is 
\begin{equation}
   \pi(\Theta \mid A, \beta) \;\propto \;
   \prod_{j > k} \frac{\lambda_{jk}}{2}\exp\left\{-\lambda_{jk}|\theta_{jk}|\right\} \mbox{I}(\Theta \succ 0),
   %\frac{1}{2\sigma}\exp\left\{-\frac{|\rho_{jk} - \mu|}{\sigma}\right\}
\end{equation}
where $\mbox{I}(\Theta \succ 0)$ is an indicator for $\Theta$ being positive-definite,
$\lambda_{jk}$ is as in \eqref{equ:network_regression_multiple}
% \piotr{Jack, I know there was a problem with handling positive dependence. However, if $\Theta$ is assumed to have 1s on the diagonal, we can use \cite{joe2006generating}, where they call the ``uniform distribution over correlation matrices'' -- this is simply the uniform distribution over normalised positive definite matrices, which is what you need to get $\mbox{I}(\Theta \succ 0)$.}
%Importantly, such a prior must also contain an indicator that the resulting joint distribution for the partial correlation matrix $\rho$ has positive density only on matrices that are non-negative definite.  \stephen{Maybe it's obvious, but where is the indicator in the formulation?}\piotr{Jack, perhaps \cite{joe2006generating} will be useful. HJ defines a suitable distribution over the space of correlation matrices.}
%The network adjusted model introduced in Section~\ref{ssec:network_penalisation} considers $\lambda$ to be depend on further unknown parameters $\beta$ and external matrices $A^{(1)}, \ldots, A^{(Q)}$. Such a model admits a Bayesian interpretation via the hierarchical prior
%\begin{align}
%    \rho_{ij} | \beta, A^{(1)}, \ldots, A^{(Q)} \sim \textrm{DE}\left(0, 1/\lambda_{ij}\left(A^{(1)}, \ldots, A^{(Q)}
%    , \beta\right)\right)
%\end{align}
and $\beta$ are now prior parameters. %One could for example consider hierarchical hyper-priors on parameters $\beta$.
The Bayesian interpretation is that the $\lambda_{jk}$'s arise from a Laplace random effects distribution with parameter $\beta$. %, which helps understand the implications of \eqref{equ:network_regression_multiple}.
The \textit{a priori} expected value of $\theta_{jk}$ is 0, which induces sparsity, and the prior variance is 
\begin{align}
   \textrm{Var}\left[\theta_{jk}~|~ \beta, A\right] = \mathbb{E}\left[\theta_{jk}^2~|~ \beta, A\right] = 
   \frac{2}{\lambda_{jk}^2}.
\end{align}
%Large values of $\lambda_{ij}\left(A^{(1)}, \ldots, A^{(Q)}, \beta\right)$ indicate a low prior variance for $\rho_{ij}$ and as a result larger probability of $\rho_{ij}$ being close to 0. On-the-other hand lower values of $\lambda_{ij}\left(A^{(1)}, \ldots, A^{(Q)}, \beta\right)$ indicate larger prior variance for $\rho_{jk}$ and a smaller prior probability of being close to 0. Further, the parametric form of $\lambda_{ij}\left(A^{(1)}, \ldots, A^{(Q)}, \beta\right)$ assumes that the log-variance of the partial correlations behave linearly in external network matrix $A^{(q)}$
Therefore \eqref{equ:network_regression_multiple} assumes that the log-variance of the partial covariances $\theta_{jk}$ depends linearly on the network data
\begin{equation}
    \log \mathbb{E}\left[\theta_{jk}^2~|~ \beta, A\right] = \log(2) - 2\left(\beta_0 + \beta_1 \bar{a}^{(1)}_{jk} + \ldots + \beta_Q \bar{a}^{(Q)}_{jk}\right).
\label{eq:priorvar_networkglasso}
\end{equation}
 
Provided one has an initial estimate of the left-hand side of \eqref{eq:priorvar_networkglasso}, which in our examples we derived from standard \GLASSO estimates of $\theta_{jk}$, one may check whether its relation to the network data is roughly linear. Such a check motivated taking the logarithm of the raw distances,  Facebook connectivities and flight passenger flow to define our networks for the \COVID data, while the stock market risk indicator networks required no transformations. See Supplementary Sections B.6 and C.4 for further details.

\subsection{Network graphical spike-and-slab LASSO}
\label{ssec:spikeslab_nglasso}

%\jack{mention \cite{vinciotti2022bayesian}?}
%\jack{Finally, \cite{vinciotti2022bayesian} proposed a graphical model that uses external network data to model the prior probability of an edge. We aim to also allow the external network to characterise the location and scale of the strength of non-zero connections as well - I THINK IN THE END IT DIDN'T DO THIS}
 
The network graphical LASSO in \eqref{equ:NGLASSO_objective} provides sparse point estimates of partial correlations and, via its Bayesian interpretation, describes how their variance depends on the network data.
In our applications, however, we also seek to describe how the proportion of non-zero partial correlations and their mean depend on the network.
For example, in the \COVID data both the probability that two counties are conditionally dependent and the mean partial correlation grow as their Facebook connection grows (Figure \ref{fig:glasso_covid}), and similarly for the stock market data (Figure C.5). %\david{Laura, please add a reference once you added the figure.} \jack{Want exactly same as Figure 1 put for Stock, either in the appendix or in the paper. I added a link to the GLASSO vs Network motivating plot, we still need the SS plot.}
 To address this issue, we developed a spike-and-slab framework that builds on the regression setting of \cite{rockova:2014} and the graphical spike-and-slab of \cite{gan:2018}. 
The main novelty is that both the slab prior probability and its parameters depend on network data. In particular, the slab need not be centered at zero, a feature that is novel\textemdash to our knowledge\textemdash  and has some independent interest.

We parameterise $\Theta$ in terms of partial correlations $\rho_{jk}$ in \eqref{Equ:partial_corr}, which facilitates interpretation and ensures that the posterior mode is invariant to scale transformations. By scale invariance we refer to the property that the estimated $\rho_{jk}$ remain the same regardless of whether one applies a scale transformation to the input data or not, see  \cite{carter2021partial} for a detailed discussion. %, showing that applying a \GLASSO penalty on the precision entries $\theta_{jk}$ is not scale-invariant, whereas applying it to the partial correlations $\rho_{jk}$ is.
%\piotr{In GLASSO you get scale invariance by using the sample correlation matrix rather than the sample covariance matrix. We should use a more direct argument for why we build on PCGLASSO in the Bayesian framework.}
We set a prior density $\pi(\mbox{diag}(\Theta),\rho)= \pi(\mbox{diag}(\Theta)) \pi(\rho)$, where $\sqrt{\Theta_{ii}} \sim \mathcal{IG}(a, b)$ with $a = 0.01$ and $b = 0.01$ reflecting an uninformative prior on the diagonal elements of $\Theta$, and
\begin{align}
%    \pi(\mbox{diag}(\Theta))&= \prod_{j=1}^p \mbox{IG}(\sqrt{\Theta_{jj}}; a, b) \nonumber \\
    \pi(\rho \mid \eta) &= C_\eta \mbox{I}(\rho \succ 0)
    \prod_{j > k} (1-w_{jk}) \mbox{DE}(\rho_{jk}; 0, s_0) + w_{jk} \mbox{DE}\left(\rho_{jk}; \eta_0^T a_{jk}, s_{jk} \right) \label{Equ:Spike-and-SlabPrior} \\
    w_{jk} &= \left( 1 + e^{-\eta_2^T a_{jk}}  \right)^{-1}, \quad s_{jk} = s_0 (1+ \exp\left\{\eta_1^T a_{jk} \right\}),
    \nonumber
\end{align}
where $C_\eta$ is the normalisation constant and $\mbox{I}(\rho \succ 0)$ a positive-definiteness indicator.
The spike is a double-exponential with zero mean and small scale $s_0$ meant to capture partial correlations that are practically zero, whereas the slab has larger variance $s_{jk}$ and may not be centered at zero.
The slab prior probability $w_{jk}$ follows a logistic regression on the network data $a_{jk}=(1,a_{jk}^{(1)},\ldots,a_{jk}^{(Q)})$, its mean $\eta_0^T a_{jk}$ depends linearly on $a_{jk}$ and its variance $s_{jk}$ is larger than $s_0$ by a factor that also depends on $a_{jk}$.
%\david{To keep things simple we should probably change $-\eta_1^T a_{jk}$ for $\eta_1^T a_{jk}$ in \eqref{Equ:Spike-and-SlabPrior}, then we say the following sentence.}
%{\color{blue}{
Specifically, positive entries in $\eta_0$ and $\eta_1$ indicate that the mean and variance (respectively) of the non-zero partial correlations increase for larger network connections $a_{jk}$, and similarly positive $\eta_2$ indicates a higher probability of a non-zero partial correlation for large $a_{jk}$.
%}}

We remark that because of the constraint $\mbox{I}(\rho \succ 0)$ the marginal prior $\pi(\rho_{jk})$ could be fairly different from the unconstrained density inside the product in \eqref{Equ:Spike-and-SlabPrior}, then $w_{jk}$ could not be interpreted as the prior probability of an edge, and similarly for $\eta_0^Ta_{jk}$ and $s_{jk}$. To address this issue we elicit prior parameters such that the indicator $\mbox{I}(\rho \succ 0)$ is satisfied with high prior probability, see below. 

Above $\eta=(\eta_0,\eta_1,\eta_2) \in \mathbb{R}^{3(Q+1)}$ are hyper-parameters driving the regression model of the partial correlations $\rho_{jk}$ onto the network data $a_{jk}$, and are a main quantity of interest in our applications.
%We remark that, although we adopt a Bayesian treatment, one could also take $- \log \pi(\rho)$ as a likelihood penalty and devise suitable optimisation algorithms. 
%\jack{I think we did both these strategies, we elicited priors for $\eta$ as you specify below, and then did empirical Bayes maximising the marginal posterior for these parameters given these priors}
%\jack{Elicitation allows us to do inference on hyperparameters, we use this in the experiments, in order to report inference for $\theta$ we resort to empirical Bayes, better empirical performance in experiments }
%Two popular strategies to set prior hyper-parameters such as $(\eta_0,\eta_1,\eta_2)$ in \eqref{Equ:Spike-and-SlabPrior} are either using prior elicitation to set default parameter values or an empirical Bayes framework where one maximises the marginal likelihood.
%In our examples we experimented with both strategies, finding qualitatively similar results.
%We focus discussion on the prior elicitation approach, as it provides a simpler interpretation of \eqref{Equ:Spike-and-SlabPrior}. % to capitalise on the fact that in a Bayesian framework $(\eta_0,\eta_1,\eta_2)$ are not tuning parameters devoid of interpretation, rather they determine beliefs on the prior expected number of non-zero partial correlations and their magnitude.
A standard strategy to set prior hyper-parameters such as $\eta$ in \eqref{Equ:Spike-and-SlabPrior} is an empirical Bayes framework where one maximises the marginal likelihood.
Such a framework allows us to do inference on the $\eta$'s themselves through the marginal posterior $\pi(\eta \mid y)$ and inference for $\Theta$ through the empirical Bayes posterior
\begin{align}
    \pi(\Theta| y, \hat{\eta})= f(y | \Theta)\pi(\Theta| \hat{\eta}),\nonumber%\label{Equ:Theta_EBposterior}
\end{align}
where $\hat{\eta}$ maximises the marginal posterior of $\eta$ given the data
\begin{align}
	\hat{\eta} :&= \argmax_\eta \pi(\eta \mid y)= \argmax_{\eta} \int \pi(\Theta, \eta | y) d\Theta= \argmax_{\eta} \int f(y | \Theta)\pi(\Theta| \eta)\pi(\eta) d\Theta.\nonumber%\label{Equ:eta_marginalMAP}%\nonumber
\end{align}
One could consider using the joint posterior $\pi(\Theta, \eta | y)$ for inference on $\Theta$ and $\eta$, but we found empirical Bayes to perform better in our experiments.  See \citet{giannoneEconomicPredictionsBig2021} for a related discussion on the desirability to learn the appropriate degree of sparsity from data in social science applications, and a related spike-and-slab proposal in a regression setting.
%this approach complements that of  which uses a spike-and-slab prior to infer the relevant predictors for mean outcomes in common economic and financial datasets.

%\jack{We are missing a normalisation constant associated with the positive definiteness of the prior here - the fix here is to say that we add this noramalising constant to the priro for $\eta$, \cite{wang_hao:2015} argued that such cancellation of prior norm constants is not too bad in a spike-and-slab graphical model setting -- as long as the norm constant affects hyper-parameters but not parameters (as in our case) and that we elected the priro such that the prior expected normalsiing contant was 1 with very high probability}
We next discuss our default elicitation for the prior $\pi(\rho \mid \eta)$.
The guiding principle was to set a minimally-informative prior, so that data may suitably update prior beliefs, while encouraging sparse solutions and preserving the interpretability of \eqref{Equ:Spike-and-SlabPrior}.
Briefly, we set $\pi(\eta)$ to be proportional to $C_\eta^{-1}$ times independent Gaussian prior densities on $(\eta_0,\eta_1,\eta_2)$. 
Adding the term $C_\eta^{-1}$ is a trick to simplify computation, since then $C_\eta$ drops from the posterior density $\pi(\Theta,\eta \mid y)$. \cite{wang_hao:2015} argued that such cancellation of prior normalisation constants does not adversely affect spike-and-slab priors in graphical model settings (as long as the constant affects hyper-parameters $\eta$ but not parameters $\Theta$, as in our case). % and that we elected the priro such that the prior expected normalsiing contant was 1 with very high probability
The prior on $\eta_2$ was set such that the prior mean number of edges is proportional to $p$, which induces sparsity, and the prior sample size can be thought of as 1, in analogy to the standard default Beta(0.5,0.5) prior in a Binomial experiment.
The prior on $\eta_1$ was set such that the prior mode of the slab's scale is $10 s_0$ and greater than $3 s_0$ with probability 0.99, i.e. the slab captures partial correlations of a larger magnitude than the slab.
Finally, the prior on $\eta_0$ was set such that the slab has zero prior mean and such that sampling entries of $\rho$ independently %\piotr{The entries of $\rho$ are sampled independently or $\rho$ is sampled independently of what?} 
from the double-exponential priors in \eqref{Equ:Spike-and-SlabPrior} returns a positive-definite matrix with 0.95 prior probability.
%\piotr{How is this $0.95$ obtained? But this I guess effectively implies that we sample from a relatively small neighbourhood of the identity. Is that fine?}. 
This ensures that $\pi(\rho \mid \eta)$ is similar to its unconstrained version where one drops the positive-definiteness indicator, as otherwise $w_{jk}$ cannot be interpreted as the marginal slab probability.

Figure A.1 plots the implied prior marginal distribution on the $\rho_{jk}$'s for both the \COVID and stock market applications showing that the prior concentrates at 0 but also features thick tails to capture true non-zero $\rho_{jk}$'s. 
%
%To assess the impact of these default prior choices, it is useful to display the implied prior marginal distribution on the $\rho_{jk}$'s. Figure \ref{Fig:Rho_prior_predictive_sim} shows that in both the \COVID and stock market applications most of the prior probability is contained in $\rho_{jk} \in (-0.5,0.5)$, which seems a sensible prior interval. 
%The prior concentrates significant mass around 0, which induces shrinkage, but also features thick tails, which favors capturing truly non-zero $\rho_{jk}$'s.
The corresponding posteriors (Figure A.1, bottom panels) set significant mass away from zero, suggesting that the prior shrinkage towards 0 was not excessive.
%%also displays the corresponding posteriors, where one appreciates significantly more mass at $\rho_{jk}$'s the fact that they are distinct from the priors indicates that the latter was not overly informative. 
Section A.3 provides further details and lists the hyper-parameter values used in our examples. Our code contains an implementation of our prior elicitation method.
 
%\david{Li/Laura: please add a panel to Figure \ref{Fig:Rho_prior_predictive_sim} with posterior draws of $\rho_{jk}$ across all $(j,k)$ for the \COVID and stock market data. Then the figure will allow us to see the transition from prior to posterior. :-)}
 
%priorSpecification_GLASSO_GOLAZO2_newSpike.Rmd
%\begin{figure}[!ht]
%\begin{center}
%\includegraphics[width =0.49\linewidth]{plot/GOLAZO_SS_COVID_prior_elicit_plot_tikz-1.pdf}
%\includegraphics[width =0.49\linewidth]{plot/GOLAZO_SS_STOCK_prior_elicit_plot_tikz-1.pdf}\\
%\includegraphics[width =0.49\linewidth]{plot/GOLAZO_SS_COVID_posterior_plot_tikz-1.pdf}
%\includegraphics[width =0.49\linewidth]{plot/GOLAZO_SS_STOCK_posterior_plot_tikz-1.pdf}
%\caption{Elicited prior distribution and posterior distribution for $\rho_{jk}$, $j = 1,\ldots, p$, $k < j$ for the \COVID data $p = 99$ and stock market data $p = 200$. %\jack{these are for both matrices equal to 0}%\jack{Add posterior draws to this plot for the whole $\rho$ show prior not too informative, maybe this goes in the appendix}
%}
%\label{Fig:Rho_prior_predictive_sim}
%\end{center}
%\end{figure}
 
\subsection{Beyond Gaussian data}\label{sec:nongaussian}
 
In certain applications such as our stock market example, data exhibit non-Gaussian behavior such as thick tails and asymmetries, even after taking logarithmic or similar transforms (see the normality checks in Section C.3).
To address this issue in this application we used a non-paranormal model, which can accommodate said departures from normality.
%a subset of the trans-elliptical family that accommodates skewness and thicker-than-normal tails.
The distribution of $y_i=(y_{i1},\ldots,y_{ip})$ is non-paranormal if there exist strictly increasing functions $f_j:\mathbb R\to \mathbb R$ for $j=1,\ldots,p$ such that the vector $f(y_i):=(f_1(y_{i1}),\ldots,f_p(y_{ip}))$ is Gaussian.
Such a non-paranormal model may be estimated by first obtaining an estimate $\hat{f}$ from the data, for which we used the \textit{R} package \texttt{huge} \citep{zhao2012huge}, and subsequently applying our methodology to the transformed data $\hat{f}(y_i)$.

An interesting property of the non-paranormal family is that the graphical model can be interpreted as in the Gaussian case. The partial correlation between the transformed $f_j(y_{ij})$ and $f_k(y_{ik})$ is zero if and only if $(y_{ij},y_{ik})$ are conditionally independent.
Partial correlations retain an interesting interpretation in the trans-elliptical family: 
zero partial correlation $\rho_{jk}=0$ indicates that $y_{ij}$ is linearly independent with any transformation of $y_{ik}$ \citep{rossell2021dependence}.

\section{Computation and inference} \label{sec:computation}
 
\subsection{Network \GLASSO} \label{ssec:optim_rho}
 
We first describe how to optimise \eqref{equ:NGLASSO_objective} for a fixed $\beta$, and subsequently how to estimate $\hat{\beta}$. The main idea is that, since $\lambda_{jk} = \lambda_{jk}\left(A, \beta\right)$ are fixed for a fixed $\beta$, the network \GLASSO objective in \eqref{equ:NGLASSO_objective} is a special case of the \GOLAZO class of models in \cite{lauritzen:2020}.
Motivated by the desire to penalise positive and negative partial correlations differently, \GOLAZO algorithms consider Gaussian graphical models with likelihood penalties of the form 
\begin{equation}
    \sum_{j=1}^p\sum_{k\neq j}\max\left\{L_{jk}\rho_{jk}, U_{jk}\rho_{jk}\right\},\label{equ:GOLZAO_general}
\end{equation}
where $-\infty \leq L_{jk}\leq 0 \leq U_{jk}\leq \infty$ are fixed. Noting that $\lambda|x| = \max\left\{-\lambda x, \lambda x \right\}$ for positive $\lambda$, we see that the penalty in \eqref{equ:NGLASSO_objective} is in the form of \eqref{equ:GOLZAO_general}  with $L_{jk}= - \lambda_{jk}$ and $U_{jk}= \lambda_{jk}$.
% \begin{equation}
%     L_{jk} = -\lambda_{jk}\left(A^{(1)}, \ldots, A^{(Q)}, \beta\right) \textrm{ and } U_{jk} = \lambda_{jk}\left(A^{(1)}, \ldots, A^{(Q)}, \beta\right).
% \end{equation}
%Penalising the Gaussian log-likelihood with \eqref{equ:GOLZAO_general} leads to 
\eqref{equ:NGLASSO_objective} is a convex problem that can be efficiently solved using a block-coordinate ascent algorithm 
%\cite{lauritzen:2020} show that it can be efficiently solved using a block-coordinate descent algorithm 
similar to that proposed for  \GLASSO in \citep{banerjee2008model}. An \R{} package is provided for \GOLAZO at \url{https://github.com/pzwiernik/golazo}.
 
%Analogously to \eqref{eq:ebic_optim}, we set $\beta$ by minimising the \BIC
%%Given an efficient algorithm (our experiments show this to be extremely efficient even for large $p$) for estimating $\hat{\rho}\left(\lambda\left(A^{(1)}, \ldots, A^{(Q)}, \beta\right)\right)$ for fixed estimates $\beta$. Our first proposal for estimating $\beta$ is minimising the \EBIC
%\begin{align}
%    \hat{\beta}_{\BIC} := \argmin_{\beta\in\mathbb{R}^{Q+1}} ~ \BIC\left(\lambda\left(A, \beta\right)\right)\label{equ:BIC_optim_beta}
%\end{align}

Obtaining $\hat{\beta}_{\BIC}$ requires maximising $\BIC(\beta)$ in \eqref{eq:bic_optim}. As usual when using information criteria to set regularisation parameters, this is a non-concave function of $\beta$ that exhibits discontinuities. We propose two optimisation approaches. In cases where only one or two external networks are available and $p$ is moderate ($p \leq 200$, say) we propose a grid-search akin to that used to set the regularisation parameter in standard \GLASSO. Section A.1 contains several analytic upper bounds to facilitate such a search. However, the dimension of the hyper-parameter $\beta$ grows with the number $K$ of external networks, hence grid searches are very costly when $K \geq 3$ and $p$ is large. % and $\BIC(\beta)$ requires the estimate $\hat{\Theta}$ which is computationally expensive for larger $p$. 
In these settings, we propose using Bayesian optimisation. Briefly, Bayesian optimisation first evaluates the objective function $\BIC(\beta)$ at a few values of the hyper-parameter $\beta$ and uses a Gaussian process to estimate $\BIC(\beta)$ for all $\beta$. Next, an acquisition function to propose new $\beta$ values at which to evaluate $\BIC(\beta)$, which are then used to update the Gaussian process estimate. 
In particular we use the \R~ package \texttt{rBayesianoptimisation} \citep{yan2016rbayesianoptimization}, with the `ucb' acquisition function and maximum function evaluations as 15 + 5$Q$, where $Q$ is the number of considered networks.
In our examples Bayes optimisation returned virtually identical results to a grid search, but incurred a significantly lower computational cost when $\BIC(\beta)$ is hard to evaluate by requiring many fewer evaluations compared with the grid search alternative.

\subsection{Spike-and-slab}{\label{Sec:Spike-and-SlabUncertainty}}
 
The full parameter of interest is $(\mbox{diag}(\Theta),\rho,\eta)$, where $\eta=(\eta_0,\eta_1,\eta_2)$ are the hyper-parameters in \eqref{Equ:Spike-and-SlabPrior}.
To approximate their posterior distribution $\pi(\mbox{diag}(\Theta),\rho,\eta \mid y)$ given the data $y$ we used Hamiltonian Monte Carlo (see \cite{neal:2011} for a review).
Specifically, we developed an \texttt{R} implementation using the \texttt{Stan} software \citep{carpenter:2017}, as well as a Python implementation using the \texttt{NumPyro} package \citep{phan:2019}.
Sections A.2 and A.4 describe further implementation details and our code provides both implementations. The purpose of the \texttt{R} version is to make our methods available to the ample \texttt{R} community, whereas \texttt{NumPyro} provides significant computational savings by using improvements in automatic differentiation and enabling the use of GPUs. 
%The savings were substantial, where \texttt{Stan} required days \texttt{NumPyro} required minutes to run essentially the same model-fitting algorithm, see Section \ref{sec:stan_vs_numpyro}.
%\david{Laura / Jack: please add the runtime comparisons.}\jack{plan to run $p = 10$ simulations and hope the comparison is stark enough} \jack{Somewhere we need to talk about the implementation of the spike-and-slab as well!}
The savings were substantial, Section D demonstrates that greater than an order of magnitude speed up was possible even in simple experimental settings.
 
The output of both implementations are $N$ posterior samples $(\mbox{diag}(\Theta^{(i)}), \rho^{(i)}, \eta^{(i)})$ for $i=1,\ldots,N$ that can be used to approximate the posterior distribution or suitable summaries such as the marginal posterior mean and standard deviation of any parameter. Of particular interest to us is to estimate the posterior probability for the presence of an edge between any two nodes $(j,k)$, i.e. that the partial correlation $\rho_{jk}$ was generated by the slab in \eqref{Equ:Spike-and-SlabPrior}.
We next discuss how to estimate said posterior probability using the posterior samples.
 
To ease notation re-write the prior as
\begin{align}
    \pi(\rho_{jk} \mid \eta)% = (1-w_{jk}(\eta)) \mbox{DE}(\rho_{jk}; 0, s_0) + w_{jk}(\eta) \mbox{DE}\left(\rho_{jk}; \eta_0^T x_{jk}, s_0 (1+ \exp\left\{ - \eta_1^T x_{jk} \right\}) \right)\nonumber\\
	%\pi(\rho_{jk}) &= (1- w_{jk})\pi(\rho_{jk}; spike) + w_{jk} \pi(\rho_{jk}; slab)\\
	&= (1- w_{jk}(\eta))\pi_0(\rho_{jk} \mid \eta) + w_{jk}(\eta) \pi_1(\rho_{jk} \mid \eta)\label{Equ:Spike-and-SlabConditioning}
\end{align}
where $\pi_0(\rho_{jk} \mid \eta)$ is the spike prior density, $\pi_1(\rho_{jk} \mid \eta)$ the slab prior density, and $w_{jk}(\eta)$ the slab prior probability.
The idea is that any $\rho_{jk}$ generated by the spike takes a near-zero value, i.e. the partial correlation is either truly zero or small enough to be practically irrelevant.
%Further, for the purposes of the model, being in the slab or the spike is considered to capture the presence or absence of an edge respectively. 
Let $z_{jk}=1$ indicate that $\rho_{jk}$ was generated from the slab and $z_{jk}=0$ that it was generated from the spike, i.e. $P(z_{jk}=1 \mid \eta)= w_{jk}$. A measure of evidence in favor of the presence of the edge is the posterior probability
\begin{align}
    P(z_{jk} =1 \mid y)= \int P(z_{jk}=1 \mid \rho_{jk}, \eta) \pi(\rho_{jk}, \eta \mid y) d\rho_{jk} d\eta, \label{eq:postprob_edge}
\end{align}
%\jack{we don't have $\eta^{(i)}$ cause we do empirical Bayes }
where from Bayes rule 
%As a result, for any $\rho_{jk}$, we can reverse the conditioning of the spike-and-slab formulation \eqref{Equ:Spike-and-SlabConditioning} using Bayes rule and calculate
\begin{align}
P(z_{jk} =1 \mid \rho_{jk}, \eta)= 
\frac{w_{jk}(\eta) \pi_1(\rho_{jk} \mid \eta)}{(1-w_{jk}(\eta)) \pi_0(\rho_{jk} \mid \eta) + w_{jk}(\eta) \pi_1(\rho_{jk} \mid \eta)}\label{Equ:slab_post_prob}.
%p(\textrm{ slab } | \rho_{jk}, \eta) &= \frac{p(\textrm{ slab }| \eta)\pi(\rho_{jk} | \textrm{ slab}, \eta) }{p(\textrm{ spike } |  \eta) \pi(\rho_{jk} | \textrm{ spike}, \eta) + p(\textrm{ slab } | \eta)\pi(\rho_{jk} | \textrm{ slab}, \eta) }\nonumber\\
%&= \frac{w(\eta)_{jk}\pi(\rho_{jk} | \textrm{ slab }) }{(1- w(\eta)_{jk})\pi(\rho_{jk} | \textrm{ spike })  + w(\eta)_{jk}\pi(\rho_{jk} | \textrm{ slab }) },
\end{align}
%where $p(\textrm{ slab } | \rho_{jk}, \eta)$ captures the probability that an edges is present between nodes $j$ and $k$ given $\rho_{jk}$ was estimated. 
%The values of $\rho_{jk}$ and $\eta$ that we condition upon could come from the posterior for these values $\pi(\rho, \eta |y)$, producing a fully Bayesian estimate of the probability. More simply however, they could from estimates $\hat{\eta}$ and $\hat{\rho}_{jk}$ outlined in \eqref{Equ:eta_marginalMAP} and \eqref{Equ:Theta_marginalMAP} respectively.
 
Given $B$ posterior samples from $\pi(\rho,\eta \mid y)$, \eqref{eq:postprob_edge} may be easily estimated by 
%This suggests two natural estimators for \eqref{eq:postprob_edge}. First,
\begin{align}
    \hat{P}(z_{jk} =1 \mid y)= \frac{1}{N} \sum_{I=1}^N P(z_{jk}=1 \mid \rho_{jk}^{(I)}, \eta^{(I)})
    \label{eq:pp_mcestimate}
\end{align}
%where $(\rho_{jk}^{(I)}, \eta^{(I)})$ for $I=1,\ldots,N$ are posterior samples from $\pi(\rho_{jk}, \eta \mid y)$.
%A second and slightly simpler alternative is to obtain the posterior mode estimators $\hat{\eta}$ and $\hat{\rho}_{jk}$ outlined in \eqref{Equ:eta_marginalMAP} and \eqref{Equ:Theta_marginalMAP} respectively, and then use
%\begin{align}
%    \hat{P}(z_{jk}=1 \mid y)= P(z_{jk} =1 \mid \hat{\rho}_{jk}, \hat{\eta}).
%    \label{eq:pp_pmestimate}
%\end{align}
The description above applies in a full Bayesian treatment where $\eta$ has a posterior distribution, in our empirical Bayes framework we simply replaced $\eta$ by $\hat{\eta}$ in \eqref{Equ:Spike-and-SlabConditioning}-\eqref{eq:pp_mcestimate}.
 
Our decision rule is to include edge $(j,k)$ whenever $\hat{P}(z_{jk}=1 \mid y) \geq t$ for some threshold $t \in [0,1]$. We used $t=0.95$.
In problems where the goal is to estimate $\Theta$ it is customary to use $t=0.5$, see \cite{barbieri:2004}. In contrast, in structural learning where one seeks to control the posterior expected false discovery proportion below some given level $\alpha$, \cite{mueller:2004} showed that the optimal threshold maximising statistical power is to set the largest $t$ such that
\begin{align}
    \frac{1}{|D|} \sum_{(j,k) \in D} \hat{P}(z_{jk}=0 \mid y) \leq \alpha
    \nonumber
\end{align}
where $D$ is the set of included edges.
%That is, one sets $t$ such that the average $\hat{P}(z_{jk}=0 \mid y)$ over included edges is below $\alpha$.
In particular, setting $t=1 - \alpha$ ensures that the posterior expected false discovery proportion is below $\alpha$.

\subsection{Empirical Bayes} \label{ssec:ebayes}
 
%As explained in Section \ref{ssec:spikeslab_nglasso}, we use the marginal posterior distributions $\pi(\eta \mid y)$ to provide inference on $\eta$, and the empirical Bayes posterior
%\begin{align}
%    \pi(\Theta| y, \hat{\eta})= f(y | \Theta)\pi(\Theta| \hat{\eta}).\label{Equ:Theta_EBposterior}
%\end{align}
%provide inference for $\Theta$. In \eqref{Equ:Theta_EBposterior} $\hat{\eta}$ maximises the marginal posterior of $\eta$ given the data
%\begin{align}
%	\hat{\eta} :&= \argmax_\eta \pi(\eta \mid y)= \argmax_{\eta} \int \pi(\Theta, \eta | y) d\Theta= \argmax_{\eta} \int f(y | \Theta)\pi(\Theta| \eta)\pi(\eta) d\Theta.\nonumber%\label{Equ:eta_marginalMAP}%\nonumber
%\end{align}
%
The empirical Bayes estimate $\hat{\eta}$ discussed in Section \ref{ssec:spikeslab_nglasso} requires marginalizing the joint posterior $\pi(\Theta, \eta \mid y)$. This is possible given $N$ posterior samples $(\Theta^{(i)},\eta^{(i)})$ for $i=1,\ldots,N$ from the latter, since then by definition $\eta^{(i)}$ are samples from $\pi(\eta \mid y)$. 
Then one may obtain $\hat{\eta}$ by maximising a kernel density estimate of $\pi(\eta \mid y)$, for example. Given that the accuracy of kernel density estimators degrades as dimensionality grows, in our examples when $\mbox{dim}(\eta)>2$ we instead obtain marginal mode estimators $\hat{\eta}_j= \argmax_{\eta_j} \pi(\eta_j \mid y)$.
%\begin{align}
%    \hat{\eta}_j= \argmax_{\eta_j} \pi(\eta_j \mid y).
%\end{align}
%\jack{I wasn't sure if we wanted this separate from the empirical Bayes discussion in Section \ref{ssec:spikeslab_nglasso}}
 
%\jack{the benefit we see in separating the prior on the diagonal and off-diagonal elements, can we link in with the stan literature on this}
 
\section{Simulation study} \label{sec:simulations}
 
We conducted a simulation study to illustrate two important practical points.
First, that when the network data are informative regarding the structure of $\Theta$, incorporating said data improves inference.
Second as just as important, that when the network data are useless inference does not suffer too much.
%as proof of principle that, by incorporating external network data, one may improve graphical model inference.
To this end, we compared standard \GLASSO with the network \GLASSO of Section \ref{ssec:network_penalisation} and the network graphical spike-and-slab of Section \ref{ssec:spikeslab_nglasso} in several settings.
We also considered the \siGGM method \cite{higgins2018integrative}, which is analogous to the network \GLASSO in \eqref{equ:NGLASSO_objective} but hyper-parameters are set to enforces the assumption that the network data are related to $\Theta$, rather than learning from data whether this is the case or not.
As discussed in Section \ref{sec:computation}, the network \GLASSO hyper-parameters $\beta$ are set via the \BIC using grid-search optimisation, and the spike-and-slab hyper-parameters $\eta$ using empirical Bayes.
%In all settings we generated 50 independent datasets where $y_i \sim \mathcal{N}(0,\Theta^{-1})$ independently across $i=1,\ldots,n$, with $\Theta_{jj}=1$, upper- and lower-diagonal terms $\Theta_{jk}=\Theta_{kj}=1$ for $j=k+1$, and $\Theta_{jk}=0$ otherwise, and there is a single binary network matrix $A$.
We considered a setting where there is a single binary network $A$ with entries $a_{jk} \in \{0,1\}$ and considered $p \in \{10,50\}$ and sample sizes $n \in \{100,200\}$ (results for $n=500$ are in Table A.2).
We then generated 50 independent datasets where $y_i \sim \mathcal{N}(0,\Theta^{-1})$ independently across $i=1,\ldots,n$.
We set the data-generating $\Theta$ to have unit diagonal and most non-zero entries along the main tri-diagonal ($\Theta_{jk}$ where $|j-k|=1$).
Specifically, a proportion of 0.95 of the tri-diagonal entries were set to non-zero values uniformly spaced in $[0.2,0.5]$.
Regarding entries outside the main tri-diagonal (i.e. $\Theta_{jk}$ where $|j-k|>1$), a proportion of $0.5/p$ were set to non-zero values uniformly spaced in $[-0.1,0.1]$
%For our simulations we first defined a tri-diagonal graph adjacency matrices $H$ of dimension $p \in \{10, 50\}$ with upper- and lower-diagonal terms $h_{jk}=h_{kj}=1$ for $j=k+1$, and $h_{jk}=0$ otherwise. From $H$ we generate $\Theta$, where for $h_{jk} = 1$, $\theta_{jk} = 0$ with probability 0.05 and with probability 0.95 $\theta_{jk}$ is spaced evenly between 0.2 and 0.5. For $h_{jk} = 0$, $\theta_{jk} = 0$ with probability $1 - \frac{0.5}{p}$, and with probability $\frac{0.5}{p}$, $\theta_{jk}$ is spaced evenly between -0.1 and 0.1. 
(i.e. the number of edges in the graphical model grows linearly with $p$).  
%
%We considered a number of variables $p = 10,20,50$ and sample sizes $n = 100, 200, 500$.
%\david{Li, are we still using 50 simulations, or perhaps a bit more? If so please update.}
%\li{Yes,we are still using 50 simulations}
 
%The primary purpose of the study is two-fold. First, to show that if the network data provide no valuable information about the graphical model structure, then our methods perform similarly to standard \GLASSO where the network data is not used. Second, to show that as the network becomes more informative, the inference provided by our methods improves gradually. 
We consider a setting where the network data are useless (independent network), and two settings where they are increasingly informative.
To measure the degree to which the network data $a_{jk} \in \{0,1\}$ is informative we count the proportion of overlaps where $a_{jk} = \mbox{I}(\Theta_{jk} \neq 0)$, 
i.e. the presence/absence of an edge in the network $A$ matches that of an edge in $\Theta$.
We considered the following settings:
%\begin{enumerate}
%    \item Independent network: the proportion of $a_{jk}=\mbox{I}(\Theta_{jk} \neq 0)$ is that expected by chance when $a_{jk}$ and $\mbox{I}(\Theta_{jk} \neq 0)$ are independent.
%    
%    \item Mildly informative network: the proportion of $a_{jk}=\mbox{I}(\Theta_{jk} \neq 0)$ is 0.75.
%    
%    \item Strongly informative network: the proportion of $a_{jk}=\mbox{I}(\Theta_{jk} \neq 0)$ is 0.85. 
%\end{enumerate}
\begin{enumerate}
\item Independent network: The tri-diagonal elements of A are set such that half of them are 1 and half of them 0, equally for the elements outside the main tri-diagonal, half of these are 1 and half of these are 0. This led to a 0.533 and 0.502 proportion of edges that agree between $A$ and $\mbox{I}(\Theta \neq 0)$ for $p = 10$ and $50$ respectively. 
\item Mildly informative network: The tri-diagonal elements of A are set such that the proportion $a_{jk} = 1$ is 0.75, alternatively for the elements outside the main tri-diagonal the proportion of $a_{jk} = 1$ is 0.25. This led to a 0.778 and 0.747 proportion of edges that agree between $A$ and $\mbox{I}(\Theta \neq 0)$ for $p = 10$ and $50$ respectively.
\item Strongly informative network: The tri-diagonal elements of A are set such that the proportion $a_{jk} = 1$ is 0.85, alternatively for the elements outside the main tri-diagonal, the proportion of $a_{jk} = 1$ is 0.15. This led to a 0.867 and 0.844  proportion of edges that agree between $A$ and $\mbox{I}(\Theta \neq 0)$ for $p = 10$ and $50$ respectively.
\end{enumerate}

Code to reproduce our simulations is available in the GitHub repository.
For each setting, we report the mean squared estimation error (MSE), the false discovery rate (FDR), and the false negative rate \citep{benjamini:1995}. The FDR is the expected proportion of false positive edges among the edges estimated to be present, a measure of type I error, whereas the FNR is the expected proportion of false negative edges among those not reported to be present, which measures statistical power.
Under the \GLASSO methods, an edge is declared if the corresponding estimate of $\rho_{jk}$ was non-zero (rounded to 5 decimal places). For the spike-and-slab model an edge is declared when the posterior probability that $\rho_{jk}$ arises from slab \eqref{Equ:slab_post_prob}, conditional on empriical Bayes estimates $\hat{\eta}$ is above 0.95. %So doing seeks to control the FDR to less than 0.05.

Table \ref{tab:sim_results0.95} summarises the results. For all sample sizes, the network \GLASSO significantly reduced the MSE when the network data were mildly or strongly informative ($A_{0.75}$ and $A_{0.85}$), whereas it attained a similar MSE to standard \GLASSO in the uninformative network setting $(A_{ind})$.
The FDR was significantly above the usually accepted level of 0.05.
%In contrast, the spike-and-slab prior attains lower FDR levels, at the cost of a slightly increased MSE. Altogether these findings suggest that our penalised likelihood framework may be more suitable for parameter estimation/prediction, whereas our Bayesian framework may be preferable in model selection or hypothesis testing problems.
Regarding the spike-and-slab formulations, they consistently achieved an FDR below 0.05 and a small FNR, and in large $p$ situations a further improvement of the MSE compared with the network-GLASSO methods. Adding network data improved the spike-and-slab MSE and FNR, particularly when $p$ was large relative to $n$. The FDR did not noticeably improve, but it was already near-zero when not using the network data. 
These findings suggest that the spike-and-slab formulations tend to attain better inference than the \GLASSO counterparts. However the latter may be more appealing in settings with pressing computational demands.
For example, in the $p=50$, $n=100$, $A_{.85}$ setting \GOLAZO took just over 5 minutes to run, whereas the \texttt{NumPyro} spike-and-slab implementation took close to 20 minutes (and \texttt{Stan} nearly 2 hours),
see Section D for further details. 

We stress that when the network data are useless ($A_{ind}$) the performance of Network \GLASSO remained similar to \GLASSO, and that of Network SS to that of a standard spike-and-slab.
In contrast the performance of \siGGM was poor in this setting, illustrating the practical value of assessing whether the network data is useful for inference, as done in our two frameworks.
In the informative network data settings the performance of \siGGM improved, although its MSE was higher than for our methodology and the FDR levels significantly above 0.05.
%We see that Network \GLASSO generally achieves a smaller MSE, considerably smaller FDR and slightly larger FNR than \siGGM. In particular the Network \GLASSO does not assume the network matrix is informative and therefore improves performance over \siGGM to the greatest extent when the network matrix is in fact uninformative. 

%\jack{
%Table \ref{tab:sim_results_higgins} \jack{appendix?} compares the network \GLASSO with the \MAP point estimation of \siGGM \cite{higgins2018integrative}. We see that Network \GLASSO generally achieves a smaller MSE, considerably smaller FDR and slightly larger FNR than \siGGM. In particular the Network \GLASSO does not assume the network matrix is informative and therefore improves performance over \siGGM to the greatest extent when the network matrix is in fact uninformative. 
%} \jack{something about uninterpretable parameetrs}

%\david{The paragraph above is an educated guess on what we'll observe, based on Li's latest results, please update accordingly.}
 
%\david{I simplified the results table to only show MSE, FDR, and FNR (the latter two defined in the standard way of \cite{benjamini:1995}. Otherwise, the results get unreadable. Li, after adding your results to Table \ref{tab:sim_results0.95}, please remove your other tables.}
 
\begin{table}[!ht]
    \centering
    %\caption{Simulation results under non-informative network $A_{ind}$, mildly and strongly informative networks $A_{0.75}$ and $A_{0.85}$. Under the SS and network SS models a 0.95 posterior probability threshold was used to declare the presence of an edge}
    \caption{Simulation results under non, mildly and strongly informative networks $A_{ind}$, $A_{0.75}$ and $A_{0.85}$. For SS and network SS models edges declared when  posterior probability $> 0.95$.}
    \begin{tabular}{|cc|ccc|ccc|} \hline
    & & \multicolumn{3}{c|}{$p=10$} & \multicolumn{3}{c|}{$p=50$} \\
       & $n$ & MSE & FDR & FNR & MSE & FDR & FNR \\ \hline
  \GLASSO                     & 100 &0.350  &0.370  &0.098  &3.505 &0.442  &0.292  \\
  Network \GLASSO, $A_{ind.}$ & 100 &0.354  &0.340  &0.122  &3.623 &0.392  &0.306  \\
  Network \GLASSO, $A_{0.75}$ & 100 &0.291  &0.258  &0.093  &2.847 &0.421  &0.251  \\
  Network \GLASSO, $A_{0.85}$ & 100 &\textbf{0.170}  &0.174  &0.120  &2.246 &0.426  &0.223  \\
  SS & 100 & 0.222 & \textbf{0.000} & 0.086 & 1.611 & \textbf{0.000} & 0.023\\
  Network SS, $A_{ind.}$      & 100 & 0.237 & 0.003 & 0.082 & 1.631 & 0.004 & 0.025  \\
  Network SS, $A_{0.75}$      & 100 & 0.234 & 0.007 & 0.073 & 1.462 & 0.005 & 0.023 \\
  Network SS, $A_{0.85}$      & 100 & 0.189 & 0.047 & 0.060 & \textbf{1.280} & 0.002 & 0.022 \\ 
\siGGM, $A_{ind}$      & 100 & 0.534 & 0.683 &0.047   & 4.815 & 0.866 & 0.017 \\ 
\siGGM, $A_{0.75}$     & 100 & 0.304 & 0.492 & 0.019   & 3.203 & 0.837 & 0.010 \\ 
\siGGM, $A_{0.85}$     & 100 & 0.197 & 0.385 & \textbf{0.028}   & 2.749 & 0.794 & \textbf{0.009} \\ \hline
  \GLASSO                     & 200 &0.184  &0.416  &0.022  &1.794 &0.476  &0.181  \\
  Network \GLASSO, $A_{ind.}$ & 200 &0.201  &0.378  &0.040  &1.871 &0.439  &0.189  \\
  Network \GLASSO, $A_{0.75}$ & 200 &0.161  &0.309  &0.022  &1.515 &0.412  &0.181  \\
  Network \GLASSO, $A_{0.85}$ & 200 &0.096  &0.204  &0.098  &1.241 &0.388  &0.173  \\
  SS & 200 & 0.109 & \textbf{0.000} & 0.056 & 0.672 & 0.002 & 0.017 \\
  Network SS, $A_{ind.}$      & 200 & 0.127 & 0.007 & 0.053 & 0.671 & \textbf{0.002} & 0.017  \\
  Network SS, $A_{0.75}$      & 200 & 0.114 & 0.007 & 0.048 & 0.597 & 0.003 & 0.015  \\
  Network SS, $A_{0.85}$      & 200 & \textbf{0.091} & 0.023 & 0.041 & \textbf{0.527} & 0.002 & 0.015 \\
  \siGGM, $A_{ind}$      & 200 & 0.273 & 0.666 & \textbf{0.015}   & 2.108 & 0.839 & 0.009 \\ 
  \siGGM, $A_{0.75}$     & 200 & 0.181 & 0.487 & \textbf{0.015}   & 1.470 & 0.797 & 0.009 \\ 
  \siGGM, $A_{0.85}$     & 200 & 0.105 & 0.381 & 0.026   & 1.138 & 0.751 & \textbf{0.008} \\
  \hline
%  \GLASSO                     & 500 &0.082  &0.367  &0.002  &0.825 &0.410  &0.032  \\
%  Network \GLASSO, $A_{ind.}$ & 500 &0.085  &0.315  &0.007  &0.766 &0.443  &0.035 \\
%  Network \GLASSO, $A_{0.75}$ & 500 &0.066  &0.270  &0.000  &0.604 &0.419  &0.031  \\
%  Network \GLASSO, $A_{0.85}$ & 500 &0.045  &0.195  &0.008  &0.512 &0.386  &0.027  \\
%  SS & 500 & \textbf{0.030} & 0.000  & 0.023 & 0.198 & 0.002 & \textbf{0.009}\\
%  Network SS, $A_{ind.}$      & 500 & 0.034 & \textbf{0.000} & 0.023 & 0.201 & \textbf{0.001} & 0.010 \\
%  Network SS, $A_{0.75}$      & 500 & 0.032 & 0.002 & \textbf{0.018} & 0.193 & 0.001 & 0.009 \\
%  Network SS, $A_{0.85}$      & 500 & 0.033 & 0.008 & 0.022 & \textbf{0.183} & 0.001 & 0.009 \\
%       \hline
    \end{tabular}
    \label{tab:sim_results0.95}
\end{table}

\section{Results} \label{sec:results}
 
%\david{Do we assess MCMC convergence or ESS anywhere? We should add at least a mention about this in the main text, how many iterations we run in each example etc.}
 
%\subsection{\COVID mortality}
\subsection{\COVID infection rates}
\label{ssec:covid_results}

Recall that the outcomes are log-infection rates for USA counties during $n = 97$ weeks and that a regression model was fit to account for various factors driving the mean infection rates.
These included week and county indicators, temperature and vaccination rate and serial correlation terms, see Section \ref{sec:applications}.
%See also Section \ref{sec:covid_dataprocessing} for a description of the data pre-processing. 
The goal is to regress the residual partial correlations between counties, which measure the extent to which \COVID co-evolved in these counties, on three network datasets.
These are a geographical closeness network $A_1$ where $a_{jk}^{(1)}$ is the reciprocal of the log-geographic distance between counties $(j,k)$ (hence larger values indicate smaller distance),
a Facebook network $A_2$ where $a_{jk}^{(2)}$ is the log-Facebook connection index between $(j,k)$,
and a flight network $A_3$ where $a_{jk}^{(3)}$ is the logarithm of 1 + the flight passenger flow between $(j,k)$ (see Section B for more details).
Pearson's correlation between $A_1$ and $A_2$ is 0.746, i.e. there is a large overlap in the information given by both networks and it is hence desirable to use a principled model to disentangle their effects.
 
As a first exercise, we used network \GLASSO to determine what network datasets are informative with respect to the target partial correlations. As $p = 332$ is large and the hyper-parameter dimension is $\mbox{dim}(\beta)=4$, we estimated $\hat{\beta}_{\BIC}$ using Bayesian optimisation, as described in Section \ref{sec:computation}.
% analyze the data using four strategies: a standard \GLASSO using no network information, our network \GLASSO using only the geographical or only the Facebook networks, and using both networks.
Table \ref{tab:results_covid} shows a summary comparing the 8 models defined by the inclusion/exclusion of each network data.
The model attaining the best \BIC value includes the geographical and Facebook networks, suggesting that they both carry relevant information to help learn the graphical model, but not the flight network. 
The estimated coefficients for both networks $(\hat{\beta}_1,\hat{\beta}_2)$ were negative, i.e. counties that are close geographically or highly-connected at Facebook are regularised less. The larger coefficient $\hat{\beta}_2$ in the joint model suggests that the effect of the Facebook network is greater. 
Interestingly, the three network-regularised solutions were significantly sparser relative to the 628 edges detected by \GLASSO.

%% OLD RESULTS ONLY P = 100
%\begin{table}[!ht]
%\centering
%\caption{Four models for the \COVID data. $A_1$ and $A_2$: networks defined by $ 1/\log(Geodistance)$ and $\log(Facebook)$. \BIC values account for the extra hyper-parameters in the network \GLASSO models. 10-fold: 10-fold cross-validated log-likelihood}
%\begin{tabular}{ccccccc}
%  \hline
%  Method & \BIC & $\hat{\beta}_0$ & $\hat{\beta}_1$ & $\hat{\beta}_2$ & Edges & 10-fold\\ 
%   \hline
%   \GLASSO & 7066.313 & -1.711 &  &  &  646 & 273.61\\ 
%   Network \GLASSO - $A_1$ & 5555.422 & -0.132 & -1.211 &  &  245 & 285.00 \\ 
%   Network \GLASSO - $A_2$ & 5561.311 & 0.263 &  & -1.368 &  217 & 286.77 \\ 
%   Network \GLASSO - $A_1$ \& $A_2$ & \textbf{5525.378} & -0.056 & -0.333 & -1.000 &  235 & \textbf{287.22}\\ 
%   \hline
%\end{tabular}
%\label{tab:results_covid}
%\end{table}

\begin{table}[!ht]
\centering
\caption{Eight models for the \COVID data. $A_1$, $A_2$ and $A_3$: networks defined by $ 1/\log(Geodist)$, $\log(Facebook)$ and $A_3 = \log(1+Flights)$. \BIC values account for the extra hyper-parameters in the network \GLASSO models. 10-fold: 10-fold cross-validated log-likelihood}
\begin{tabular}{cccccccc}
  \hline
  Method & \BIC & $\hat{\beta}_0$ & $\hat{\beta}_1$ & $\hat{\beta}_2$& $\hat{\beta}_3$ & Edges & 10-fold\\ 
  \hline
  \GLASSO & 23158.558 & -1.376 &  &  &  & 2637 & 113.208 \\ 
  Network \GLASSO - $A_1$ & 16237.865 & 0.230 & -1.053 &  &  & 1427 & 122.692 \\ 
  Network \GLASSO - $A_2$ & 15207.178 & 0.738 &  & -1.301 &  & 1430 & 122.516 \\ 
  Network \GLASSO - $A_3$ & 24227.657 & -1.295 &  &  & -0.085 & 2602 & 102.794 \\ 
  Network \GLASSO - $A_1$ \& $A_2$ & \textbf{15064.079} & 1.500 & 0.355 & -1.695 &  & 1197 & \textbf{124.372} \\ 
  Network \GLASSO - $A_1$ \& $A_3$ & 16057.853 & 0.527 & -1.193 &  & 0.531 & 1377 & 119.583 \\ 
  Network \GLASSO - $A_2$ \& $A_3$ & 15217.319 & 0.493 &  & -1.131 & 0.377 & 1339 & 122.396 \\ 
  Network \GLASSO - $A_1$, $A_2$ \& $A_3$ & 15448.091 & 0.212 & -0.063 & -1.093 & -0.103 & 1598 & 121.276 \\ 
  \hline
\end{tabular}
\label{tab:results_covid}
\end{table}
 
Despite these solutions being sparser, they included some edges that were not included by \GLASSO.
Figure \ref{fig:unique_edges_covid_geodist} shows edges that were only selected when adding the geographical network $A_1$, which largely correspond to counties that are close to each other. Figure \ref{fig:unique_edges_covid_facebook} shows an analogous plot when using the Facebook network $A_2$, interestingly there are connections between faraway counties in the west, north-east and south-east. Figure B.8 further portrays the estimated graphical model when using both networks. %\jack{Section \ref{Sec:MapPlots} visualises the network given by non-zero elements of the estimated $\Theta$ on top of a map of the U.S}
 
To further assess the relative performance of the eight models, we undertook a 10-fold cross-validation exercise where we assessed the log-likelihood (as a measure of predictive accuracy) in an out-of-sample fashion.
%, iteratively leaving out different 10\% chunks of the data which are used to  assess the log-likelihood of each model fitted on the remaining 90\% of the data in an out-of-sample fashion. 
The models incorporating the Facebook and geographical network also performed much better than standard \GLASSO according to this predictive criterion, despite being remarkably sparser (1,197 vs. 2,637 edges).
 
\begin{figure}[!ht]%[hbt!]
\centering
\begin{subfigure}[b]{\linewidth}
    \includegraphics[trim= {0.0cm 0.5cm 0.0cm 0.5cm}, clip, width = 1\linewidth]{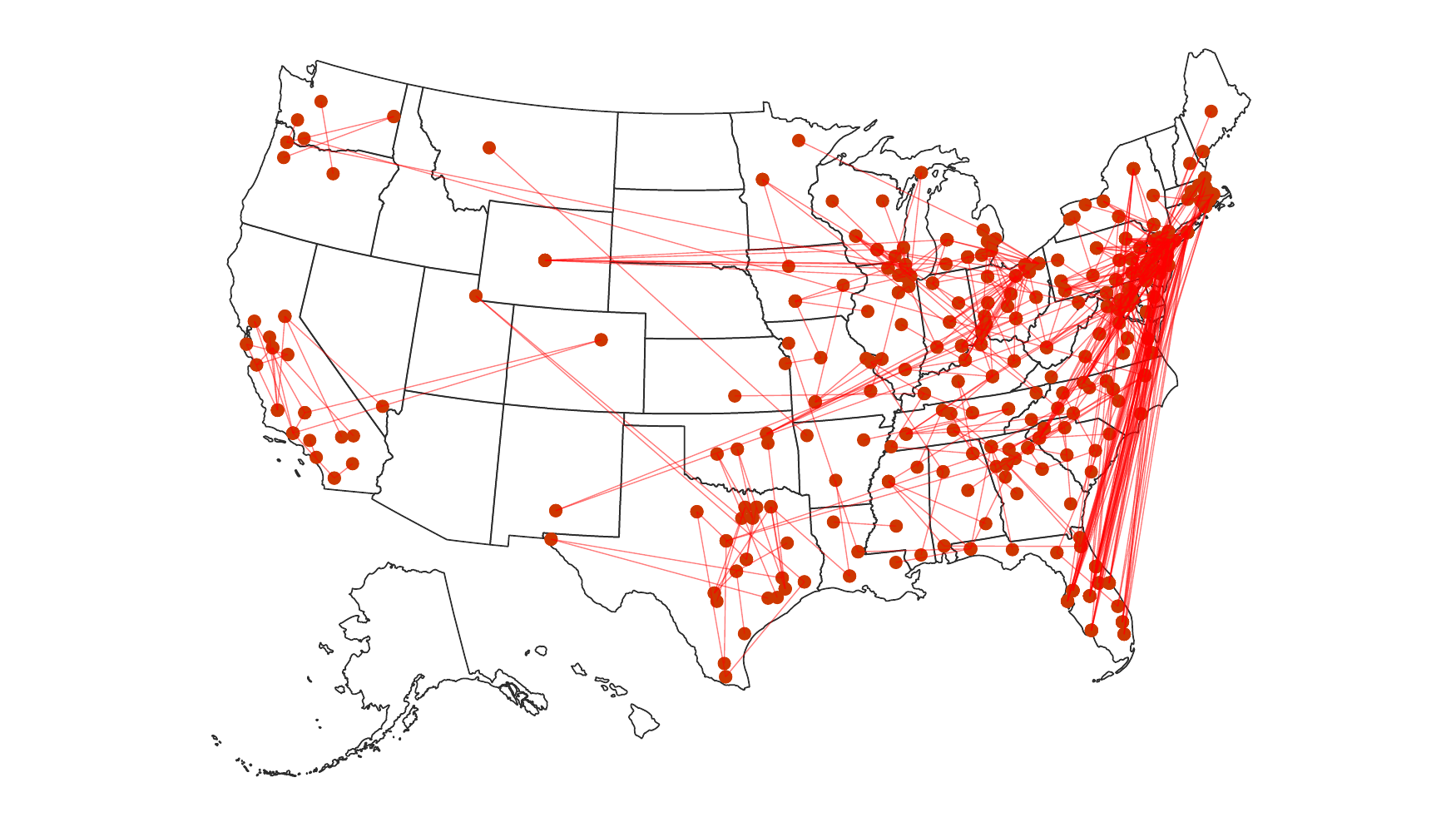}
    \caption{Edges identified by Network \GLASSO{} - $A_1$ (geographical network) but not by \GLASSO}
    \label{fig:unique_edges_covid_geodist}
\end{subfigure}
\vfill
\begin{subfigure}[b]{\linewidth}
    \includegraphics[trim= {0.0cm 0.5cm 0.0cm 0.5cm}, clip, width = 1\linewidth]{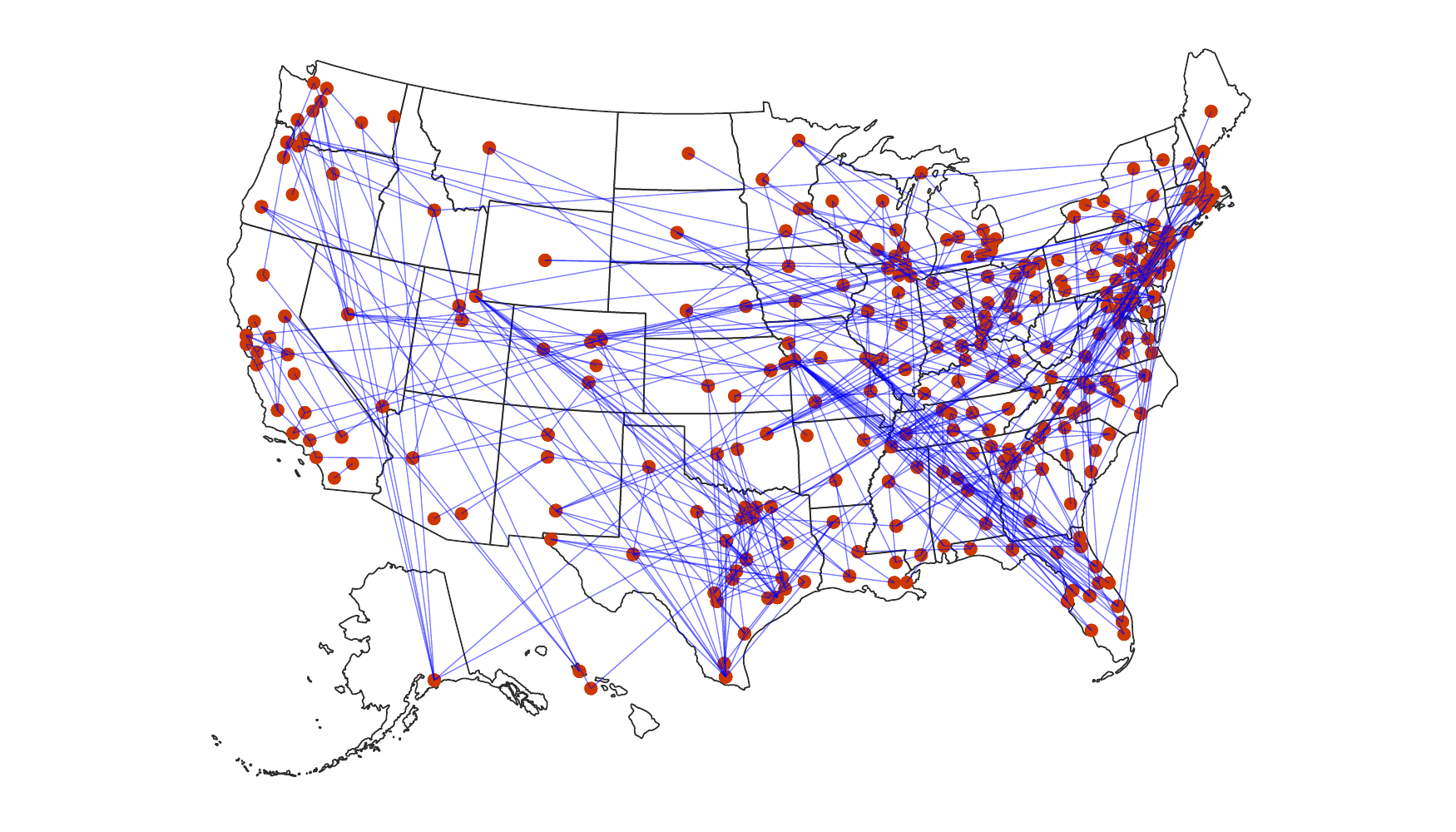}
    \caption{Edges identified by Network \GLASSO{} - $A_2$ (Facebook network) but not by \GLASSO}
    \label{fig:unique_edges_covid_facebook}
\end{subfigure}
\caption{Edges identified by Network \GLASSO but not by standard \GLASSO. 
%\jack{\BIC - $\gamma_{\EBIC} = 0$}. 
%The coordinates of Honolulu (Hawaii) have been adjusted from (-164.44361, 23.87280) to (-158.2019740, 21.4613654) for presentation
%and subsequent plots of the same type were treated in the same way.
}
%trim={<left> <lower> <right> <upper>}
\label{fig:unique_edges_covid}
\end{figure}

We next applied our spike-and-slab framework to obtain further insights on how the proportion of edge connections, as well as the mean partial correlation, depend on the two networks.
We initially ran 20,000 MCMC iterations, thinning to 1 in every 10, to sample from $\pi(\mbox{diag}(\Theta),\rho,\eta \mid y)$. The resulting chains for network hyperparameters $\eta$ had average effective sample size (\ESS) of 473.8 and average R-hat value of 1.004, providing us with sufficient confidence to use these chains to do inference on $\eta$ and produce empirical Bayes estimates. We then ran a second MCMC, fixing $\hat{\eta}$, for 4,000 iterations thinning to 1 in every 10. The resulting chains for partial correlations $\rho$ had an average \ESS of 372.8 and an average R-hat value of 1.003, suggesting that the chain converged.

As discussed earlier, the bottom panels in Figure \ref{fig:glasso_covid} display the fitted spike-and-slab distribution as a function of both the geographical closeness and Facebook networks. The corresponding plot for the flight network is in Figure B.6.
Table \ref{Tab:SS_COVID_CIs} presents the corresponding (empirical Bayes) hyper-parameter estimates, and Figure B.7 displays the estimated prior slab mean and prior slab probability as functions of the networks.
%\david{Below I adjusted the text assuming that we go with the parameterization where $\rho_{jk}$ is the partial correlation (not the negative pc) and that $\eta_1>0$ indicates larger $\rho_{jk}$. The sign of some of the entries in Table 3 may need to be adjusted, I guess those for $\eta_1$?}
%{\color{blue}{
Recall that positive entries in $\eta_0$ and $\eta_1$ indicate that the mean and variance (respectively)
%}}.
of the non-zero $\rho_{jk}$, i.e. the slab location and variance parameters, increase for counties that are strongly connected in the network.
%{\color{blue}{
Similarly, positive entries in $\eta_2$ indicates a higher probability of there being a non-zero partial correlation between such counties.
%}}
Table \ref{Tab:SS_COVID_CIs} hence shows that counties strongly connected in the Facebook and geographic networks had more non-zero partial correlations (relative to less connected counties), and that both the mean and variance of the partial correlations were also larger.
%was estimated to increase with both the geographic closeness and Facebook networks (positive $\hat{\eta}_{01}$, $\hat{\eta}_{02}$) but 
The flight passenger network was estimated to have no effect on there being a non-zero partial correlation, nor on their mean, and a mild effect on the variance of non-zero partial correlation (in agreement with the \BIC and cross-validation results in Table \ref{tab:results_covid}). 
%Their variance (slab dispersion parameter) was also estimated to increase with both geographic closeness and Facebook networks (negative $\hat{\eta}_{11}$, $\hat{\eta}_{12}$) and to decrease with the flight passenger network (negative $\hat{\eta}_{13}$). 
%The estimated positive coefficients for $(\hat{\eta}_{21},\hat{\eta}_{22})$, parameterising the slab probability (i.e. non-zero partial correlation) in terms of the geographic closeness and Facebook networks inidicate that bigger network values correspond to greater probabilities of being non-zero. The estimate $\hat{\eta}_{23}$ was negative indicating the opposite for the flight passenger network, however its credibility interval included 0.
%The estimated coefficients $(\hat{\eta}_{21},\hat{\eta}_{22})$ parameterising the slab probability (i.e. non-zero partial correlation) had opposite signs. Larger Facebook connectivity was associated with a larger probability of the corresponding $\rho_{jk}$ being non-zero, while under the Geographical distance the opposite was observed. 
%However the credibility interval for the latter ($\eta_{21}$) included 0, i.e. while there was strong evidence that the geographical distance affects the mean and variance of the partial correlations, this was not the case for the probability of a non-zero partial correlation (after accounting for the Facebook index).
%Indeed, Figure \ref{fig:ss_covid_probslab} illustrates that the probability of a non-zero partial correlation increases from near-zero to near-one as the Facebook index grows, whereas that is not the case for the geographical index.
The coefficients for the Facebook network are larger in absolute value than those of the geographical network indicating that the Facebook network has the stronger association with the dependence in \COVID rates. This is further illustrated in Figure B.7. % illustrates how the three networks affect the mean of the non-zero partial correlations and their probability of being non-zero, while both the geographical closeness and Facebook networks act to increase these, the affect of the Facebook network is greater than the geographical networks 
These results illustrate the greater flexibility provided by the network spike-and-slab models to portray the relation between the network data and the partial correlations. 
%We also note that the coefficients of the Facebook network are estimated as uniformly bigger than those of the Geographical network, suggesting the former is more informative. 
For completeness, Table B.1 summarises the selected graphical model under a 0.5 and 0.95 posterior probability threshold for declaring an edge.

Altogether, our results support that there is a fairly strong association between social media connections and the co-evolution of the pandemic, even when accounting for geographical closeness and a number of factors driving the mean structure, and that said association is not driven by airplane travel.

%\jack{From the first MCMC we can plot posterior for $\eta$'s, do any of them contain 0?? The bonus of Bayesian e.g. Does $\eta$'s for geographical contain 0?? Can't get from frequentist?, plot of posterior mean and 95\% interval }
 
\begin{table}[!ht]
\centering
\caption{Network spike-and-slab empirical Bayes (marginal MAP) estimates and 95\% posterior intervals for \COVID data. $A_1$,  $A_2$ and $A_3$: networks defined by
$1/log(Geodist)$, $log(Facebook)$ and $log(1+Flights)$. Bold values where the credibility interval includes 0.
%\jack{can we gie more infomative names, location, sclae, prob} 
}
\begin{tabular}{ccccc}
\toprule
& Intercept & $A_1$ & $A_2$ & $A_3$  \\
\midrule
$\eta_0$ (slab location) & -0.008 & 0.006 & 0.017 & \textbf{0.0} \\
95\% interval & (-0.009, -0.005) & (0.003,0.008) & (0.014,0.018) & (-0.002,0.002) \\
$\eta_1$ (slab dispersion) & 2.285 & 0.054 & 0.178 & -0.071 \\
95\% interval & (2.110, 2.507) & (0.003, 0.105) & (0.108, 0.240) & (-0.144, -0.002) \\
$\eta_2$ (slab probability) & -2.694 & 0.336 & 0.771 & \textbf{-0.1} \\
95\% interval & (-3.088,-2.397) & (0.154, 0.513) & (0.608, 0.949) & (-0.247, 0.047) \\
\bottomrule
\end{tabular}

% Original sign:
% \begin{tabular}{ccccc}
% \toprule
% & Intercept & $A_1$ & $A_2$ & $A_3$  \\
% \midrule
% $\eta_0$ (slab location) & 0.008 & -0.006 & -0.017 & 0.0 \\
% 95\% interval & (0.005, 0.009) & (-0.008, -0.003) & (-0.018, -0.014) & (-0.002,0.002) \\
% $\eta_1$ (slab dispersion) & -2.285 & -0.054 & -0.178 & 0.071 \\
% 95\% interval & (-2.507,-2.1100001) & (-0.105, -0.003) & (-0.240, -0.108) & (0.002, 0.14400001) \\
% $\eta_2$ (slab probability) & -2.694 & 0.336 & 0.771 & -0.1 \\
% 95\% interval & (-3.088, -2.397) & (0.15400001, 0.513) & (0.608, 0.94900006) & (-0.24700001, , 0.047) \\
% \bottomrule
% \end{tabular}

\label{Tab:SS_COVID_CIs}
\end{table}

\subsection{Stock market excess returns} \label{ssec:stock_market}

%\david{To give this section a bit more protagonism we could consider adding a figure. For example, the two bottom panels in Fig C.5 (potentially also Fig C.7 if we had the space). We wanna show how our model provides inference that helps interpret the graphical model, just showing a table doesn't fully do that. I made a number of cuts elsewhere so that hopefully we have space for this.}
%\jack{Network names Economic Risks Similairty Index and Policy Risks Similarity Index}

Recall that the outcomes are log-daily excess returns of $p=366$ US companies. The first network is an economic risks network $A_1$ where $a^{(1)}_{jk}$ is the Pearson's correlation between vectors of economic risks faced by firms $j$ and $k$.  The $r$th element of these vectors is $\log(1 + \textrm{prop}_r)$ where $\textrm{prop}_r$ is proportion of 10-K terms that reflect the $r$th economic risk according to the dictionaries of \citet{bakerPolicyNewsStock2019}.  $A_2$ is the equivalent but for vectors of policy risks.
%\david{The description of how the networks were computed doesn't match what we said in Section \ref{sec:applications}, there we didn't mention taking logs, and we mentioned proportion of term uses rather than counts.}
Pearson's correlation between the two networks was 0.301, suggesting that they provide largely different information. See Section C for a description of the data pre-processing.
 
We firstly run \GLASSO using no network data and then network \GLASSO using only the Economic network, only the Policy network, and finally using both networks. Table \ref{tab:results_stock} compares these four models. The model including both networks attained the best \BIC value and their estimated parameters $(\hat{\beta}_1,\hat{\beta}_2)$ are both negative. That is, partial correlations between companies that have large connections in the network are more likely to be non-zero, and are hence less regularised. 
The estimated graphical model when using both networks is sparser than under standard \GLASSO. %, and all the network \GLASSO models are preferable to standard \GLASSO.
 
%\jack{Insert some summary of the data, maybe comparing edges across the sector - want to say that incorporating risk information allows more general structure that sector specific}

%%Old results, only p = 200 
%\begin{table}[!ht]
%\centering
%\caption{Four models for the stock market data. $A_1$ is the Economic network, $A_2$ the Policy network. \BIC values account for the extra hyper-parameters in the network \GLASSO models. 
%10-fold is the 10-fold cross-validation log-likelihood
%\begin{tabular}{ccccccc}
%  \hline
%  Method & \BIC & $\hat{\beta}_0$ & $\hat{\beta}_1$ & $\hat{\beta}_2$ & Edges & 10-fold\\ 
%   \hline
%   \GLASSO & 47505.870 & -1.474 &  &  &  631 & -6857.00\\ 
%   Network \GLASSO - $A_1$ & 46927.615 & -1.289 & -0.368 &  &  664 & -6834.45 \\ 
%   Network \GLASSO - $A_2$ & 47012.765 & -1.105 &  & -0.474 &  604 & \textbf{-6828.51} \\ 
%   Network \GLASSO - $A_1$ \& $A_2$ & \textbf{46853.544} & -1.222 & -0.306 & -0.167 &  654 & -6831.10 \\ 
%   \hline
%\end{tabular}
%\label{tab:results_stock}
%\end{table}

\begin{table}[!ht]
\centering
\caption{Four models for the stock market data. $A_1$ is the Economic network, $A_2$ the Policy network. \BIC values account for the extra hyper-parameters in the network \GLASSO models. 
10-fold is the 10-fold cross-validation log-likelihood}
\begin{tabular}{ccccccc}
  \hline
  Method & \BIC & $\hat{\beta}_0$ & $\hat{\beta}_1$ & $\hat{\beta}_2$ & Edges & 10-fold\\ 
   \hline
 \GLASSO & 74078.33 & -1.639 &  &  & 2623 & -467.128 \\ 
 Network \GLASSO $A_1$ & 72459.55 & -1.137 & -0.677 &  & 2770 & \textbf{-463.683} \\ 
 Network \GLASSO $A_2$ & 73857.75 & -1.107 &  & --0.776 & 2211 & -467.128 \\ 
 Network \GLASSO $A_1$ \& $A_2$ & \textbf{72392.50} & -0.176 & -0.932 & -0.671 & 2058 & -467.42 \\ 
   \hline
\end{tabular}
\label{tab:results_stock}
\end{table}

To further assess the four models we evaluated their out-of-sample log-likelihood using 10-fold cross-validation. The models incorporating only the economic risks network performed best in this prediction exercise. 

\begin{figure}%[hbt!]
\centering
\includegraphics[width =0.49\linewidth]{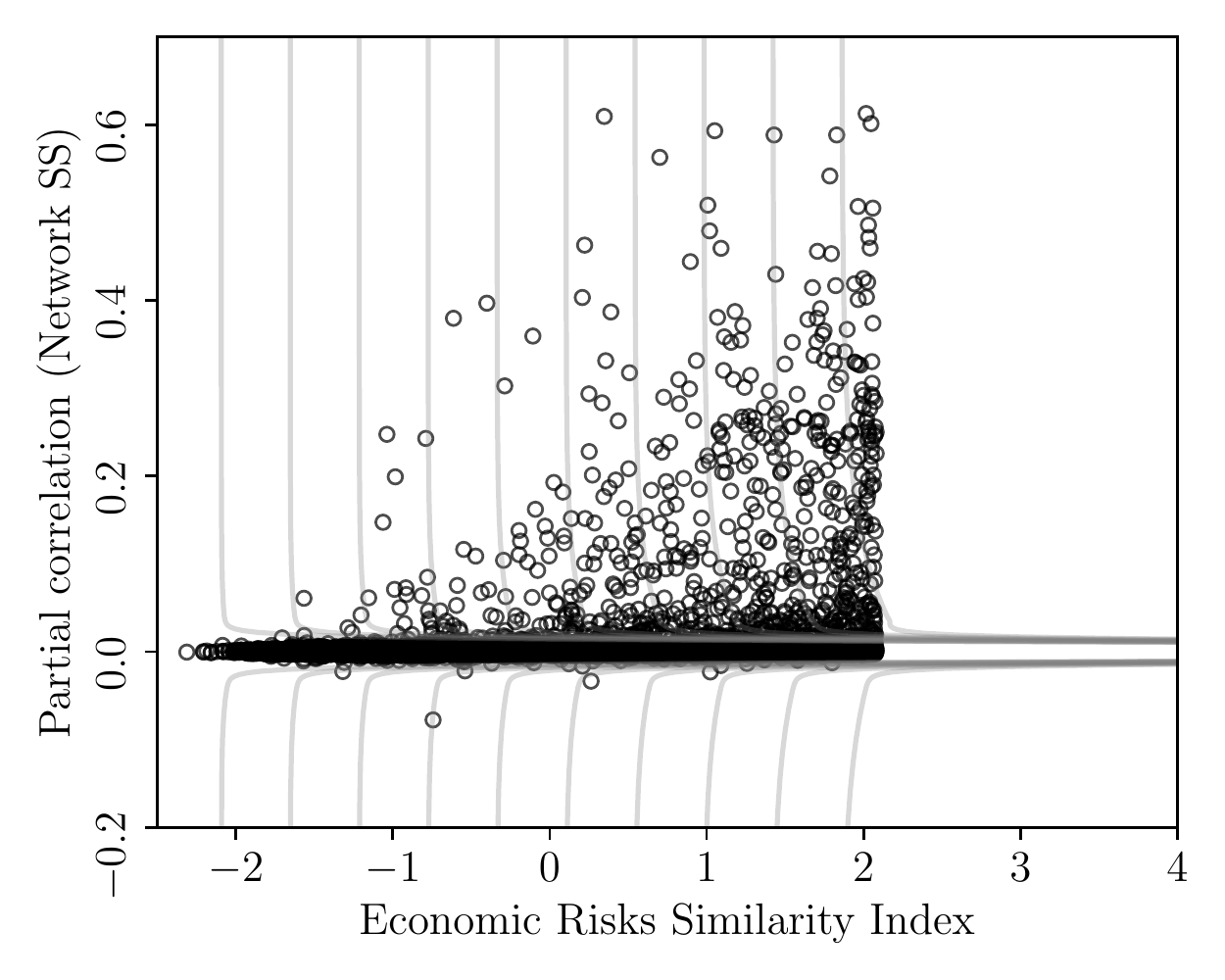}
\includegraphics[width =0.49\linewidth]{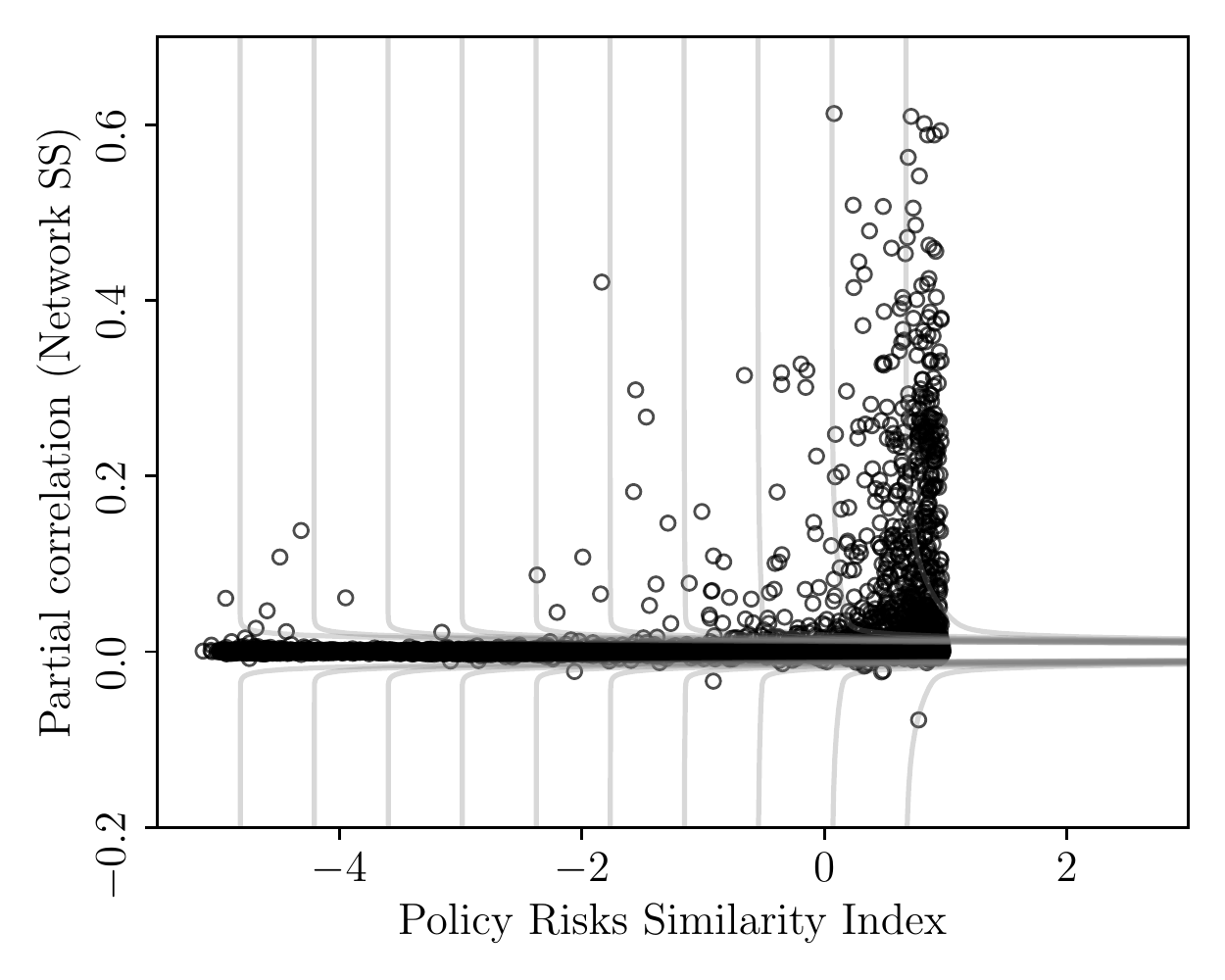}
\includegraphics[width =0.49\linewidth]{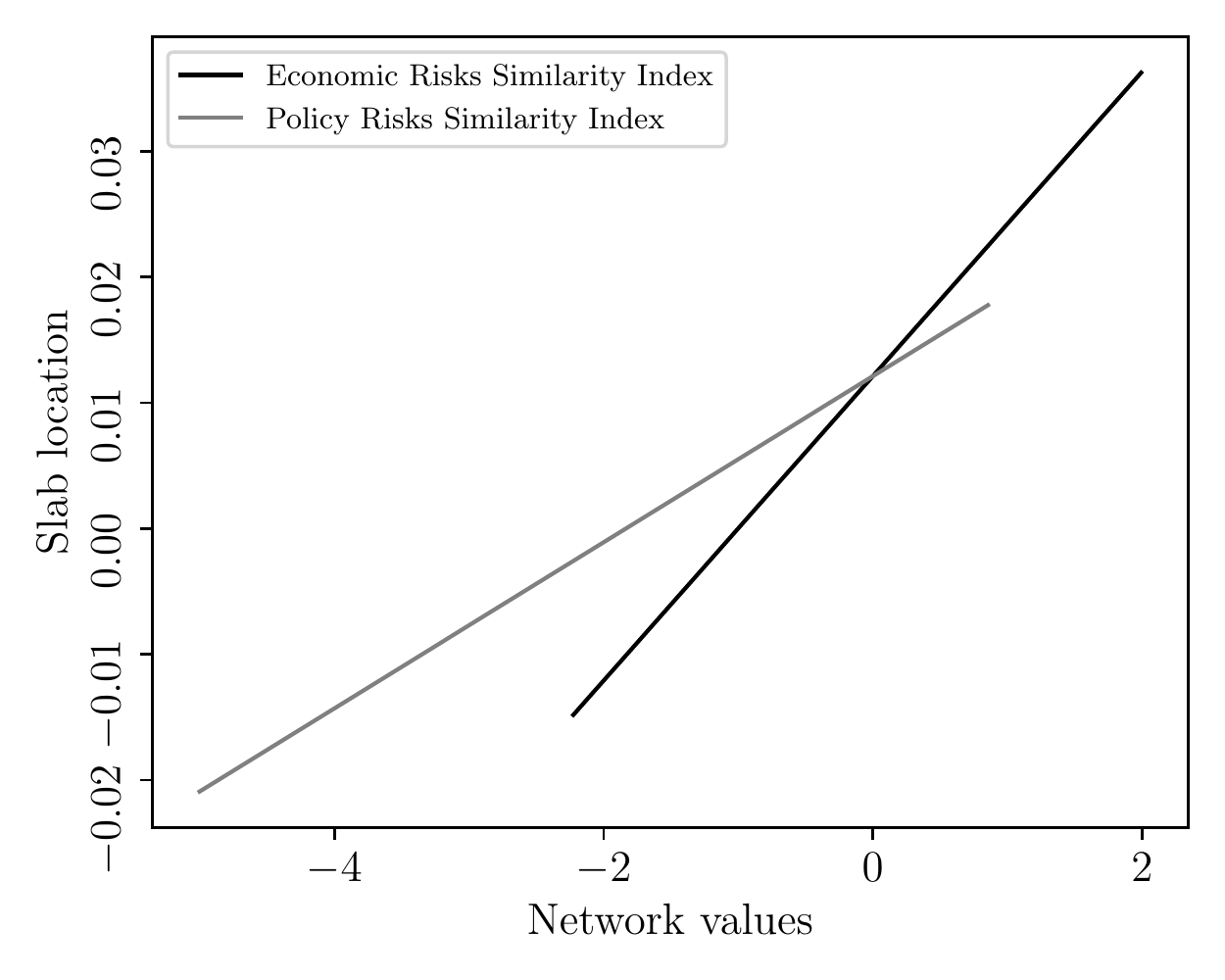}
\includegraphics[width =0.49\linewidth]{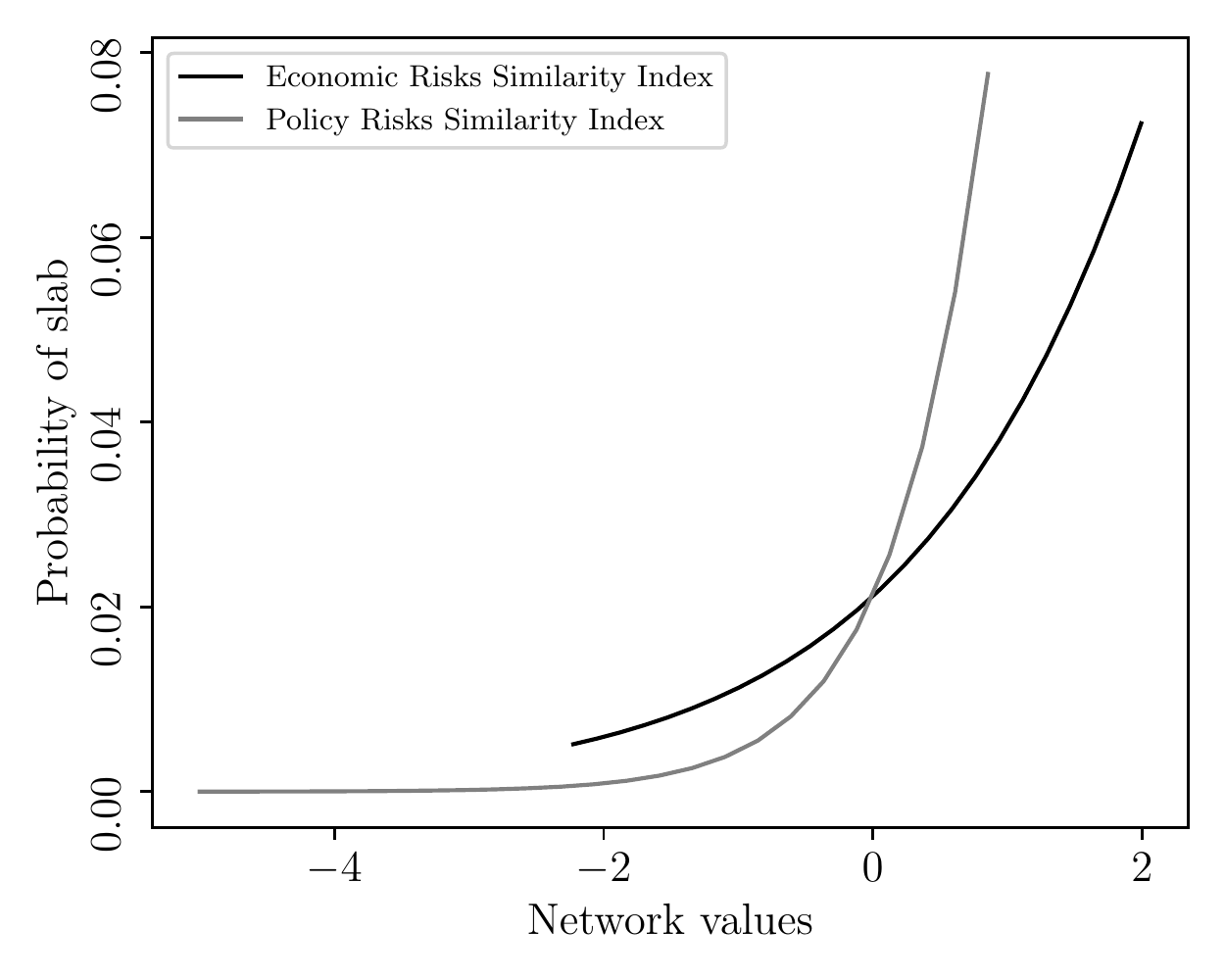}
\caption{Residual partial correlations of the stock market excess returns across firms vs Economy risk (left) and Policy risk (right).  
Top panel: fitted spike-and-slab distributions and fitted partial correlations estimated with network spike-and-slab model. %Bottom panel: prior slab probability as a function of both networks.
Bottom panel: Stock-market data: Slab location (\textbf{left}) and slab probability (\textbf{right}) as a function of both networks estimated by empirical Bayes. %We consider the slab here as the location of the partial correlations, the negation of the $\rho$ model parameters, which is why negative hyperparameter for $\eta_{0}$ (Table \ref{Tab:SS_Stock_CIs}) result in positive effects \jack{This will be corrected in the text anyway}.
}
\label{Fig:GLASSO_vs_Stock}
\end{figure}

To gain further insights into the relation between partial correlations and the network data, we applied our spike-and-slab framework. 
%\jack{Check Laura's input: We ran MCMC once for 10,000 iterations. For the network hyperparameters $\eta$, the resulting chains have minimal ESS and maximal absolute R-hat difference from 1 values of 275.184 and 0.005 respectively (mean ESS: 453.873, mean(abs(rhat-1)): 0.002) allowing us to do inference on these and to produce empirical Bayes estimates. We then ran MCMC again, fixing $\eta$ at their empirical Bayes estimates, for 4,000 iterations thinning to one in every 10 and achieve ESS and R-hat for the $\rho$’s of: min ESS 11.431, mean ESS 370.032, max(abs(rhat-1)): 0.220, mean(abs(rhat-1)): 0.003}
We initially run 10,000 MCMC iterations, thinning to 1 in every 10, to sample from $\pi(\mbox{diag}(\Theta),\rho,\eta \mid y)$. The resulting chains for network hyper-parameters $\eta$ had an average effective sample size (\ESS) of 453.9 and an average R-hat value of 1.002, suggesting MCMC convergence. %providing us with sufficient confidence to use these chains to do inference on $\eta$ and 
From these samples we obtained empirical Bayes estimate $\hat{\eta}$ and then ran a second MCMC, fixing $\hat{\eta}$, for 4,000 iterations thinning to 1 in every 10. The resulting chains for partial correlations $\rho$ had an average \ESS of 370.0 and an average R-hat of 1.003.

% \david{Fig 3 looks a bit strange. In the upper panels I would've expected the slab to have positive mean for larger network connections, based on where I see the estimated partial correlations, but instead the slabs seem to also be at zero? (perhaps the plotted slab covered by points?)}
Figure \ref{Fig:GLASSO_vs_Stock} shows the estimated spike-and-slab distributions for the partial correlations as a function of both networks, and Table \ref{Tab:SS_Stock_CIs} the corresponding hyper-parameter estimates. Recall that the network values were standardised so although they are correlations they do not lie in $[-1, 1]$. %Further, Figure \ref{fig:ss_stock_probslab} displays the estimated effect of both networks on the slab location (mean non-zero partial correlation) and its probability. 
Companies with strong connections in either network are estimated to have both a larger probability of a non-zero partial correlation (positive $\eta_2$) and a larger mean non-zero partial correlation (positive $\eta_0$). Interestingly, the policy network had the larger positive effect on the probability of a non-zero partial correlation, but its effect on their mean is smaller than the economic network's. %Unlike in the \COVID example, for large vales of the networks we still see many 0 partial correlations estimates and therefore the affect of the network 

% \david{Steve, the paragraph below probably needs to be updated.}

In short, both the economic and policy risk networks appear to contain independent information about the partial correlations among firms' stock returns.  The \GLASSO model suggests economic risks are more associated with such correlations: when both networks are included, the estimated impact of the economic risk network on the regularisation is stronger than that of the policy risk network and the out-of-sample predictive exercises prefers only the model with the economic network.  At the same time, the additional structure of the spike-and-slab model reveals a more subtle pattern.  The relationship between the strength of the connection in the economic risks networks and the partial correlation evolves more smoothly than for the strength of the connection in the policy network.  Only when firms are strongly connected in the latter is there an inferred impact on their partial correlations.

Since the pioneering work of \citet{hassanFirmLevelPoliticalRisk2019}, economists have used word-count-based approaches to measure and evaluate firm-level exposure to political and policy risks.  But exposure to policy risks is in part a function of economic risks: for example, firms exposed to air travel via their business model will also be exposed to regulation of the Federal Aviation Administration.  Our findings suggest that, after accounting for shared economic risks, shared policy risks only matter for excess return co-movement once firms are strongly connected.

For completeness, Table C.1 summarises the selected graphical model under a 0.5 and 0.95 posterior probability threshold for declaring an edge.
 
\begin{table}[!ht]
\centering
\caption{Network spike-and-slab empirical Bayes (marginal MAP) estimates and 95\% posterior credible intervals for the stock market data. $A_1$ is the Economic network, $A_2$ the Policy network.}
% FLIPPED locations
\begin{tabular}{cccc}
\toprule
 & intercept & $A_1$ & $A_2$ \\
\midrule
$\eta_0$ (slab location)  & 0.012 & 0.012 & 0.007 \\
95\% interval & (0.01, 0.014) & (0.009,0.014) & (0.003,0.009) \\
$\eta_1$ (slab dispersion) & 2.943 & 0.284 & -0.353 \\
95\% interval & (2.827,3.063) & (0.194, 0.369) & (-0.447,-0.249) \\
$\eta_2$ (slab probability) & -3.834 & 0.644 & 1.59 \\
95\% interval & (-4.122,-3.553) & (0.519, 0.808) & (1.267, 1.93) \\
\bottomrule
\end{tabular}

% % ORIGINAL locations
% \begin{tabular}{cccc}
% \toprule
%  & intercept & $A_1$ & $A_2$ \\
% \midrule
% $\eta_0$ (slab location)  & -0.012 & -0.012 & -0.007 \\
% 95\% interval & (-0.014, -0.01) & (-0.014,-0.009) & (-0.009,-0.003) \\
% $\eta_1$ (slab dispersion) & -2.943 & -0.284 & 0.353 \\
% 95\% interval & (-3.063,-2.827) & (-0.369, -0.194) & (0.249,0.447) \\
% $\eta_2$ (slab probability) & -3.834 & 0.644 & 1.59 \\
% 95\% interval & (-4.122, -3.553) & (0.519, 0.808) & (1.267, 1.93) \\
% \bottomrule
% \end{tabular}
\label{Tab:SS_Stock_CIs}
\end{table}
 
\section{Discussion} \label{sec:discussion}

We believe that our two frameworks to regress a graphical model on network data should have interest beyond our motivating \COVID and stock market applications.
%The penalised likelihood framework allows to regularise partial correlations between two variables based on the strength of their network connection(s), whereas 
Specifically, the Bayesian framework provides a rich model to depict the probability that parameters are non-zero as well as the distribution of non-zero parameters. Such a framework should find applicability in many other problems, for example high-dimensional regression or factor models.
%The former approach was operationalised using the \GOLAZO algorithm \citep{lauritzen:2020} and network hyper-parameters were learned using the \BIC, while the Bayesian model was implemented in the \texttt{NumPyro} probabilistic programming language \citep{bingham:2019,phan:2019} and empirical Bayes was used to estimate the network hyper-parameters. 
Our results showed that the external (network) data was particularly helpful in situations where the problem dimension was large relative to the sample size $n$, as is often the case in applications. Further, we observed that the ability to learn hyper-parameters ameliorated the consequences in a worst-case scenario where one introduces uninformative external data. 
 
Our \COVID application found that geographical closeness and a Facebook connectivity network were both informative about the dynamics of \COVID cases, allowing for the estimation of a sparser graphical model that predicted better out-of-sample. The Facebook network had a greater association with \COVID, suggesting for example to consider social media campaigns to help improve disease outcomes.
We stress that our findings should be understood as associations between social media and disease progression, rather than causal connections. For example, although \cite{nyhan2023like} found that 
Facebook feeds are skewed towards politically like-minded sources, there was little evidence that increasing exposure to more diverse sources reduced polarization. It is therefore possible that the Facebook index serves as a proxy for like-minded attitudes, rather than shared social media causally driving people to have similar attitudes.
In the stock market application we found similarity in firm risk exposures to economic and policy risks were both informative of the firms co-evolution in excess returns. While the policy network appeared to have a stronger relationship to whether two firms were connected, the economic network was more predictive of the behaviour of connected observations. These findings suggest that by understanding better the role played by risks declared by firms on their stock market behavior, it may be possible to design better portfolio strategies.
 
Further methodological work could consider richer relationships for how the graphical models can depend on the networks, e.g. by considering non-parametric models. Another interesting avenue would be developing computational methods to scale our algorithms to even higher dimensions. %, i.e. the setting where adding network information is more useful.

\section*{Acknowledgements}

JJ, DR and PZ were partially funded by Government of Spain's Plan Nacional PGC2018-101643-B-I00 and the Ayudas Fundación BBVA %a Equipos de Investigación Cientifica 2017 and 
Proyectos de Investigación Científica en Matemáticas 2021,
grants. 
JJ was also partially funded by Juan
de la Cierva Formación de la Agencia Estatal de Investigación  FJC2020-046348-I,
and DR by Europa Excelencia EUR2020-112096, the AEI/10.13039/501100011033 and European Union “NextGenerationEU”/PRT,
and grant Consolidaci\'on investigadora CNS2022-135963 by the AEI. 
DR and PZ also acknowledge funding from Huawei Technologies Ireland.
LL acknowledges the financial support from the China Scholarship Council (CSC) (Grant No.202006240148).
SH gratefully acknowledges funding from European Research Council Consolidator Grant 864863, which supported his time and the work of LB.
DR and PZ also acknowledge funding from  the Ayudas Fundación BBVA a Equipos de Investigación Cientifica 2021 and from Huawei Technologies Ireland.

We thank Dennis Kristensen for helpful comments and Du Phan for advising us on our models' implementation in \texttt{NumPyro}.

\bibliographystyle{plainnat}
\bibliography{references}

\begin{thebibliography}{62}
\providecommand{\natexlab}[1]{#1}
\providecommand{\url}[1]{\texttt{#1}}
\expandafter\ifx\csname urlstyle\endcsname\relax
  \providecommand{\doi}[1]{doi: #1}\else
  \providecommand{\doi}{doi: \begingroup \urlstyle{rm}\Url}\fi

\bibitem[Allcott et~al.(2020)Allcott, Boxell, Conway, Gentzkow, Thaler, and
  Yang]{allcott2020polarization}
Hunt Allcott, Levi Boxell, Jacob Conway, Matthew Gentzkow, Michael Thaler, and
  David Yang.
\newblock Polarization and public health: Partisan differences in social
  distancing during the coronavirus pandemic.
\newblock \emph{Journal of public economics}, 191:\penalty0 104254, 2020.

\bibitem[Baker et~al.(2019)Baker, Bloom, Davis, and
  Kost]{bakerPolicyNewsStock2019}
Scott~R. Baker, Nicholas Bloom, Steven~J. Davis, and Kyle~J. Kost.
\newblock Policy {{News}} and {{Stock Market Volatility}}.
\newblock \emph{National Bureau of Economic Research Working Paper Series},
  \penalty0 (w25720), 2019.

\bibitem[Banerjee et~al.(2008)Banerjee, El~Ghaoui, and
  d'Aspremont]{banerjee2008model}
Onureena Banerjee, Laurent El~Ghaoui, and Alexandre d'Aspremont.
\newblock Model selection through sparse maximum likelihood estimation for
  multivariate gaussian or binary data.
\newblock \emph{The Journal of Machine Learning Research}, 9:\penalty0
  485--516, 2008.

\bibitem[Barbieri and Berger(2004)]{barbieri:2004}
M.M. Barbieri and J.O. Berger.
\newblock Optimal predictive model selection.
\newblock \emph{The Annals of Statistics}, 32\penalty0 (3):\penalty0 870--897,
  2004.

\bibitem[Benjamini and Hochberg(1995)]{benjamini:1995}
Y.~Benjamini and Y.~Hochberg.
\newblock Controlling the false discovery rate: A practical and powerful
  approach to multiple testing.
\newblock \emph{Journal of the Royal Statistical Society B}, 57\penalty0
  (1):\penalty0 289--300, 1995.

\bibitem[Bingham et~al.(2019)Bingham, Chen, Jankowiak, Obermeyer, Pradhan,
  Karaletsos, Singh, Szerlip, Horsfall, and Goodman]{bingham:2019}
Eli Bingham, Jonathan~P. Chen, Martin Jankowiak, Fritz Obermeyer, Neeraj
  Pradhan, Theofanis Karaletsos, Rohit Singh, Paul~A. Szerlip, Paul Horsfall,
  and Noah~D. Goodman.
\newblock Pyro: Deep universal probabilistic programming.
\newblock \emph{Journal of Machine Learning Research}, 20:\penalty0 28:1--28:6,
  2019.

\bibitem[Bu and Lederer(2021)]{bu2021integrating}
Yunqi Bu and Johannes Lederer.
\newblock Integrating additional knowledge into the estimation of graphical
  models.
\newblock \emph{The international journal of biostatistics}, 18\penalty0
  (1):\penalty0 1--17, 2021.

\bibitem[Bureau(2020)]{Population:2020}
U.S.~Census Bureau.
\newblock Us 2019 population data, 2020.
\newblock Available from github:
  \url{https://www2.census.gov/programs-surveys/popest/tables/2010-2019/counties/totals/
  }.

\bibitem[Carpenter et~al.(2017)Carpenter, Gelman, Hoffman, Lee, Goodrich,
  Betancourt, Brubaker, Guo, Li, and Riddell]{carpenter:2017}
Bob Carpenter, Andrew Gelman, Matthew~D Hoffman, Daniel Lee, Ben Goodrich,
  Michael Betancourt, Marcus Brubaker, Jiqiang Guo, Peter Li, and Allen
  Riddell.
\newblock Stan: A probabilistic programming language.
\newblock \emph{Journal of statistical software}, 76\penalty0 (1):\penalty0
  1--32, 2017.

\bibitem[Carter et~al.(2021)Carter, Rossell, and Smith]{carter2021partial}
Jack~Storror Carter, David Rossell, and Jim~Q Smith.
\newblock Partial correlation graphical lasso.
\newblock \emph{arXiv preprint arXiv:2104.10099}, 2021.

\bibitem[Cassese et~al.(2014)Cassese, Guindani, and Vannucci]{cassese:2014}
Alberto Cassese, Michele Guindani, and Marina Vannucci.
\newblock A {B}ayesian integrative model for genetical genomics with spatially
  informed variable selection.
\newblock \emph{Cancer informatics}, 13:\penalty0 S13784, 2014.

\bibitem[Chen and Chen(2008)]{chen:2008}
J.~Chen and Z.~Chen.
\newblock Extended {B}ayesian information criteria for model selection with
  large model spaces.
\newblock \emph{Biometrika}, 95\penalty0 (3):\penalty0 759--771, 2008.

\bibitem[Chen et~al.(2021)Chen, Chatterjee, Landi, and Shi]{chen_tinghuei:2021}
Ting-Huei Chen, Nilanjan Chatterjee, Maria~Teresa Landi, and Jianxin Shi.
\newblock A penalized regression framework for building polygenic risk models
  based on summary statistics from genome-wide association studies and
  incorporating external information.
\newblock \emph{Journal of the American Statistical Association}, 116\penalty0
  (533):\penalty0 133--143, 2021.

\bibitem[Chiang et~al.(2017)Chiang, Guindani, Yeh, Haneef, Stern, and
  Vannucci]{chiang:2017}
Sharon Chiang, Michele Guindani, Hsiang~J Yeh, Zulfi Haneef, John~M Stern, and
  Marina Vannucci.
\newblock Bayesian vector autoregressive model for multi-subject effective
  connectivity inference using multi-modal neuroimaging data.
\newblock \emph{Human brain mapping}, 38\penalty0 (3):\penalty0 1311--1332,
  2017.

\bibitem[CSSE(2020{\natexlab{a}})]{COVID:2020}
CSSE.
\newblock Covid19 infection rates, 2020{\natexlab{a}}.
\newblock Available from github:
  \url{https://github.com/CSSEGISandData/COVID-19/blob/master/csse_covid_19_data/csse_covid_19_time_series/time_series_covid19_confirmed_US.csv}.

\bibitem[CSSE(2020{\natexlab{b}})]{Policy:2020}
CSSE.
\newblock Government coronavirus policies, 2020{\natexlab{b}}.
\newblock Available from github:
  \url{https://github.com/CSSEGISandData/COVID-19_Unified-Dataset }.

\bibitem[CSSE(2020{\natexlab{c}})]{Temperature:2020}
CSSE.
\newblock Daily average near-surface temperature, 2020{\natexlab{c}}.
\newblock Available from github:
  \url{https://github.com/CSSEGISandData/COVID-19_Unified-Dataset/tree/master/Hydromet}.

\bibitem[CSSE(2020{\natexlab{d}})]{Vaccination:2020}
CSSE.
\newblock U.s. vaccination data, 2020{\natexlab{d}}.
\newblock Available from github:
  \url{https://github.com/govex/COVID-19/tree/master/data_tables/vaccine_data/us_data/time_series}.

\bibitem[Davis et~al.(2020)Davis, Hansen, and
  {Seminario-Amez}]{davisFirmLevelRiskExposures2020}
Steven~J. Davis, Stephen Hansen, and Cristhian {Seminario-Amez}.
\newblock Firm-{{Level Risk Exposures}} and {{Stock Returns}} in the {{Wake}}
  of {{COVID-19}}.
\newblock Working {{Paper}} 27867, {National Bureau of Economic Research},
  2020.

\bibitem[Duane et~al.(1987)Duane, Kennedy, Pendleton, and
  Roweth]{duane1987hybrid}
Simon Duane, Anthony~D Kennedy, Brian~J Pendleton, and Duncan Roweth.
\newblock Hybrid {M}onte {C}arlo.
\newblock \emph{Physics Letters B}, 195\penalty0 (2):\penalty0 216--222, 1987.

\bibitem[Elton and Gruber(1973)]{eltonEstimatingDependenceStructure1973}
Edwin~J. Elton and Martin~J. Gruber.
\newblock Estimating the {{Dependence Structure}} of {{Share
  Prices--Implications}} for {{Portfolio Selection}}.
\newblock \emph{The Journal of Finance}, 28\penalty0 (5):\penalty0 1203--1232,
  1973.
\newblock ISSN 0022-1082.
\newblock \doi{10.2307/2978758}.

\bibitem[Fama and French(1993)]{FAMA19933}
Eugene~F. Fama and Kenneth~R. French.
\newblock Common risk factors in the returns on stocks and bonds.
\newblock \emph{Journal of Financial Economics}, 33\penalty0 (1):\penalty0
  3--56, 1993.
\newblock ISSN 0304-405X.
\newblock \doi{https://doi.org/10.1016/0304-405X(93)90023-5}.
\newblock URL
  \url{https://www.sciencedirect.com/science/article/pii/0304405X93900235}.

\bibitem[Fan et~al.(2009)Fan, Feng, and Wu]{fan:2009}
Jianqing Fan, Yang Feng, and Yichao Wu.
\newblock {Network exploration via the adaptive LASSO and SCAD penalties}.
\newblock \emph{Annals of Applied Statistics}, 3\penalty0 (2):\penalty0
  521--541, 2009.

\bibitem[Fan and Tang(2013)]{fan_yingying:2013}
Yingying Fan and Cheng~Yong Tang.
\newblock Tuning parameter selection in high dimensional penalized likelihood.
\newblock \emph{Journal of the Royal Statistical Society B}, 75\penalty0
  (3):\penalty0 531--552, 2013.

\bibitem[Foygel and Drton(2010)]{foygel2010extended}
Rina Foygel and Mathias Drton.
\newblock Extended bayesian information criteria for gaussian graphical models.
\newblock \emph{Advances in Neural Information Processing Systems},
  23:\penalty0 604--612, 2010.

\bibitem[Friedman et~al.(2008)Friedman, Hastie, and Tibshirani]{friedman:2008}
J.~Friedman, T.~Hastie, and R.~Tibshirani.
\newblock Sparse inverse covariance estimation with the graphical lasso.
\newblock \emph{Biostatistics}, 9\penalty0 (3):\penalty0 432--441, 2008.

\bibitem[Gan et~al.(2018)Gan, Narisetty, and Liang]{gan:2018}
L.~Gan, N.N. Narisetty, and F.~Liang.
\newblock Bayesian regularization for graphical models with unequal shrinkage.
\newblock \emph{Journal of the American Statistical Association},
  just-accepted:\penalty0 1--14, 2018.

\bibitem[Giannone et~al.(2021)Giannone, Lenza, and
  Primiceri]{giannoneEconomicPredictionsBig2021}
Domenico Giannone, Michele Lenza, and Giorgio~E. Primiceri.
\newblock Economic {{Predictions With Big Data}}: {{The Illusion}} of
  {{Sparsity}}.
\newblock \emph{Econometrica}, 89\penalty0 (5):\penalty0 2409--2437, 2021.
\newblock ISSN 0012-9682.
\newblock \doi{10.3982/ECTA17842}.

\bibitem[Goto and Xu(2015)]{gotoImprovingMeanVariance2015}
Shingo Goto and Yan Xu.
\newblock Improving {{Mean Variance Optimization}} through {{Sparse Hedging
  Restrictions}}.
\newblock \emph{The Journal of Financial and Quantitative Analysis},
  50\penalty0 (6):\penalty0 1415--1441, 2015.
\newblock ISSN 0022-1090.

\bibitem[Hanley and Hoberg(2019)]{hanleyDynamicInterpretationEmerging2019}
Kathleen~Weiss Hanley and Gerard Hoberg.
\newblock Dynamic {{Interpretation}} of {{Emerging Risks}} in the {{Financial
  Sector}}.
\newblock \emph{The Review of Financial Studies}, 32\penalty0 (12):\penalty0
  4543--4603, 2019.
\newblock ISSN 0893-9454.
\newblock \doi{10.1093/rfs/hhz023}.

\bibitem[Hassan et~al.(2019)Hassan, Hollander, {van Lent}, and
  Tahoun]{hassanFirmLevelPoliticalRisk2019}
Tarek~A Hassan, Stephan Hollander, Laurence {van Lent}, and Ahmed Tahoun.
\newblock Firm-{{Level Political Risk}}: {{Measurement}} and {{Effects}}.
\newblock \emph{The Quarterly Journal of Economics}, 134\penalty0 (4):\penalty0
  2135--2202, 2019.
\newblock ISSN 0033-5533.
\newblock \doi{10.1093/qje/qjz021}.

\bibitem[Higgins et~al.(2018)Higgins, Kundu, and Guo]{higgins2018integrative}
Ixavier~A Higgins, Suprateek Kundu, and Ying Guo.
\newblock Integrative bayesian analysis of brain functional networks
  incorporating anatomical knowledge.
\newblock \emph{NeuroImage}, 181:\penalty0 263--278, 2018.

\bibitem[Hoffman and Gelman(2014)]{hoffman2014no}
Matthew~D Hoffman and Andrew Gelman.
\newblock The {N}o-{U}-{T}urn sampler: adaptively setting path lengths in
  {H}amiltonian {M}onte {C}arlo.
\newblock \emph{Journal of Machine Learning Research}, 15\penalty0
  (1):\penalty0 1593--1623, 2014.

\bibitem[Kuchler et~al.(2021)Kuchler, Russel, and Stroebel]{kuchler:2021}
Theresa Kuchler, Dominic Russel, and Johannes Stroebel.
\newblock The geographic spread of covid-19 correlates with the structure of
  social networks as measured by facebook.
\newblock \emph{Journal of Urban Economics}, page 103314, 2021.

\bibitem[Kuismin and Sillanp{\"a}{\"a}(2021)]{kuismin}
Markku Kuismin and Mikko~J Sillanp{\"a}{\"a}.
\newblock Mcpese: Monte carlo penalty selection for graphical lasso.
\newblock \emph{Bioinformatics}, 37\penalty0 (5):\penalty0 726--727, 2021.

\bibitem[Lauritzen and Zwiernik(2020)]{lauritzen:2020}
Steffen Lauritzen and Piotr Zwiernik.
\newblock Locally associated graphical models and mixed convex exponential
  families.
\newblock \emph{arXiv}, 2008.04688:\penalty0 1--34, 2020.
\newblock to appear in Annals of Statistics.

\bibitem[Markowitz(1952)]{markowitzPortfolioSelection1952}
Harry Markowitz.
\newblock Portfolio {{Selection}}.
\newblock \emph{The Journal of Finance}, 7\penalty0 (1):\penalty0 77--91, 1952.
\newblock ISSN 0022-1082.
\newblock \doi{10.2307/2975974}.

\bibitem[M{\"{u}}ller et~al.(2004)M{\"{u}}ller, Parmigiani, Robert, and
  Rousseau]{mueller:2004}
P.~M{\"{u}}ller, G.~Parmigiani, C.~Robert, and J.~Rousseau.
\newblock Optimal sample size for multiple testing: the case of gene expression
  microarrays.
\newblock \emph{Journal of the American Statistical Association}, 99\penalty0
  (468):\penalty0 990--1001, 2004.

\bibitem[Neal(2011)]{neal:2011}
Radford Neal.
\newblock \emph{{MCMC} using {H}amiltonian dynamics}, pages 113--162.
\newblock Chapman and Hall/CRC, 2011.

\bibitem[Ng et~al.(2012)Ng, Varoquaux, Poline, and Thirion]{ng2012novel}
Bernard Ng, Ga{\"e}l Varoquaux, Jean-Baptiste Poline, and Bertrand Thirion.
\newblock A novel sparse graphical approach for multimodal brain connectivity
  inference.
\newblock In \emph{International Conference on Medical Image Computing and
  Computer-Assisted Intervention}, pages 707--714. Springer, 2012.

\bibitem[Nyhan et~al.(2023)Nyhan, Settle, Thorson, Wojcieszak, Barber{\'a},
  Chen, Allcott, Brown, Crespo-Tenorio, Dimmery, et~al.]{nyhan2023like}
Brendan Nyhan, Jaime Settle, Emily Thorson, Magdalena Wojcieszak, Pablo
  Barber{\'a}, Annie~Y Chen, Hunt Allcott, Taylor Brown, Adriana
  Crespo-Tenorio, Drew Dimmery, et~al.
\newblock Like-minded sources on facebook are prevalent but not polarizing.
\newblock \emph{Nature}, 620\penalty0 (7972):\penalty0 137--144, 2023.

\bibitem[Peterson et~al.(2016)Peterson, Stingo, and
  Vannucci]{peterson2016joint}
Christine~B Peterson, Francesco~C Stingo, and Marina Vannucci.
\newblock Joint bayesian variable and graph selection for regression models
  with network-structured predictors.
\newblock \emph{Statistics in medicine}, 35\penalty0 (7):\penalty0 1017--1031,
  2016.

\bibitem[Phan et~al.(2019)Phan, Pradhan, and Jankowiak]{phan:2019}
Du~Phan, Neeraj Pradhan, and Martin Jankowiak.
\newblock Composable effects for flexible and accelerated probabilistic
  programming in numpyro.
\newblock \emph{arXiv}, 1912.11554:\penalty0 1--10, 2019.

\bibitem[Pineda-Pardo et~al.(2014)Pineda-Pardo, Bru{\~n}a, Woolrich, Marcos,
  Nobre, Maest{\'u}, and Vidaurre]{pineda2014guiding}
Jos{\'e}~Angel Pineda-Pardo, Ricardo Bru{\~n}a, Mark Woolrich, Alberto Marcos,
  Anna~C Nobre, Fernando Maest{\'u}, and Diego Vidaurre.
\newblock Guiding functional connectivity estimation by structural connectivity
  in meg: an application to discrimination of conditions of mild cognitive
  impairment.
\newblock \emph{Neuroimage}, 101:\penalty0 765--777, 2014.

\bibitem[Quintana and Conti(2013)]{quintana_ma:2013}
MA~Quintana and DV~Conti.
\newblock Integrative variable selection via {B}ayesian model uncertainty.
\newblock \emph{Statistics in medicine}, 32\penalty0 (28):\penalty0 4938--4953,
  2013.

\bibitem[Rockova and George(2014)]{rockova:2014}
V.~Rockova and E.I. George.
\newblock {EMVS}: The {EM} approach to {B}ayesian variable selection.
\newblock \emph{Journal of the American Statistical Association}, 109\penalty0
  (506):\penalty0 828--846, 2014.

\bibitem[Rossell and Zwiernik(2021)]{rossell2021dependence}
David Rossell and Piotr Zwiernik.
\newblock Dependence in elliptical partial correlation graphs.
\newblock \emph{Electronic Journal of Statistics}, 15\penalty0 (2):\penalty0
  4236--4263, 2021.

\bibitem[Schwarz(1978)]{schwarz:1978}
G.~Schwarz.
\newblock Estimating the dimension of a model.
\newblock \emph{Annals of Statistics}, 6:\penalty0 461--464, 1978.

\bibitem[Scott and Berger(2010)]{scott:2010}
J.G. Scott and J.O Berger.
\newblock Bayes and empirical {B}ayes multiplicity adjustment in the variable
  selection problem.
\newblock \emph{The Annals of Statistics}, 38\penalty0 (5):\penalty0
  2587--2619, 2010.

\bibitem[Senneret et~al.(2016)Senneret, Malevergne, Abry, Perrin, and
  Jaffr{\`e}s]{senneretCovariancePrecisionMatrix2016}
Marc Senneret, Yannick Malevergne, Patrice Abry, Gerald Perrin, and Laurent
  Jaffr{\`e}s.
\newblock Covariance {{Versus Precision Matrix Estimation}} for {{Efficient
  Asset Allocation}}.
\newblock \emph{IEEE Journal of Selected Topics in Signal Processing},
  10\penalty0 (6):\penalty0 982--993, 2016.
\newblock ISSN 1941-0484.
\newblock \doi{10.1109/JSTSP.2016.2577546}.

\bibitem[Sinnott(1984)]{sinnott1984virtues}
Roger~W Sinnott.
\newblock Virtues of the haversine.
\newblock \emph{Sky and telescope}, 68\penalty0 (2):\penalty0 158, 1984.

\bibitem[Stingo et~al.(2010)Stingo, Chen, Vannucci, Barrier, and
  Mirkes]{stingo:2010}
Francesco~C Stingo, Yian~A Chen, Marina Vannucci, Marianne Barrier, and
  Philip~E Mirkes.
\newblock A {B}ayesian graphical modeling approach to micro{RNA} regulatory
  network inference.
\newblock \emph{The annals of applied statistics}, 4\penalty0 (4):\penalty0
  2024--2048, 2010.

\bibitem[Stingo et~al.(2011)Stingo, Chen, Tadesse, and Vannucci]{stingo:2011}
Francesco~C Stingo, Yian~A Chen, Mahlet~G Tadesse, and Marina Vannucci.
\newblock Incorporating biological information into linear models: A {B}ayesian
  approach to the selection of pathways and genes.
\newblock \emph{The annals of applied statistics}, 5\penalty0 (3):\penalty0
  1--24, 2011.

\bibitem[Wang(2012)]{wang_hao:2012}
Hao Wang.
\newblock Bayesian graphical {LASSO} models and efficient posterior
  computation.
\newblock \emph{Bayesian Analysis}, 7\penalty0 (4):\penalty0 867--886, 2012.

\bibitem[Wang(2015)]{wang_hao:2015}
Hao Wang.
\newblock Scaling it up: Stochastic search structure learning in graphical
  models.
\newblock \emph{Bayesian Analysis}, 10\penalty0 (2):\penalty0 351--377, 2015.

\bibitem[Wang et~al.(2016)Wang, Ren, and Gu]{wang_lingxiao:2016}
Lingxiao Wang, Xiang Ren, and Quanquan Gu.
\newblock Precision matrix estimation in high dimensional {G}aussian graphical
  models with faster rates.
\newblock \emph{Proceedings of the 19th International Conference on Artificial
  Intelligence and Statistics, AISTATS 2016}, 51:\penalty0 177, 2016.

\bibitem[Wang and Zhu(2011)]{wang_tao:2011}
Tao Wang and Lixing Zhu.
\newblock Consistent tuning parameter selection in high dimensional sparse
  linear regression.
\newblock \emph{Journal of Multivariate Analysis}, 102\penalty0 (7):\penalty0
  1141--1151, 2011.

\bibitem[Yan(2016)]{yan2016rbayesianoptimization}
Yachen Yan.
\newblock rbayesianoptimization: Bayesian optimization of hyperparameters.
\newblock \emph{R package version}, 1\penalty0 (0), 2016.

\bibitem[Yuan and Lin(2007)]{yuan2007model}
Ming Yuan and Yi~Lin.
\newblock Model selection and estimation in the gaussian graphical model.
\newblock \emph{Biometrika}, 94\penalty0 (1):\penalty0 19--35, 2007.

\bibitem[Zanella and Roberts(2021)]{zanella2021multilevel}
Giacomo Zanella and Gareth Roberts.
\newblock Multilevel linear models, gibbs samplers and multigrid decompositions
  (with discussion).
\newblock \emph{Bayesian Analysis}, 16\penalty0 (4):\penalty0 1309--1391, 2021.

\bibitem[Zhang et~al.(2010)Zhang, Li, and Tsai]{zhang_yiyun:2010}
Yiyun Zhang, Runze Li, and Chih-Ling Tsai.
\newblock Regularization parameter selections via generalized information
  criterion.
\newblock \emph{Journal of the American Statistical Association}, 105\penalty0
  (489):\penalty0 312--323, 2010.

\bibitem[Zhao et~al.(2012)Zhao, Liu, Roeder, Lafferty, and
  Wasserman]{zhao2012huge}
Tuo Zhao, Han Liu, Kathryn Roeder, John Lafferty, and Larry Wasserman.
\newblock The huge package for high-dimensional undirected graph estimation in
  r.
\newblock \emph{The Journal of Machine Learning Research}, 13\penalty0
  (1):\penalty0 1059--1062, 2012.

\end{thebibliography}
 
\newpage 
\appendix
%\addcontentsline{toc}{section}{Supplementary Material} % Add the appendix text to the document TOC
%\part{Supplementary Material} % Start the appendix part
%\parttoc % Insert the appendix TOC
%If i want to change their names then look at this link
%https://tex.stackexchange.com/questions/67253/changing-heading-of-appendix
%https://tex.stackexchange.com/questions/215886/numbering-of-lemmas-in-appendix

\section*{Supplementary Material} 

Section \ref{Sec:Implmentation_details} provides further details for the implementation of our network \GLASSO and network spike-and-slab models. Section \ref{sec:covid_dataprocessing} contains further information related to our \COVID data application, including the data collection, preprocessing, linear model estimation and diagnostic checks, network specification and linearity check, as well as further results and figures. Section \ref{sec:stock_dataprocessing} provides analogous information for the stock market data. Lastly, Section \ref{sec:stan_vs_numpyro} provides a performance comparison of the network \GLASSO frequentist model with the network spike-and-slab Bayesian model using \texttt{Stan} and \texttt{NumPyro}. 
Code to implement all of our experiments and data pre-processing is available at \url{https://github.com/llaurabat91/graphical-models-external-networks}.

\section{Implementation details for network \GLASSO and network spike-and-slab}{\label{Sec:Implmentation_details}}

\subsection{Bounding the region for optimal \GOLAZO hyperparameters $\beta$.}{\label{app:regions_beta}}

%\jack{These may need generalising to beyond binary matrices but are useful to have here for reference}

%We describe simple bounds to limit the grid search for the hyper-parameter $\hat{\beta}$ optimising the \BIC/\EBIC.
The \GOLAZO algorithm (Section~8.1 in \cite{lauritzen:2020}) is a block coordinate descent algorithm where the $j$-th row is optimised with other entries of $\Sigma$ fixed by solving a quadratic program
\begin{equation}\label{eq:dual}
\min_d \; d^T (\Sigma_{\setminus j})^{-1} d \qquad\mbox{subject to }  |\Sigma_{ij}-S_{ij}|\leq \lambda_{ij}\mbox{ for all }i<j\mbox{ and }\Sigma_{ii}=S_{ii}\mbox{ for all }i,
\end{equation}
where $d$ contains the off-diagonal entries of the $j$-th row of $\Sigma$ (the diagonal entry always satisfies $\Sigma_{jj}=S_{jj}$). 

The following lemma guarantees that for large enough $\lambda_{ij}$ the solution is to set all parameter estimates to zero.

%In our case $\lambda_{ij}=\lambda_0$ if $ij\notin E$ and $\lambda_{ij}=\lambda_1$ if $ij\in E$.
%Denote $M_0=\max_{ij\notin E}|S_{ij}|$ and $M_1=\max_{ij\in E}|S_{ij}|$.
%\begin{lemma}
%If $\lambda_0\geq M_0$ and $\lambda_1\geq M_1$ then $y=0$ is the optimal $y$ in  \eqref{eq:dual} is %the origin.
%\end{lemma}
%\begin{proof}
%Under the given condition $y=0$ is always feasible. Since this is the global minimum, the result follows.
%\end{proof}

%By the above lemma we can always assume that either $\lambda_0\leq M_0$ or $\lambda_1\leq M_1$. Denoting $m_i=\log M_i$ for $i=0,1$ we get
%$$
%\beta_0\leq m_0+\overline B \beta_1\qquad\mbox{or}\qquad  \beta_0\leq m_1-(1-\overline B) \beta_1
%$$
%From this we can get explicit constraints on the possible $\beta_0,\beta_1$ to consider. In particular, for every $\beta_1$ we get an explicit upper bound on $\beta_0$:
%$$
%\beta_0\leq \max\{m_0+\overline B \beta_1,m_1-(1-\overline B) \beta_1\}.
%$$
%Similarly, if we fix $\beta_0$ then the resulting constraints on $\beta_1$ take the form:
%$$
%\beta_1\geq \frac{1}{\overline B}(\beta_0-m_0)\qquad\mbox{or}\qquad \beta_1\leq \frac{1}{1-\overline B}(m_1-\beta_0)
%$$
%If $\beta_0$ is sufficiently small then this does not give any valid constraint, that is, $\beta_1\in \mathbb R$. However, for the type of $\beta_0$ we see in our simulations, we get two disjoint (infinite) intervals.

\begin{lemma}
If $\lambda_{jk}\geq |S_{jk}|$ for all $k\neq j$ then $d=0$ optimises \eqref{eq:dual}.
\end{lemma}
\begin{proof}
Under the given condition $d=0$ is always feasible. Since $d=0$ is also the global minimum, the result follows.
\end{proof}

We can therefore assume that $\lambda_{jk}< |S_{jk}|$ for at least one pair $(j,k)$. That is, we may restrict attention to $\beta$ satisfying 
$$\max_{j \neq k} \log(\lambda_{jk})\;=\; \max_{j \neq k} \beta_0 + \sum_{q=1}^Q \beta_q \bar{a}_{jk}^{(q)} \;\leq\; \max_{j \neq k} \log |S_{jk}|.$$
Note that this expression bounds the range of possible optima for each $\beta_q$ given the rest, and in particular for $\beta_0$ we obtain
$$
\beta_0\leq \max_{j\neq k}\{\log |S_{jk}|- \sum_{q=1}^Q \beta_q \bar{a}_{jk}^{(q)} \},
$$
which is $\leq \max_{j \neq k} \log |S_{jk}|$ at the initialisation step where $\beta_1=\ldots =\beta_Q=0$.

In particular, we propose the following procedure. First initialise $\hat{\beta}_0$ (the first entry in $\hat{\beta}$), such that $\hat{\lambda} = \exp(\hat{\beta}_0)$, where $\hat{\lambda}$ maximises the \BIC in (4) over a univariate grid. Assuming that all variables in $y_i$ are standardised to unit sample variance, the grid search is facilitated by Lemma \ref{eq:dual} and the analytic upper bound that 
\begin{equation}
    \hat{\beta}_0 \leq \log\left(\max_{k \neq j}\left\{\left|R_{jk}\right|\right\}\right), %\leq 0
\end{equation}
where $R$ is the empirical correlation matrix. %\jack{We moved from S to R - is it obvious in the paper we apply the GOLAZO to R?} 

Second, we conduct a grid search on the whole vector $\beta$, with the first entry being centered around $\hat{\beta}_0$. 
The grid search is again facilitated by the Lemma \ref{eq:dual} which shows that one may restrict attention to $\beta$ such that
$$
\max_{j \neq k} \lambda_{jk}= \max_{j \neq k} e^{ \beta_0 + \sum_{q=1}^Q \beta_q \bar{a}_{jk}^{(q)} } \leq 1-|R_{jk}|
%-|S_{jk}|+\sqrt{S_{jj}S_{kk}}
$$
since increasing $\lambda_{jk}$ beyond this bound will not change $\hat{\Theta}$.
Within the grid search, we also use the solution obtained for a particular $\beta$ as a warm start for subsequent values of $\beta$.

%\color{purple}
Further, the fact that $\Sigma$ in \eqref{eq:dual} must be positive definite allows for the construction of further simple bounds. For every $i\neq j$ we necessarily have $\Sigma_{jk}^2\leq  \Sigma_{jj}\Sigma_{kk}=S_{jj}S_{kk}$, or equivalently, $\Sigma_{jk}\in[- \sqrt{S_{jj}S_{kk}},\sqrt{S_{jj}S_{kk}}]$. It follows that, without loss of generality, we can restrict attention to that $\lambda_{jk}\leq \sqrt{S_{jj}S_{kk}}-|S_{jk}|$ giving
$$
\beta_0 + \sum_{q=1}^Q \beta_q \bar{a}_{jk}^{(q)}\;\leq\; \log(\sqrt{S_{jj}S_{kk}}-|S_{jk}|)\qquad\mbox{for all }j\neq k.
$$

\subsection{Implementation of spike-and-slab}{\label{ssec:ss_implmentation}}

If the spike has a very small variance, or the slab has too bigger variance it can be difficult for an \MCMC sampler to efficiently explore both spaces. We use a rescaling trick to facilitate efficient \MCMC inference for the network spike-and-slab model. Rather than sample directly from $\pi(\rho)$ as defined by (8), for each $\rho_{jk}$ we define latent variables $\tilde{\rho}_{jk}^{spike}$, $\tilde{\rho}_{jk}^{slab}$ and $u_{jk}$. We then sample 
\begin{align}
    \tilde{\rho}_{jk}^{spike} &\sim \mbox{DE}\left(0, 1 \right), \quad \tilde{\rho}_{jk}^{slab} \sim \mbox{DE}\left(0, 1 \right) \textrm{ and}\quad  u_{jk} \sim \mbox{Unif}[0, 1],   \nonumber
\end{align}
and set 
\begin{align}
    \rho_{jk} = \mbox{I}(u_{jk} > w_{jk})\left(s_0 \times \tilde{\rho}_{jk}^{spike}\right) + \mbox{I}(u_{jk} \leq w_{jk})\left(\eta_0^T a_{jk} + s_{jk} \times \tilde{\rho}^{slab}_{jk}\right).\nonumber
\end{align}
It is straightforward to see that the marginal distribution of $\rho_{jk}$ matches that defined in (8). Lastly, to make such an implementation suitable for \MCMC samplers that require differentiability, we approximate the indicator $\mbox{I}(u_{jk} > w_{jk})$ with a sigmoid function 
\begin{align}
    \mbox{I}(x\geq 0) &\approx \sigma_{k}(x) = \frac{1}{1 + \exp( - kx)} \textrm{ for large } k,\nonumber
\end{align}
taking $k = 100$.

\subsection{Prior elicitation}
\label{ssec:prior_elicitation}

%Here we discuss the prior elicitation for the Bayesian \GLASSO, \GOLAZO, \GLASSOSS, and \GOLAZOSS. 
We elicit spike-and-slab prior parameters $(\eta_0,\eta_1,\eta_2)$ that encourage sparse solutions, avoid pathological values, and maintain their specified intuition whilst being minimally informative.
We finish this section with a table of the values used in the simulations and in our applications.
For interpretability, we treat the spike's scale parameter $s_0$ as a constant. Recall that the spike captures partial correlations $\rho_{jk}$ that are considered to be 0 for all practical purposes, which here we consider to be $|\rho_{jk}| < 0.01$.
We hence set $s_0$ such that the spike has most of its density below this threshold, i.e. $\Pi(\rho_{ij} \in (-\tau, \tau); s_0) = 0.95$, where $\tau = 0.01$. This gave the value $s_0 = 0.003$

Consider first the hyperparameters $(\eta_{00},\eta_{10},\eta_{20})$ defining the intercept of the regression of the slab's mean, variance, and prior probability on the network data. We set the priors
\begin{align}
    \eta_{00} &\sim \mathcal{N}\left(0, g_{0}^2\right)\nonumber\\
    \eta_{10} &\sim \mathcal{N}\left(m_{1}, g_{1}^2\right)\nonumber\\
    \eta_{20} &\sim \mathcal{N}\left(m_{2}, g_{2}^2\right).\nonumber
\end{align}
For the hyperparameters that capture the effect of each network $A^{(q)}$, where $q=1,\ldots, Q$, we set
\begin{align}
    \eta_{0q} &\sim \mathcal{N}\left(0, g_{0}^2\right)\nonumber\\
    \eta_{1q} &\sim \mathcal{N}\left(0, g_{1}^2\right)\nonumber\\
    \eta_{2q} &\sim \mathcal{N}\left(0, g_{2}^2\right).\nonumber
\end{align}
%
%\david{Jack, I simplified the notation of the prior hyperparameters, before there was a triple index that did not seem to be really needed (given the defaults that we use). To avoid nasty questions best to present as simple a version as possible. Also prior sd's are denote by $g$'s rather than $s$'s to avoid confusing the notation with sample covariances and the slab's scale $s_1 := s_0(1 + \exp\left\{-\eta_1\right\})$.}
%
Centering the prior of $\eta_{00}$ at 0 encodes the absence of information about whether partial correlations are positive or negative on average.
Similarly, centering the priors of $(\eta_{0q},\eta_{1q},\eta_{2q})$ at zero reflects no prior knowledge on whether the network data are predictive of $\rho$ and in which direction.
To set the remaining hyperparameters we assume the networks have been standardised and conduct the prior elicitation for the average value of the networks (i.e. $\bar{a}_{jk}^{(q)}=0$ for all networks $q$). As a result, our prior elicitation is invariant to the network(s) considered.
%The remaining prior hyperparameters were set as follows.

The prior on $\eta_2$ was set based on sparsity and minimal informativeness considerations.
Specifically, we set the prior expected number of edges (non-zero partial correlations) to scale linearly with $p$, so that each node is expected to have a constant degree as $p$ grows. 
%None of the methods we propose in this paper produce exact 0's for the partial correlation or precision matrices, \jack{reference}. As a result, we post-process both the Bayesian and Frequentist estimates of $\Theta$ by rounding them to a certain number of decimal places. This is important when evaluating the \EBIC objective function or calculating false positives and negatives. For both our simulated and real experiments we round at 3 decimal places, meaning that values below $5 \times 10^{-4}$ are considered as 0.
%Unlike the double exponential prior associated with \GLASSO and \GOLAZO,
When all networks are at their average value 
%(i.e. $\bar{a}_{jk}^{(q)}=0$ for all networks $q$) 
the slab prior probability is $w = 1/(1 + e^{-\eta_{20}})$.
A standard non-informative prior on slab prior probabilities is a $\mbox{Beta}(m_w v_w, m_w (1-v_w))$ distribution \citep{scott:2010}, where $m_w$ is the prior mean and $v_w$ is often interpreted as the prior `sample size'. We take the minimally informative choice $v_w=1$. Regarding $m_w$, we set it such that the prior expected number of edges is $p$. Since the prior expected number of edges is
\begin{align}
    %\textrm{expected \# edges } &= 
    \mathbb{E}\left[\sum_{j = 1}^p \sum_{k < j}\mathbb{I}(\rho_{jk} \in \textrm{slab})\right]
                                = \frac{p(p-1)}{2}w\nonumber,
\end{align}     
for $m_w = \frac{2}{(p-1)}$ the expected number of edges is $p$. 
Based on these considerations, we set the $(m_2,g_2^2)$ featuring in the prior of $\eta_{20}$ and $\eta_{2q}$ so that the implied prior on $w$ has the same mean and variance as the Beta prior described above.

%\david{Jack, your text said that you matched the prior mean and variance of $\eta_2$, which wouldn't make much sense. I assumed you meant that you set the prior on $\eta_2$ such that the implied prior on $w$ has the same mean and variance as the Beta prior?
%Also, your earlier description referred to $\eta_2$ but really you meant $\eta_{20}$ if I understood correctly?}

Regarding the prior on $\eta_1$, we considered that for the slab to capture non-zero partial correlations its prior scale parameter at the average value of the networks $s_{jk} = s_0(1 + \exp\left\{-\eta_{10}\right\})$ should be significantly larger than that of the spike $s_0$.
%so as not to become too `spikey'. The intuition of the spike-and-slab prior is that the spike corresponds to those evaluates that are 0, and the slab captures the behaviour of non-zero values. 
%If the slab becomes too `spikey', then it is possible that values that were truly non-zero get assigned to the spike rather than the slab. 
We hence set $m_{1}$ and $g_{1}$ such that the prior mode of $s_1$ is $10 \times s_0$, as well as $s_1 > 3 \times s_0$ with prior probability 0.99.

Finally, the prior on $\eta_0$ was set based on prior positive-definiteness considerations.
Specifically, the positive-definiteness indicator $\mbox{I}(\rho \succ 0)$ induces dependence in the spike-and-slab prior density,
%Although the Bayesian extensions to the \PCGLASSO and \GOLAZO specify independent priors for the off-diagonal elements of $\rho$, the joint distribution of $\rho$ is constrained to give 0 density to any non-positive-definite matrix.  As a result, if generating such elements independently does not result in a positive definite matrix with high probability, then the implicit prior after constraining to the space of positive-definite matrices 
i.e. it can produce a joint prior that is vastly different from the product of independent priors on each $\rho_{jk}$. Such a discrepancy is undesirable for prior interpretation, particularly in our setting where the priors and their hyperparameters are objects of interest that describe how $\rho_{jk}$ depends on network data. To address this issue we set prior parameters such that the prior probability of $\rho$ being positive definite when independently sampling its elements is at least 0.95.
Conditional on the priors specified for $(\eta_1,\eta_2)$, $g_{0}$ was set to the largest value (i.e. least informative) that guarantees at least 0.95 probability that $\rho$ is positive-definite under independent sampling from the unconstrained spike-and-slab prior components. 

%\subsubsection{\GOLAZOSS}

%\david{Jack, please remember to delete whatever old elicited values we're no longer using. We should provide them for both applications.}

%\subsubsection{Elicited values - OLD}

%The below table summarises the values elicited for the dimensions $p$ and network matrices considered in this paper. 

%\begin{table}[ht]
%\centering
%\begin{tabular}{rrrrrrr}
%  \hline
%   $p = 10$  & $m_{000}$ & $s_{000}$ & $m_{010}$ & $s_{010}$ & $m_{020}$ & $s_{020}$ \\ 
%   \hline
%   \GLASSOSS & 0.000 & 0.159 & -4.595 & 1.042 & -2.722 & 3.056 \\ 
%   \GOLAZOSS: A-true & 0.000 & 0.106 & -4.595 & 0.650 & -2.444 & 2.167 \\ 
%   \GOLAZOSS: A-semidep & 0.000 & 0.121 & -4.595 & 0.700 & -2.444 & 1.944 \\ 
%   \GOLAZOSS: A-indep & 0.000 & 0.132 & -4.595 & 0.750 & -2.833 & 2.167 \\ 
%   \hline
% \end{tabular}
% \label{Tab:SS_prior_hyperparameters}
% \end{table}

%priorSpecification_GLASSO_GOLAZO2.Rmd
%\begin{figure}[H]
%\begin{center}
%\includegraphics[trim= {0.0cm 0.0cm 0.0cm 0.0cm}, clip,  
%width=0.49\columnwidth]{plot/SS_simulations/GLASSO_SS_p10_prior_elicit_plot_tikz-1.pdf}
%trim={<left> <lower> <right> <upper>}
%\caption{Simulations from the \GLASSO prior predictive for $\rho$ with priopr hyperparameters as in \ref{Tab:SS_prior_hyperparameters}
%}
%\label{Fig:Rho_prior_predictive_sim}
%\end{center}
%
%\end{figure}

\subsubsection{Elicited values}

%\jack{I MADE A MESS HERE, THESE SHOULD BE VALUES ELICITED BY THE GLASSO-SS VERSION WHICH IS THE SAME AS ASSUMING THE NETWORKS WERE 0. THIS HAS BEE DONE FOR P=10 AND P=50 BU TNOT THE COVID AND STOCK MARKET DATA}

Table \ref{Tab:SS_prior_hyperparameters} presents the elicited values used in our simulations and real data examples. Code to elicit priors following the specification above for further examples is available in the GitHub repository.
%
%
%% IF WE RUN THESE AGAIN REMEBER THAT I FUCKED UP HERE SLIGHTLY
%%
\begin{table}%[ht]
\caption{Network spike-and-slab prior hyperparameters}
\centering
\begin{tabular}{rcccc}
  \hline
 & $p = 10$ & $p = 50$ & \COVID data ($p=332$) & Stock data ($p=366$) \\ 
  \hline
  $s_0$ & 0.003 & 0.003 & 0.003 & 0.003 \\ 
  $g_0$ & 0.145 & 0.152 & 0.002 & 0.002 \\ 
  $m_1$ & -2.197 & -2.197 & -2.197 & -2.197 \\ 
  $g_1$ & 0.661 & 0.661 & 0.3 & 0.35 \\ 
  $m_2$ & -2.722 & -6.737 & -7.789& -10.16 \\ 
  $g_2$ & 3.278 & 3.395 & 1.02 & 1.81 \\ 
\hline
\end{tabular}
\label{Tab:SS_prior_hyperparameters}
\end{table}
As the dimension of the data increases, only the prior for $\eta_2$ changes greatly. This is a result of the assumption that the number of edges grows linearly with $p$, and therefore $\Theta$ is \textit{a priori} assumed more sparse for larger $p$.

%Taken from the paper 
To assess the impact of these default prior choices, it is useful to display the implied prior marginal distribution on the $\rho_{jk}$'s. Figure \ref{Fig:Rho_prior_predictive_sim} shows that in both the \COVID and stock market applications most of the prior probability is contained in $\rho_{jk} \in (-0.5,0.5)$, which seems a sensible prior interval. 
The prior concentrates significant mass around 0, which induces shrinkage, but also features thick tails, which favors capturing truly non-zero $\rho_{jk}$'s.
Indeed, the corresponding posteriors (Fig. \ref{Fig:Rho_prior_predictive_sim}, bottom panels) set significant mass away from zero, suggesting that the prior shrinkage towards 0 was not excessive.

%priorSpecification_GLASSO_GOLAZO2_newSpike.Rmd
\begin{figure}[!ht]
\begin{center}
\includegraphics[width =0.49\linewidth]{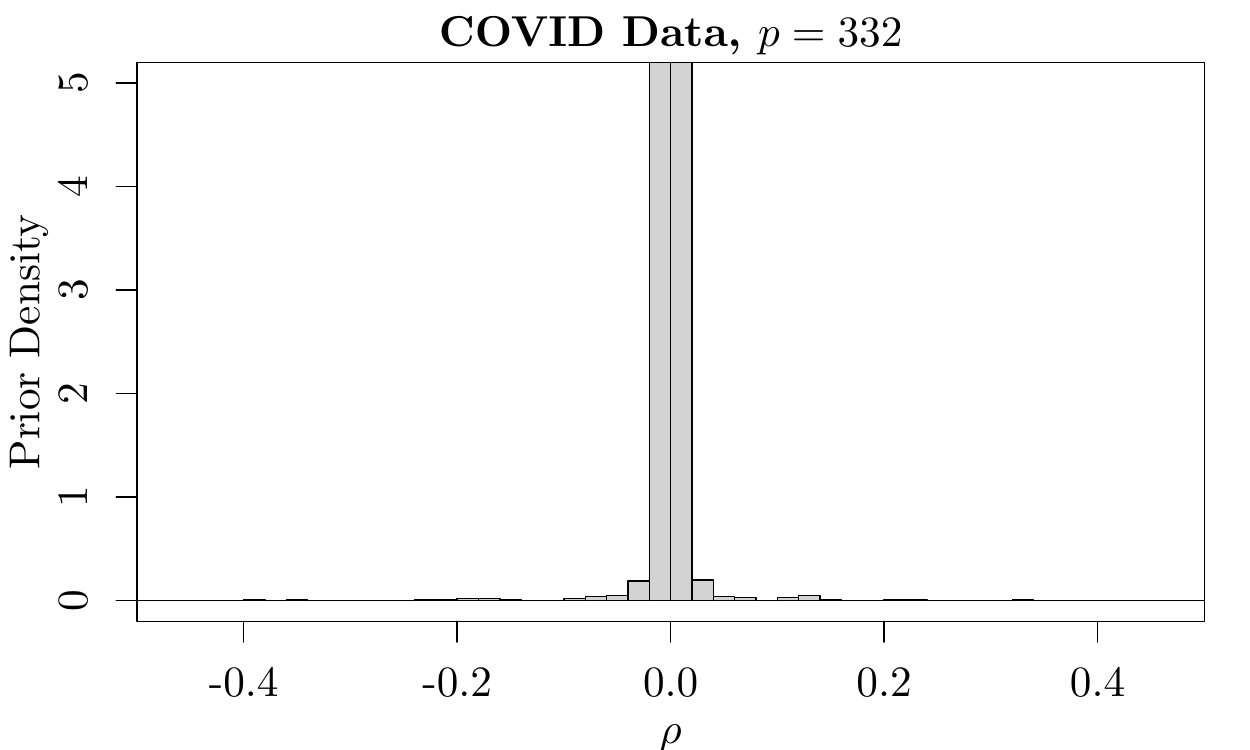}
\includegraphics[width =0.49\linewidth]{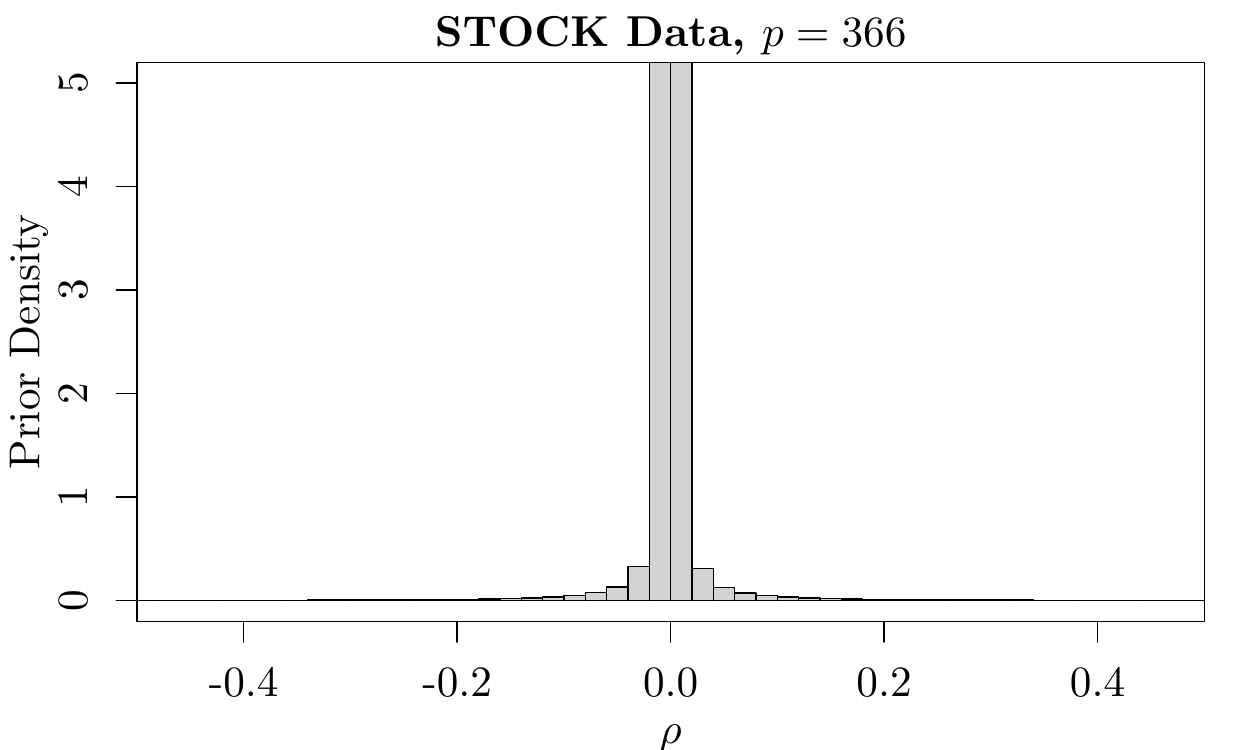}\\
\includegraphics[width =0.49\linewidth]{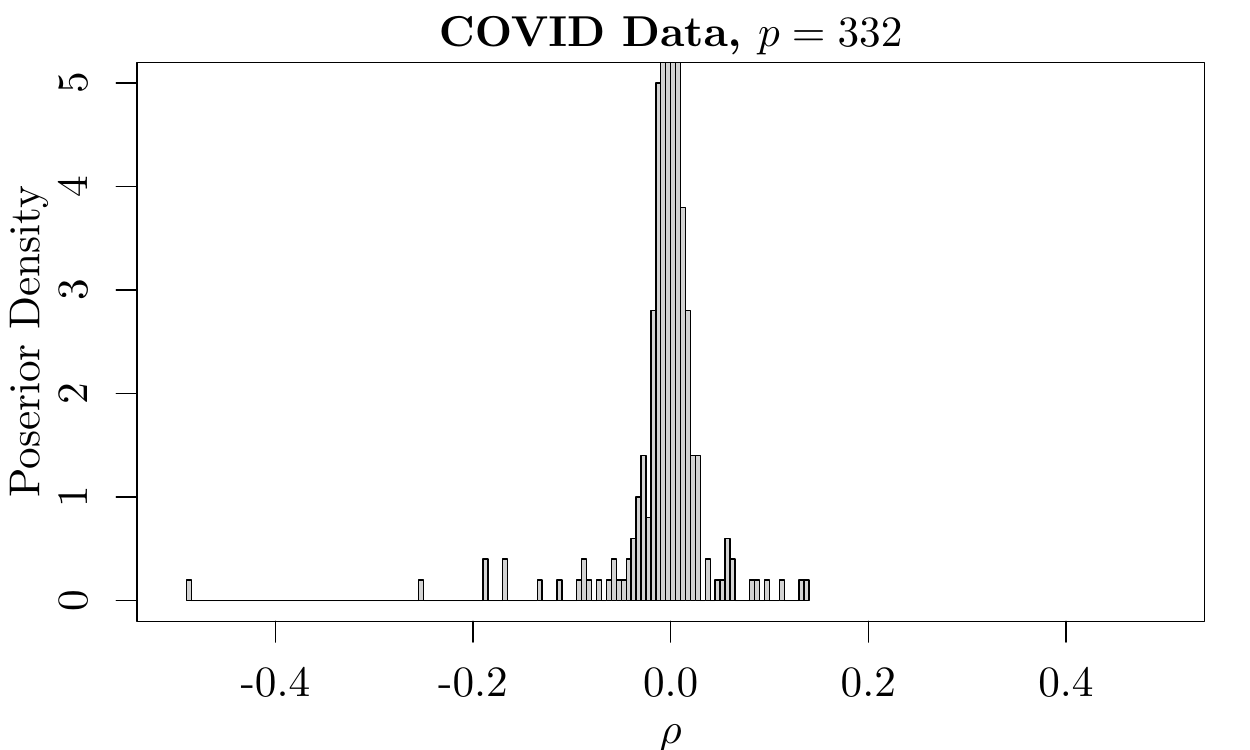}
\includegraphics[width =0.49\linewidth]{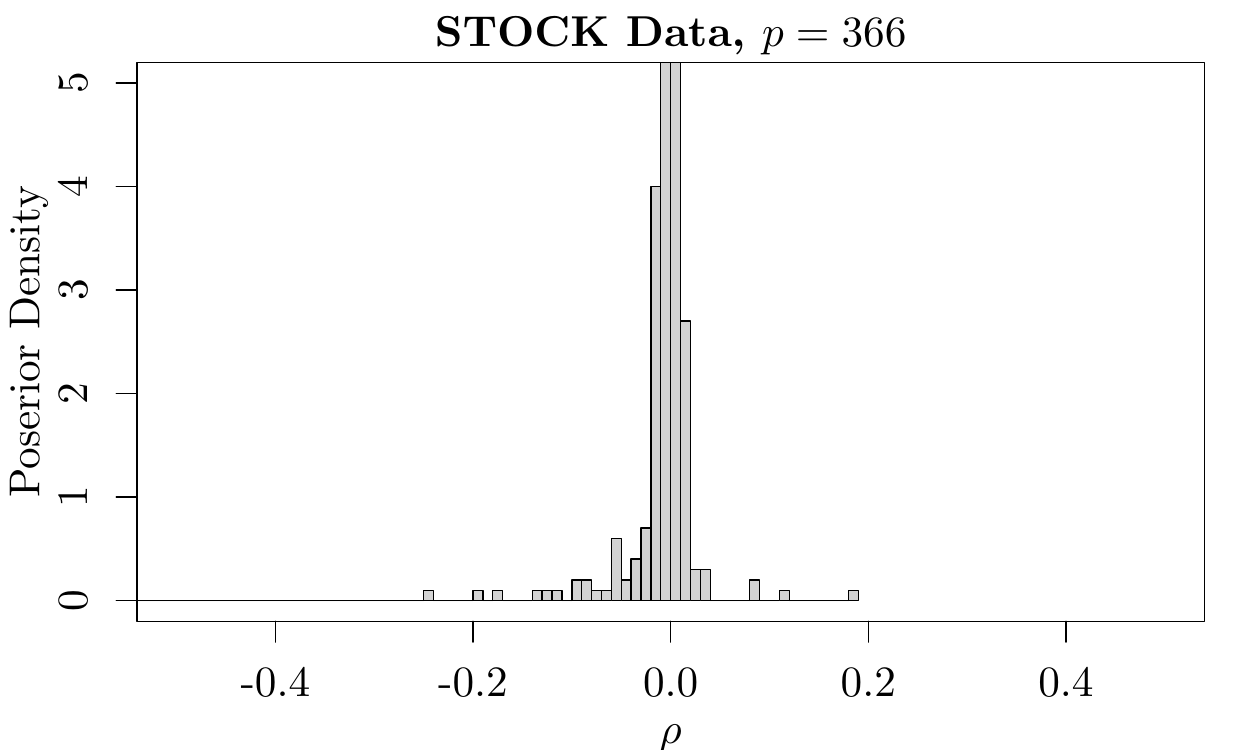}
\caption{Elicited prior distribution and posterior distribution for $\rho_{jk}$, $j = 1,\ldots, p$, $k < j$ for the \COVID data $p = 332$ and stock market data $p = 366$. %\jack{these are for both matrices equal to 0}%\jack{Add posterior draws to this plot for the whole $\rho$ show prior not too informative, maybe this goes in the appendix}
}
\label{Fig:Rho_prior_predictive_sim}
\end{center}
\end{figure}

\subsection{Reparametrisation of the network hyperparameters}{\label{ssec:ss_repar}}

An advantage of the Bayesian network spike-and-slab approach is that it allows us to do inference for the network hyperparameter as was done in Tables 3 and 5. Such inferences, however, require that the effective sample size (\ESS) of the sampled hyperparameters is sufficiently high. We observed empirically that hyperparameters attain lower \ESS. Although this phenomenon has not been studied in our graphical model settings, in hierarchical models it is well understood that parameters associated to higher levels have strictly slower MCMC mixing, and that said mixing can be improved by reparameterising the problem \citep{zanella2021multilevel}. We applied the following transformation of the hyperparameters to facilitate their sampling.

%Rather than sample,
%\begin{align}
%    \eta_{i0} &\sim \mathcal{N}\left(m_i,  g_{0}^2\right)\nonumber\\
%    \eta_{iq} &\sim \mathcal{N}\left(0, g_{i}^2\right), \quad i = 1, 2, 3, \quad j = 0, 1, \ldots, Q,\nonumber
%\end{align}
Rather than sample directly from the priors for the hyperparameters as outlined in Section \ref{ssec:prior_elicitation}, we reparameterised and sampled 
\begin{align}
    \tilde{\eta}_{iq} &\sim \mathcal{N}\left(0, \frac{p(p-1)}{2n}\right), \quad i = 1, 2, 3, \quad q = 0, 1, \ldots, Q.\nonumber
\end{align}
The original $\eta$ hyperparameters can then be recovered as
\begin{align}
    \eta_{i0} &= m_i + \tilde{\eta}_{i0} \times g_i / \sqrt{p(p-1)/2n},\nonumber\\
    \eta_{iq} &= 0 + \tilde{\eta}_{iq} \times g_i / \sqrt{p(p-1)/2n}, \quad i = 1, 2, 3, \quad q =  1, \ldots, Q,\nonumber
\end{align}
where $m_0 := 0$. 
The idea behind this is to first standardise the $\eta$'s to all have mean 0 and variance 1, before adjusting the variance of the $\tilde{\eta}$'s by the square-root of the ratio of the number of $\rho$'s ($p(p-1)/2$) from which the $\eta$'s are learned, to the number of observations $Y$ ($n$) from which the $\rho$'s themselves are learned. Such a reparametrisation leaves the model completely unchanged, but we found this improved the ESS of the $\eta's$.

%\subsection{Simulation results using the \EBIC}{\label{Sec:EBIC_simulations}}

\subsection{Additional simulation results}{\label{Sec:EBIC_simulations}}

\subsubsection{Simulations with $n = 500$}

Table \ref{tab:sim_results0.95_n500} presents simulation results from Section 4 in the additional case where the sample size $n = 500$. These show that when $n$ is large relative to $p$, network information helps to a lesser extent. 

\begin{table}[!ht]
    \centering
    \renewcommand{\arraystretch}{0.75}
    %\smaller
    \caption{Simulation results for $n = 500$ under non, mildly and strongly informative networks $A_{ind}$, $A_{0.75}$ and $A_{0.85}$. For SS and network SS models edges declared when  posterior probability $> 0.95$.}
    \begin{tabular}{|cc|ccc|ccc|} \hline
    & & \multicolumn{3}{c|}{$p=10$} & \multicolumn{3}{c|}{$p=50$} \\
       & $n$ & MSE & FDR & FNR & MSE & FDR & FNR \\ \hline
  \GLASSO                     & 500 &0.082  &0.367  &0.002  &0.825 &0.410  &0.032  \\
  Network \GLASSO, $A_{ind.}$ & 500 &0.085  &0.315  &0.007  &0.766 &0.443  &0.035 \\
  Network \GLASSO, $A_{0.75}$ & 500 &0.066  &0.270  &\textbf{0.000}  &0.604 &0.419  &0.031  \\
  Network \GLASSO, $A_{0.85}$ & 500 &0.045  &0.195  &0.008  &0.512 &0.386  &0.027  \\
  SS & 500 & \textbf{0.030} & 0.000  & 0.023 & 0.198 & 0.002 & 0.009\\
  Network SS, $A_{ind.}$      & 500 & 0.034 & \textbf{0.000} & 0.023 & 0.201 & \textbf{0.001} & 0.010 \\
  Network SS, $A_{0.75}$      & 500 & 0.032 & 0.002 & \textbf{0.018} & 0.193 & 0.001 & 0.009 \\
  Network SS, $A_{0.85}$      & 500 & 0.033 & 0.008 & 0.022 & 0.183 & 0.001 & 0.009 \\
  \siGGM, $A_{ind}$     & 500 & 0.104 & 0.658 & 0.000   & 0.968 & 0.775 & \textbf{0.007} \\ 
  \siGGM, $A_{0.75}$    & 500 & 0.068 & 0.478 & 0.001   & 0.606 & 0.712 & 0.008 \\   
  \siGGM, $A_{0.85}$    & 500 & 0.047 & 0.375 & 0.021   & 0.524 & 0.683 & 0.008 \\
       \hline
    \end{tabular}
 
    \label{tab:sim_results0.95_n500}
\end{table}

\subsubsection{The \EBIC to learn the network hyperparameters}

As a sensitivity check, we also consider using the \EBIC \citep{chen:2008} to select hyperparameters for the \GLASSO and Network \GLASSO models
\begin{align}
    \EBIC(\lambda) &= -2\ell_n(\hat{\Theta}(\lambda)) + |\mathbf{E}(\hat{\Theta}(\lambda))|\log n + 4|\mathbf{E}(\hat{\Theta}(\lambda))|\gamma_{\EBIC} \log p
    \label{eq:ebic_optim}
\end{align}
Compared with the \BIC, \eqref{eq:ebic_optim} has an additional complexity penalty, controlled by hyperparameter $\gamma$. \cite{foygel2010extended} recommend $\gamma_{\EBIC} \in[0, 0.5]$ where $\gamma_{\EBIC} = 0$ recovers the \BIC. Table \ref{tab:sim_results_EBIC} presents the results of the experiments introduced in Section 4 when using the \EBIC with $\gamma_{\EBIC} = 0.5$ to select hyperparameters. Comparing these results with Table 1 shows that using the \EBIC reduced the FDR relative to the \BIC, however, this generally results in much more conservative edge selection which damaged the MSE.

\begin{table}[]
    \centering
    %\caption{\GLASSO and network \GLASSO simulation results under non-informative network $A_{ind}$, mildly and strongly informative networks  $A_{0.75}$ and $A_{0.85}$ with \EBIC rule ($\gamma_{\EBIC} = 0.5$) for learning the $\beta$ hyperparameters. 
    \caption{\GLASSO and network \GLASSO simulation results under non, mildly and strongly informative networks $A_{ind}$, $A_{0.75}$ and $A_{0.85}$ with \EBIC rule ($\gamma_{\EBIC} = 0.5$) for learning the $\beta$ hyperparameters.
    %\jack{Why do we have NaN FDR, do we have 0 over 0? we can define this as 0} \li{Yes, we have several 0 over 0 here.} \jack{THANKS LI David said when we have 0/0 this for FDR is defined to be 0, can you replace the NaNs you get with 0's and then redo the average please, I think this only happens here}
    %\li{Thanks Jack. Sure no problem, I've updated the results.}
    } 
    \label{tab:sim_results_EBIC}
    \begin{tabular}{|cc|ccc|ccc|} \hline
    & & \multicolumn{3}{c|}{$p=10$} & \multicolumn{3}{c|}{$p=50$} \\
       & $n$ & MSE & FDR & FNR & MSE & FDR & FNR \\ \hline
  \GLASSO                     & 100 &0.474  & 0.243   &0.176  &6.628 &0.163    &0.566 \\
  Network \GLASSO, $A_{ind.}$ & 100 &0.556  &0.163  &0.253  &7.008 &0.128  &0.632  \\
  Network \GLASSO, $A_{0.75}$ & 100 &0.383  &0.138  &0.162  &5.691 &0.112  &0.504   \\
  Network \GLASSO, $A_{0.85}$ & 100 &0.195  &0.103  &0.153  &4.566 &0.098  &0.414   \\ \hline
  %Network SS, $A_{ind.}$      & 100 &  &  &  & &  &  \\
  %Network SS, $A_{0.75}$      & 100 &  &  &  & &  &  \\
  %Network SS, $A_{0.85}$      & 100 &  &  &  & &  &  \\ \hline
  \GLASSO                     & 200 &0.254  &0.283  &0.060  &2.726 &0.224  &0.241   \\
  Network \GLASSO, $A_{ind.}$ & 200 &0.265  &0.223  &0.082  &2.678 &0.227  &0.248   \\
  Network \GLASSO, $A_{0.75}$ & 200 &0.200  &0.176  &0.058  &2.155 &0.206  &0.216   \\
  Network \GLASSO, $A_{0.85}$ & 200 &0.108  &0.118  &0.120  &1.837 &0.188  &0.207   \\ \hline
  %Network SS, $A_{ind.}$      & 200 &  &  &  & &  &  \\
  %Network SS, $A_{0.75}$      & 200 &  &  &  & &  &  \\
  %Network SS, $A_{0.85}$      & 200 &  &  &  & &  &  \\ \hline
  \GLASSO                     & 500 &0.101  &0.281  &0.004  &0.958 &0.327  &0.138   \\
  Network \GLASSO, $A_{ind.}$ & 500 &0.099  &0.235  &0.011  &1.002 &0.286  &0.142  \\
  Network \GLASSO, $A_{0.75}$ & 500 &0.074  &0.185  &0.000  &0.781 &0.272  &0.153   \\
  Network \GLASSO, $A_{0.85}$ & 500 &0.051  &0.116  &0.096  &0.698 &0.214  &0.158   \\
  %Network SS, $A_{ind.}$      & 500 &  &  &  & &  &  \\
  %Network SS, $A_{0.75}$      & 500 &  &  &  & &  &  \\
  %Network SS, $A_{0.85}$      & 500 &  &  &  & &  &  \\
       \hline
    \end{tabular}

\end{table}

%\newpage
\section{\COVID data analysis}
\label{sec:covid_dataprocessing}

This section provides additional details for the analysis of the \COVID infection rate data. 

\subsection{Data sources}

To undertake our analysis, we collected and combined the following datasets.

1. U.S. population data\\
%We selected the top 100 counties for analysis based on 2019 U.S. population data. Data were sourced from \url{ https://www2.census.gov/programs-surveys/popest/tables/2010-2019/counties/totals/}. Due to the lack of statistics for the District of Columbia in the \COVID policy dataset, we proceed with the top 99 counties.\\
U.S. population data for 2019 were sourced from \url{ https://www2.census.gov/programs-surveys/popest/tables/2010-2019/counties/totals/}. \\

2. FIPS code data \\
To allow for a better match between different datasets, we also extracted the ``FIPS code" that uniquely identifies counties within the U.S. from the U.S. Bureau of Labor Statistics \url{https://www.bls.gov/cew/classifications/areas/sic-area-titles.htm}.\\

3. \COVID infection data \\
Time series data of confirmed \COVID infections in each U.S. county was obtained from \url{https://github.com/CSSEGISandData/COVID-19/blob/master/csse_covid_19_data/csse_covid_19_time_series/time_series_covid19_confirmed_U.S..csv}. Figure \ref{Fig:COVID_raw_cases} plots of the weekly aggregated confirmed \COVID infections\\

\begin{figure}%[hbt!]
\centering
\includegraphics[width =0.9\linewidth]{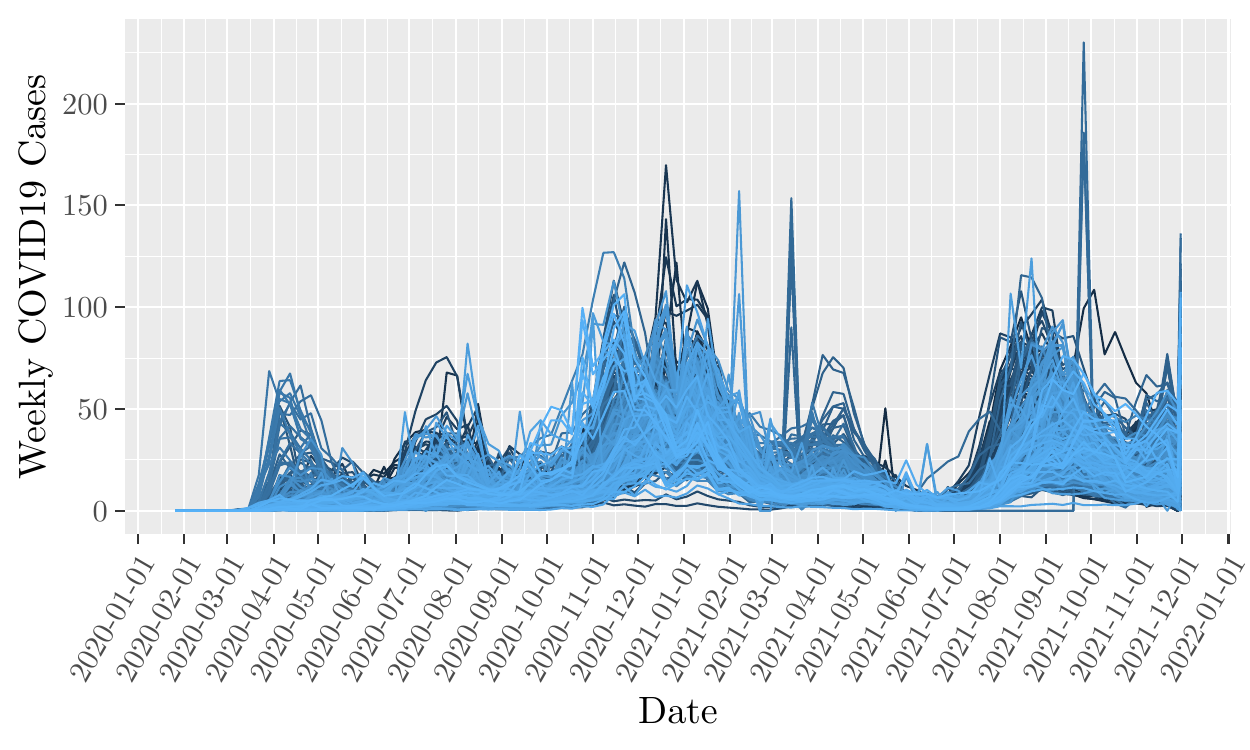} 
\caption{Weekly \COVID Cases per county for the 99 biggest counties in the U.S.}
\label{Fig:COVID_raw_cases}
\end{figure}

4. \COVID vaccination data\\
State-level vaccination data was obtained from  \url{https://github.com/govex/COVID-19/tree/master/data_tables/vaccine_data/us_data/time_series}.\\

5. Policy data\\
The Oxford \COVID Government Response Tracker \url{https://github.com/CSSEGISandData/COVID-19_Unified-Dataset} %(GitHub repo, university website) 
tracks individual policy measures across 20 indicators. They also calculate several indices to give an overall impression of government activity. We used their Containment and Health indices to summarise the policy variables.  \\

6. Temperature data\\
We extracted the daily average near-surface air temperature from the `Hydromet' folder of the above repository \url{https://github.com/CSSEGISandData/COVID-19_Unified-Dataset/tree/master/Hydromet}.\\

7. U.S. area data\\
Population densities were obtained by dividing the county population by the area of the region. Area data were obtained from the U.S. Census Bureau \url{https://tigerweb.geo.census.gov/tigerwebmain/TIGERweb_main.html}.\\

8. Geocloseness data\\
To measure the Geographical distance between two counties we use the Haversine distance \citep{sinnott1984virtues} which assumes the earth is spherical.  The latitude and longitude of each county were downloaded from the U.S. Census Bureau \url{https://tigerweb.geo.census.gov/tigerwebmain/TIGERweb_main.html}.\\

9. Facebook connectivity data\\
The Facebook Social Connectedness Index (SCI), obtained from \url{https://data.humdata.org/dataset/social-connectedness-index}, uses an anonymised snapshot of all active Facebook users and their friendship networks to measure the intensity of connectedness between locations. Specifically, it measures the relative probability that two individuals across two locations are friends with each other on Facebook.
 \\

10. Flight connectivity data\\
Flight data between airports in the US was downloaded from \url{https://essd.copernicus.org/articles/13/357/2021/essd-13-357-2021.html}
 \\

11. Airport Information\\
Name, ICAO code, and Geographical location of US airports were extracted from  \url{https://www.flightradar24.com/52.52,13.39/4}. This data allowed us to assign airports to the county(s) that they were part of. \\

12. county-to-MSA crosswalk \\
When counties are part of large urban areas known as  metropolitan statistical area (MSA) we allocate the flights proportionally to all counties in the MSA.  For example, a flight from JFK in New York to LAX in Los Angeles is not recorded between just those two counties, but rather allocated between all county pairs that form the NY and LA MSAs proportionally to the populations of these counties in the MSAs. The membership of counties to MSAs was downloaded from \url{https://www.census.gov/geographies/reference-files/time-series/demo/metro-micro/delineation-files.html}.\\

13. Flight capacities\\
The capacity of certain plane models was extracted from the folowing links
\begin{itemize}
    \item \url{https://www.seatguru.com/airlines/American_Airlines/fleetinfo.php}
    \item \url{https://www.seatguru.com/airlines/Delta_Airlines/fleetinfo.php}
    \item \url{https://www.seatguru.com/airlines/JetBlue_Airways/fleetinfo.php}
    \item \url{https://www.seatguru.com/airlines/Southwest_Airlines/fleetinfo.php}
    \item \url{https://www.seatguru.com/airlines/Spirit_Airlines/fleetinfo.php},
\end{itemize}
allowing for the estimation of the number of passengers on each flight.\\

Producing our flight connectivity network required the following steps 
\begin{itemize}
    \item Assign each airport to the county in which it is located and distribute flights between counties that make up MSA's to the other counties in the MSA proportionally to their population
    \item Use airline capacity data to estimate the number of passengers on each flight and therefore the number of passengers flowing between two counties
    \item Standardise this by the population of each county to estimate population flow
\end{itemize}

\subsection{Data processing}

Once the data was collected, some minimal data preprocessing was required to prepare the data for our analysis. This consisted mainly of variable transformation and imputing of missing values.

\subsubsection{Variables transformation}
Natural logarithms were taken of the variables `\textit{confirmed case}', `\textit{population density}' and `\textit{number of vaccinations}'.

\subsubsection{Missing values}
In addition, there were missing values in covariates Containment and Health Index data (CHI) as well as the  Temperature (Temp) data and the Vaccination data. We imputed these missing values as follows

\begin{enumerate}
    \item The CHI values were calculated as a function of different policy measures (\url{https://github.com/OxCGRT/covid-policy-tracker/blob/master/documentation/index_methodology.md}). On several occasions either these policy measures or their flags were missing. We imputed these as follows
    \begin{itemize}
        \item Missing flags were imputed as 0's, i.e. no flags
        \item Missing values before the first recorded value were imputed as 0, i.e. assuming no measures were in place before the first recorded measure
        \item Missing values in between two recorded values were imputed as an average of the before and after measures
        \item Missing values after the last recorded value were imputed as the last seen measure, i.e. assuming a continuation
    \end{itemize}
    \item For the Temp data, the temperatures for 18 counties were not recorded at all. We imputed these using the nearest county geographically whose temperature data was available.
    
    \item The vaccination data was only recorded from the 14th of December 2020 and therefore all vaccination counts before this date were imputed as 0's.
\end{enumerate}

\subsection{Meta-County Clustering}{\label{Sec:MetaCountyClustering}}

%\jack{R County pop started with 3142 counties, the FIPS 3144}
%\jack{we removed 5 counties with no Facebook connections with any other county? say not appearing in the facebook network sp 3139}
%\jack{and we removed DC and 3138}
%\jack{We removed one further county becasue of the population density info - not sure why 3137}
%\jack{I sent Laura data with 3137 counties}
%\jack{Laura removed 8 counties from this that weren;t present in the flight data leaving 3129}
%\jack{PythonPopulation data has 3141 counties, Flight data has 3144 counties}

Before Clustering the data we removed some counties whose data were not available. From the FIPS data we downloaded there were 3144 counties. District of Columbia did not have \COVID policy variable available and we could not compute the population density for Valdez-Cordova Census Area in Alaska so we removed these. Five counties were removed as they were not available in the SCI index and 8 counties in Connecticut were removed because they did not have any flight connection data. This left 3129 counties.

Starting with 3129 counties, we hierarchically clustered small counties together such that the resulting meta-counties all have population greater than 500,000. The clustering procedure is described in the following steps and is implemented by our code.
\begin{itemize}
    \item[] \textbf{Step 1}: Remove the `big' counties. Any county whose population was greater than 500,000 is extracted and left unchanged. There were 136 `big' counties leaving 2993 `small' counties.
    \item[] \textbf{Step 2}: Cluster small counties with each state. 
    \begin{itemize}
        \item[i)]  Within each state find the smallest `small' county and combine this with the `small' county within that state whose centroid is closest to create a `meta-county'
        \item[ii)] Update the county centroid as the average of the latitude and longitude of the two combined counties 
    \end{itemize}
    This procedure is repeated with each state until either all of the `meta-counties' have population greater than 500,000, or there is only one meta-county left for that state. This resulted in 196 meta counties 
\end{itemize}

This clustering procedure resulted in $p = 332$ counties. Once the meta-counties have been created the number of cases were summed, populations are combined, areas combined, temperatures averaged and the vaccination and CHI variables inputted for that state. %\laura{I confirm all these. The final dataset also includes averaged $density\_sk$ and averaged $per\_daily\_vaccinated$ (while $daily\_vaccinated$ is inputted); I do not think we used any of those?} 
This allowed the model described below to be estimated together on the large counties and the meta-counties made up of smaller counties.

\subsection{Model description}
Our final response variable is the log of the weekly \COVID infections per 10,000 members of the population (i.e. cases / population $\times$ 10,000). This results in data $y_1,\ldots, y_n$ where $y_i= (y_{i1}, ..., y_{ip})$ is the log of the standardised weekly \COVID infections at week $i$ in the $p=332$ counties and meta-counties. The sample interval is from 22 January 2020 to 30 November 2021 resulting in $n = 97$ weeks of data. 

Our graphical model posits $y_i \sim \mathcal N_p(\mu_i,\,\Theta^{-1})$ where $\mu_i= (\mu_{i1},..., \mu_{ip})$. For convenience, we decouple the estimate of $\mu_i$ from $\Theta$. We pose a regression model for $\mu_{ij}$ and then estimate $\Theta$ using the residuals of this model assuming zero mean as in Section 2. Our generalised additive regression model for $y_{ij}$ can be summarised as follows 
\begin{align}
log(confirmed)_{ij}  = b_0 &+ b_1 \times log(Lag_{confirmed})_{ij} +
b_2 \times log(popdensi)_j\nonumber\\
&+b_3 \times Cum\_vaccinated_{i,state_j} +
b_4 \times CHI_{i,state_j}  \nonumber\\
&+s(Temp)_{ij} + \gamma_2 \times Time_2 + ... + 
\gamma_T \times Time_T + \epsilon_{ij} \nonumber
\end{align}
where 
\begin{itemize}[itemindent=0pt]%,leftmargin=*]
    \item[(1)] $log(confirmed)_{ij}$ represents the natural logarithm of weekly per 10,000 people confirmed case in county $j$ at time $i$.
    \item [(2)] $log(Lag_{confirmed})_{ij}$ a first-order auto-regressive term measuring the infection rate at the previous time point $i - 1$ for each county $j$
    \item [(3)] $log(popdensi)_j$ is the population density for county $j$
    \item [(4)] $Cum$\_$vaccinated_{i, state_j}$ is the cumulative number of vaccinated individuals in the state to which county $j$ belongs by time $i$
    \item [(5)] $CHI_{i, state_j}$ represented the Containment and Health Index summarising \COVID policies/measures put in place in the state to which county $j$ belongs and time $i$ (wearing masks, closing schools, etc.)
    \item[(6)] $s(Temp)_{ij}$ is a non-parametric smooth of the average temperature for county $j$ at time $i$ implemented in \texttt{mgcv} package in \textit{R}
    \item[(7)] $Time_i$ is an indicator for week $i$ and provides a weekly fixed effect term estimating the mean infections across all counties at time $i$
    \item[(8)] $\epsilon_{ij}$ are the residuals of county $j$ at time $i$
\end{itemize}
With such a model we aim to remove the effect of the most relevant covariates that drive the mean number of infections, allowing $\Theta^{-1}$ to capture dependencies unexplained by these covariates. % and investigate the influence of the distance between counties (Geographical distance) and the degree of social connections between counties (Facebook index).

%\jack{Parameter Estimates?}

\subsection{Checking model goodness-of-fit}
\label{ssec:gof_COVID}

The main assumptions behind our assumed model require that the residuals $\epsilon_{ij}$ are Gaussian distributed and independent across $i=1,\ldots,n$ time points. We provide diagnostic plots to check these assumptions.

Figure \ref{Fig:COVID_residuals} plots the fitted values $\hat{y}_{ij}$ and each of the predictors against the residuals $\epsilon_{ij}$. This demonstrates that the assumption that the covariates are linearly related to the response is satisfactory and that the residuals appear reasonably homoskedastic. Figure \ref{Fig:COVID_qq} shows a histogram of the standardised residuals and Q-Q-normal plots for $\epsilon_{ij}$. The Gaussian assumption is tenable here.

\begin{figure}%[hbt!]
\centering
\includegraphics[width =0.49\linewidth]{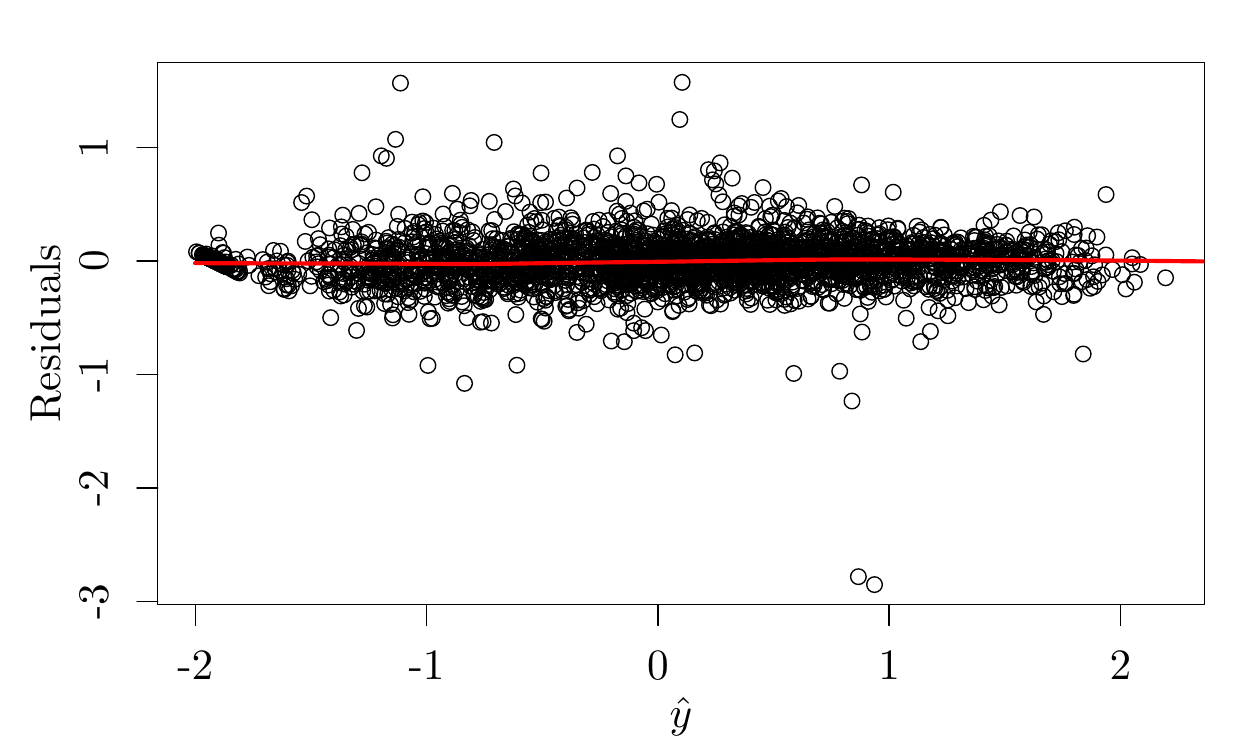} 
\includegraphics[width =0.49\linewidth]{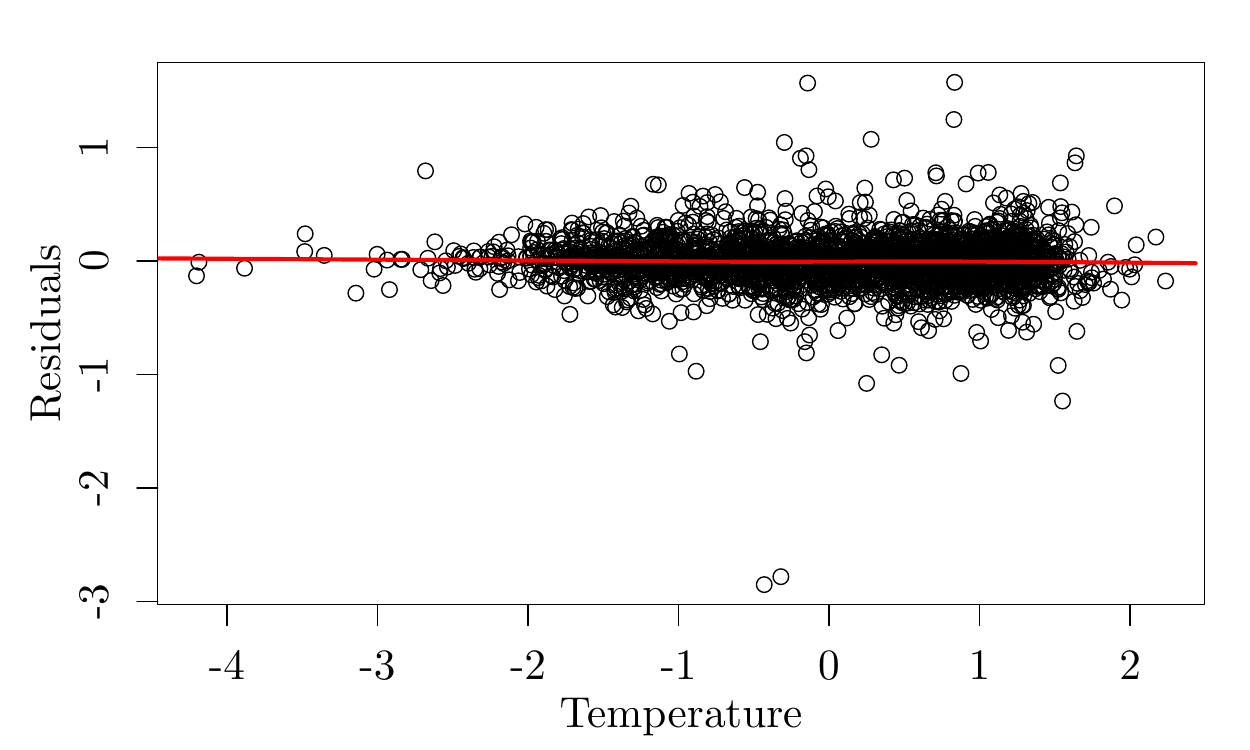}
\includegraphics[width =0.49\linewidth]{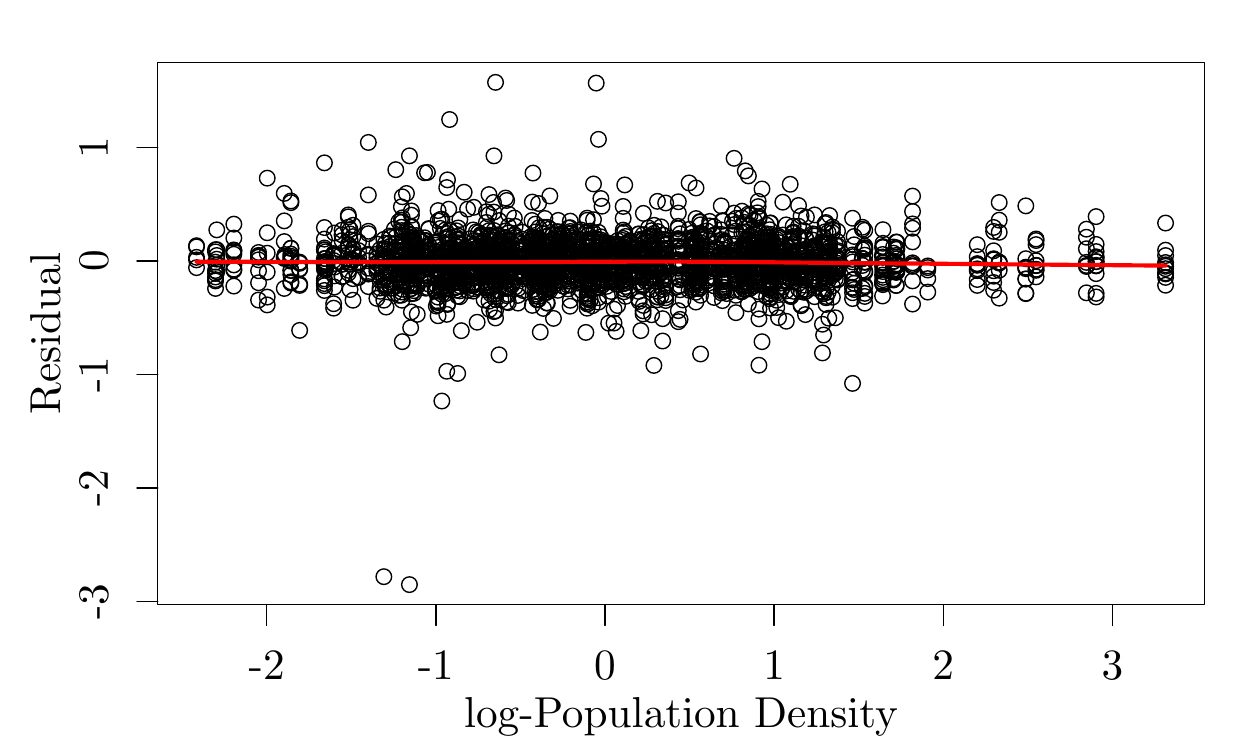}
\includegraphics[width =0.49\linewidth]{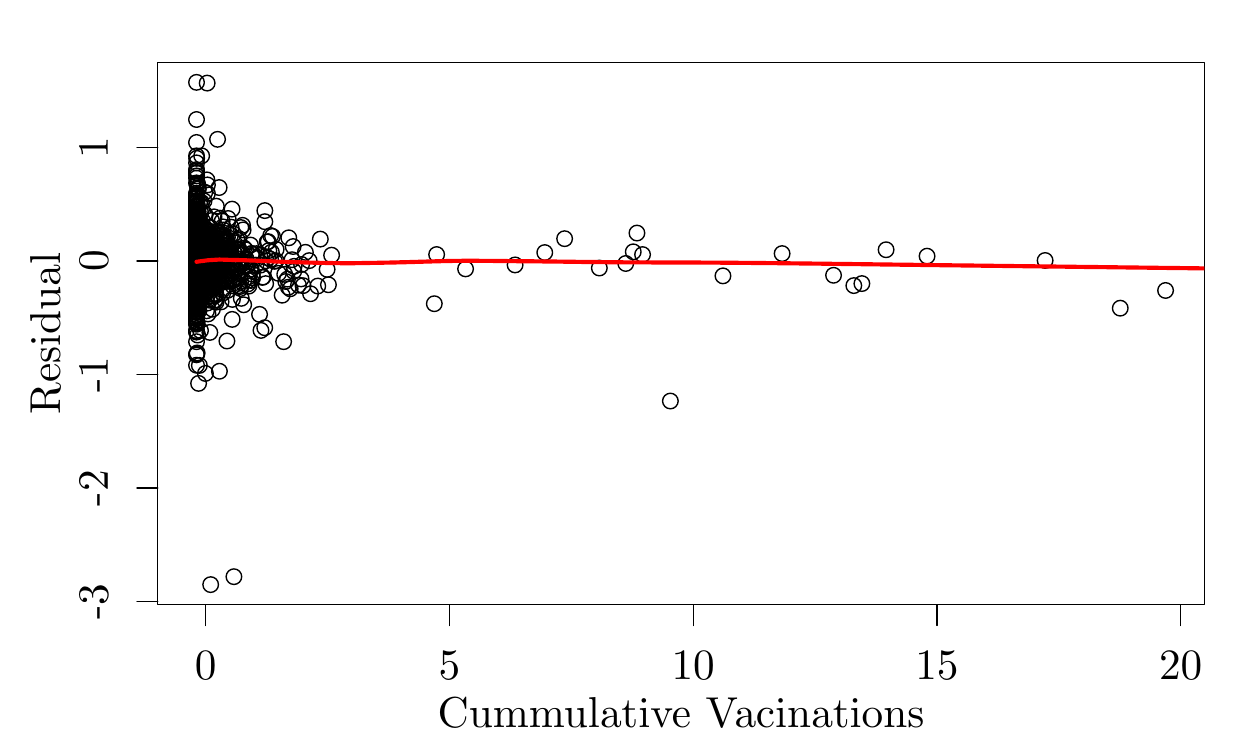}
\includegraphics[width =0.49\linewidth]{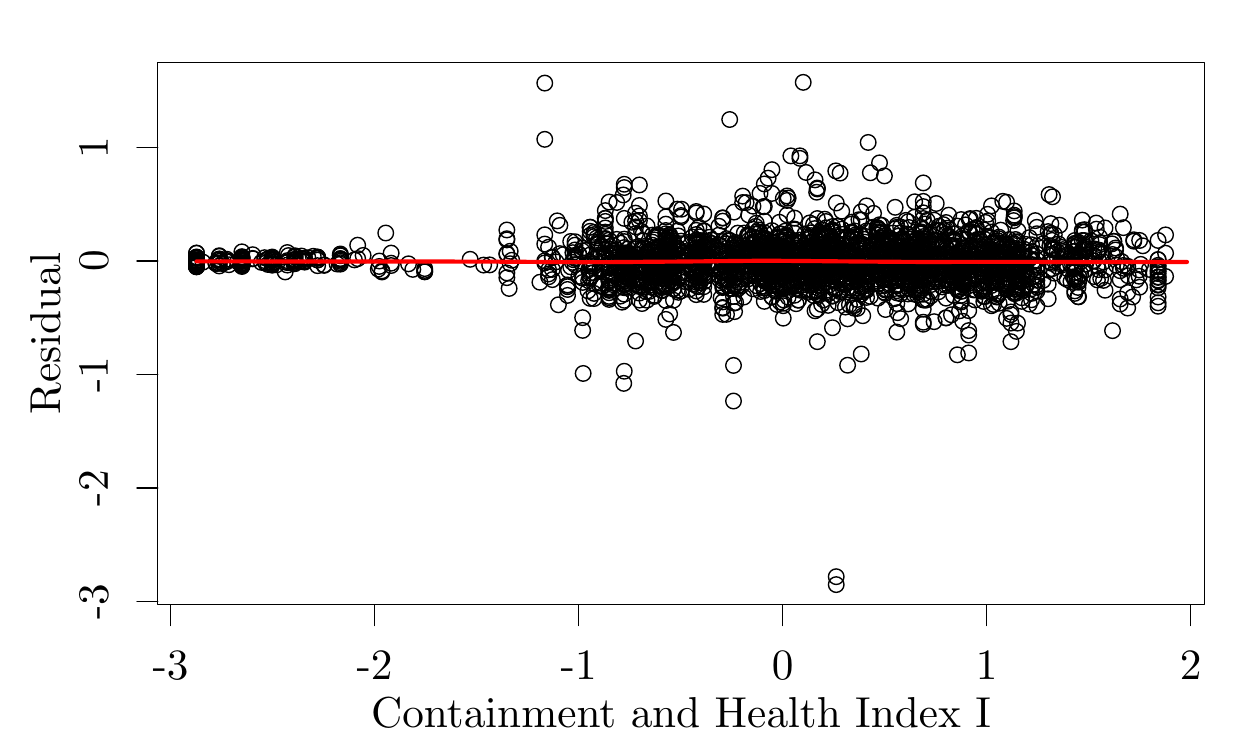}
\caption{Plots of the fitted values and each covariateagainst the residuals for the \COVID data. The {\color{red}{\textbf{red}}} line corresponds to the LOWESS smooth.}
\label{Fig:COVID_residuals}
\end{figure}

\begin{figure}%[hbt!]
\centering
\includegraphics[width =0.49\linewidth]{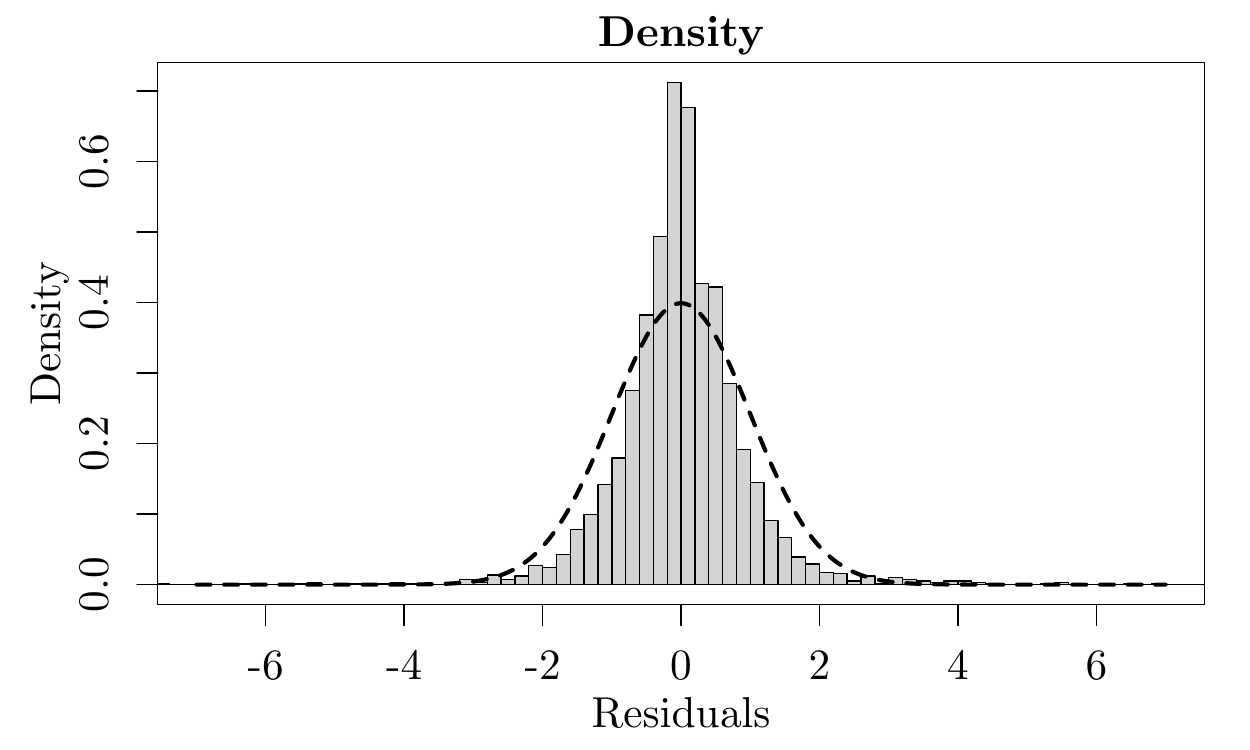} 
\includegraphics[width =0.49\linewidth]{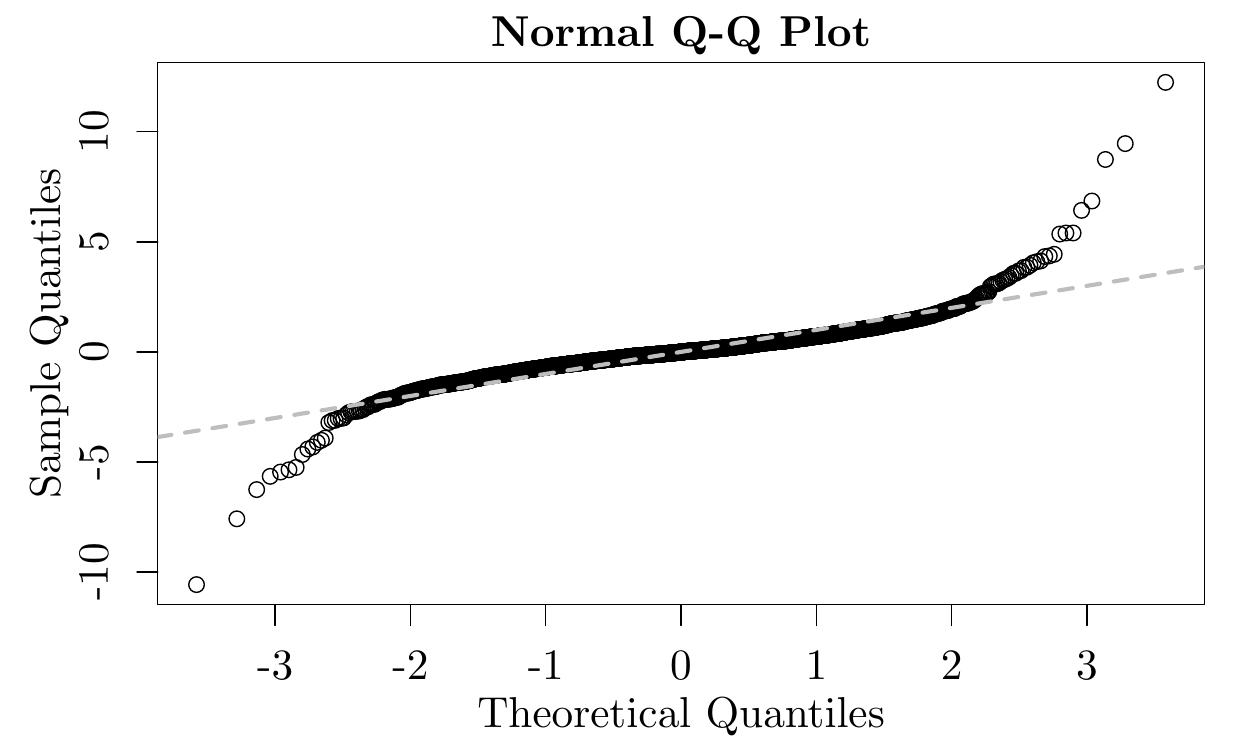}
\caption{\COVID data. \textbf{Left} Histogram of the standardised residuals compared with the standard Gaussian density. \textbf{Right} Q-Q Normal plot of the standardised residuals.}
\label{Fig:COVID_qq}
\end{figure}

The raw \COVID data exhibited strong serial correlation. To address this issue we added a first-order auto-regressive term. Figure \ref{Fig:COVID_ACF} plots the autocorrelation functions and partial autocorrelation functions for further lags after incorporating the AR1 term. These indicate that higher-order terms are unnecessary. After adding an AR1 term the interpretation of the errors (and their covariance) changes: they measure the infection rate relative to the covariates and to the infection rate of the previous week, i.e. they capture whether certain counties are growing faster/slower than expected (relative to the next week). So the model is investigating the growth rates, rather than absolute infection numbers.
\begin{figure}%[hbt!]
\includegraphics[width =0.49\linewidth]{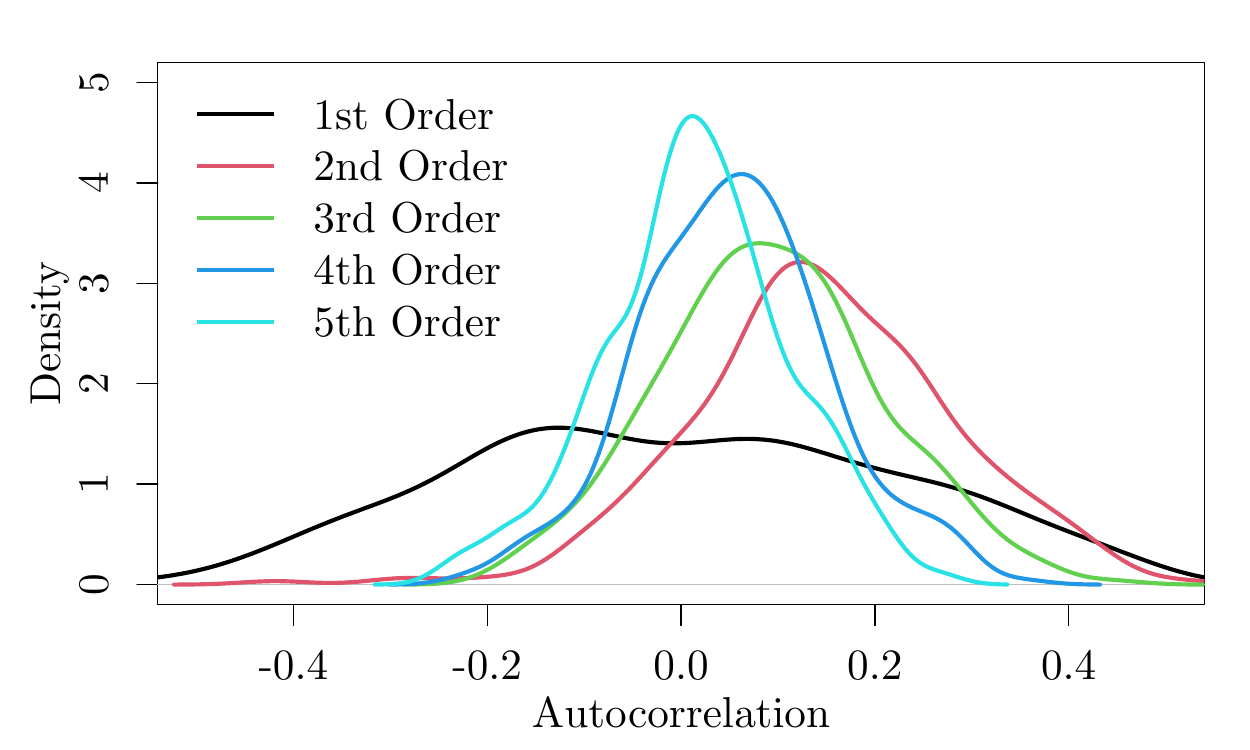} 
\includegraphics[width =0.49\linewidth]{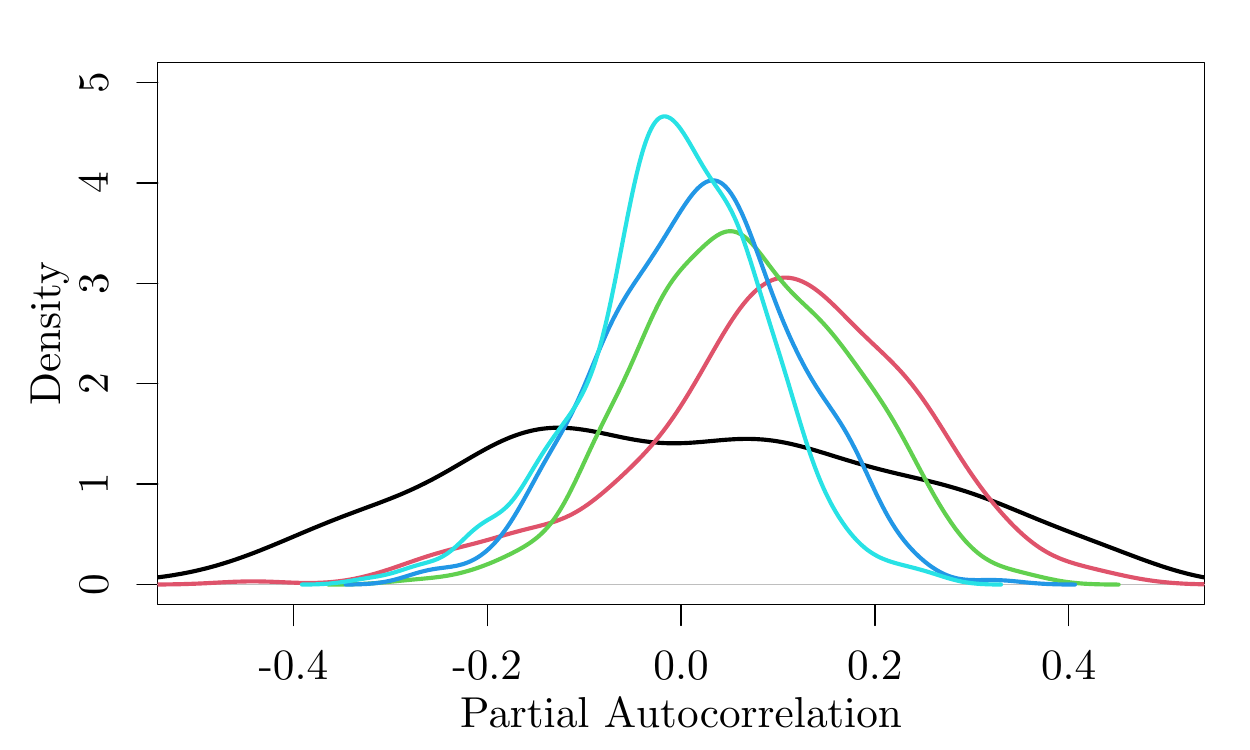}
\caption{Residual autocorrelation functions and partial autocorrelation functions after incorporating the AR1 term for the \COVID data.}
\label{Fig:COVID_ACF}
\end{figure}

\subsection{The network predictors}
\label{ssec:NetworkLinearity_COVID}

A further assumption of our proposed network \GLASSO models, as discussed in Section 2.1, is that there is a linear relation between $\log \mathbb{E}[\rho_{jk}^2 | A]$ and the network entries $a_{jk}^{(q)}$. 
To achieve linearity we defined our two network predictors as 
\begin{align}
    A_1 := 1/\log(Geodist), \quad A_2 := \log(Facebook), \quad A_3 := \log(1+Flights).\nonumber
\end{align}
Figure \ref{fig:lincheck} illustrates that after such transformations, the assumption of linearity is reasonably satisfied.
%displays such graphical checks for the \COVID and stock market data and the two networks considered in each application.
%Although the relation is not perfectly linear, the approximation seems good enough for our practical purposes.

%\david{Li, can you please polish Fig \ref{fig:lincheck}? First, only use equispaced bins, forget about quantile bins. Second, do not show blue lines, just the dots. Third, the figure should show only 4 panels: two for COVID on top and two for stock market on the bottom. Save each of these panels in a separate pdf file, and add them to the latex tabular command that I used (replace the A1-equi-obs.pdf by the corresponding file name)}

\begin{figure}%[hbt!]
\centering
%\begin{tabular}{cc}
%COVID, Facebook network & COVID, Geographical network \\
%\includegraphics[width =0.5\linewidth]{plot/A1-equi-obs.pdf} &
%\includegraphics[width =0.5\linewidth]{plot/A1-equi-obs.pdf} \\
\includegraphics[width =0.49\linewidth]{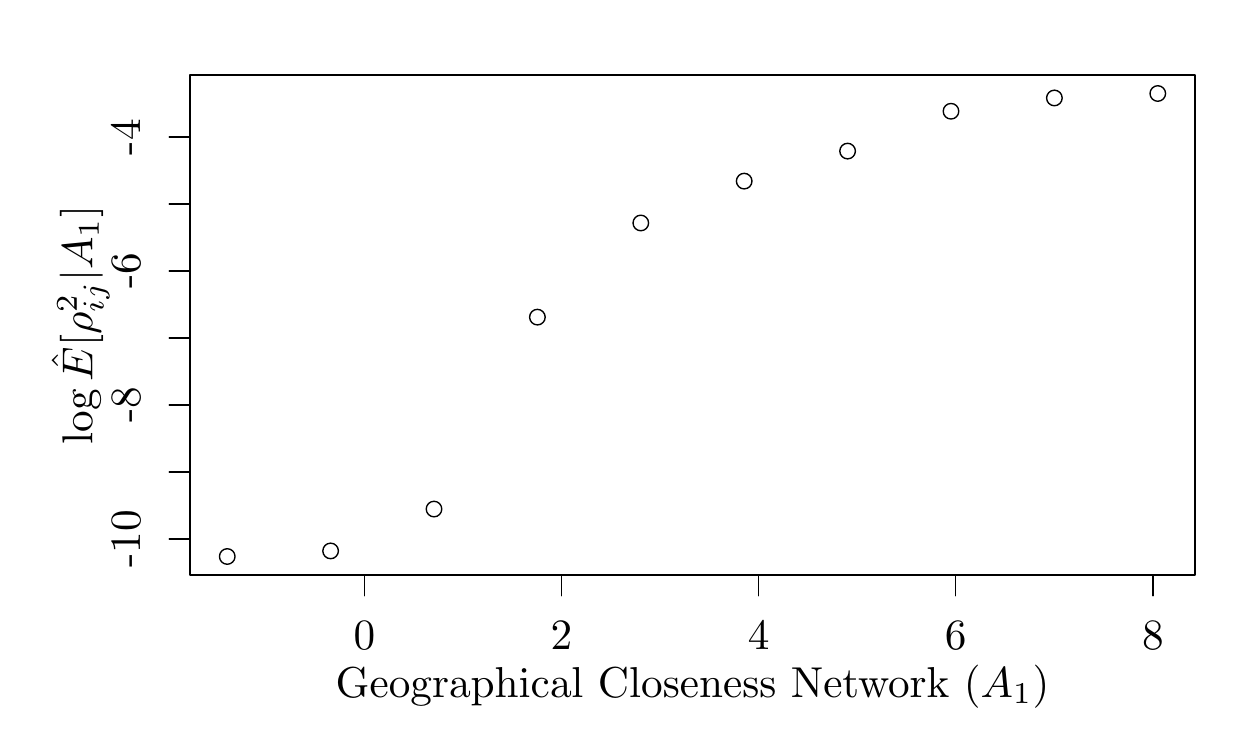}
\includegraphics[width =0.49\linewidth]{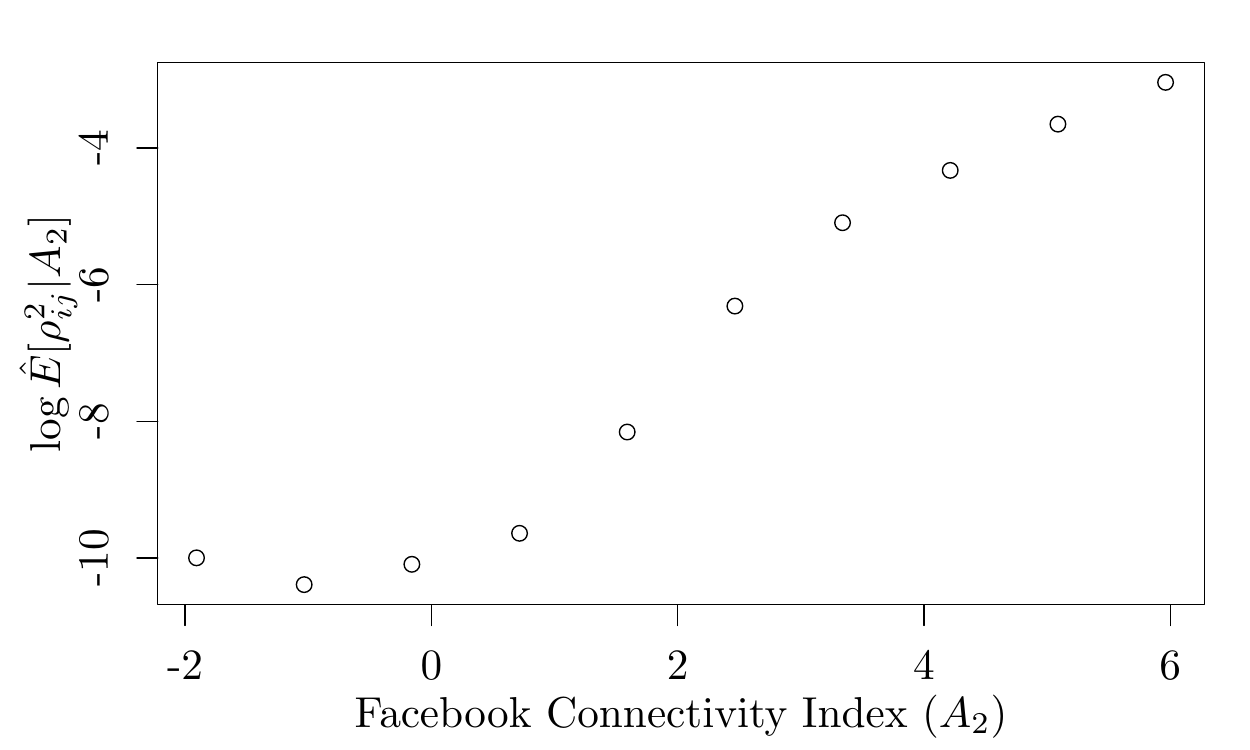}
\includegraphics[width =0.49\linewidth]{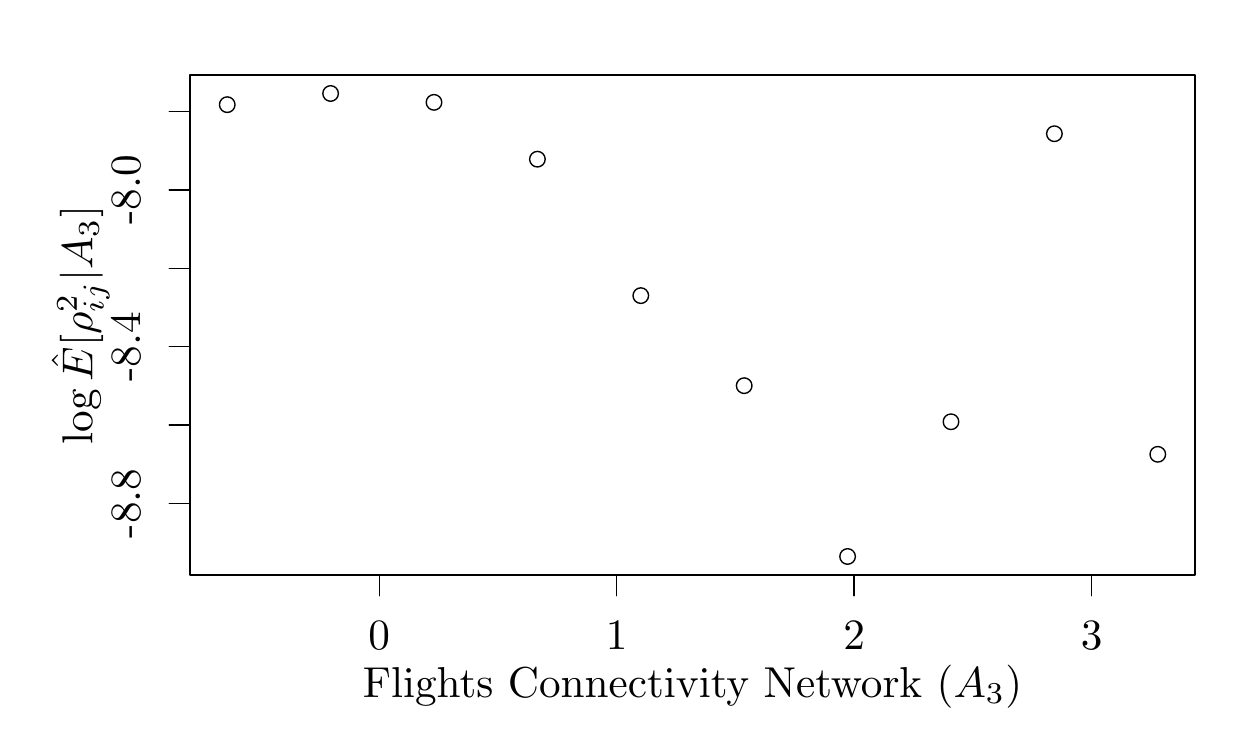}
%\includegraphics[width =0.49\linewidth]{plot/COVID_new/GLASSO_vs_A1_2_0.5_tikz-1.pdf}
%\includegraphics[width =0.49\linewidth]{plot/COVID_new/GLASSO_vs_A2_2_0.5_tikz-1.pdf}
%Stock market, Economic network & Stock market, Policy network \\
%\includegraphics[width =0.5\linewidth]{plot/A1-equi-obs.pdf}
%\includegraphics[width =0.5\linewidth]{plot/A1-equi-obs.pdf} \\
%\end{tabular}
\caption{Assessing the linear relation between $\log \mathbb{E}[\hat{\rho}_{jk}^2| A]$ and the network matrices, where $\hat{\rho}_{jk}$ is the \GLASSO estimate. The points represent the log-mean values of $\hat{\rho}_{jk}^2$ within 10 equispaced bins defined for each network.
}
\label{fig:lincheck}
\end{figure}

\subsection{Supplementary figures}

The top of Figure \ref{fig:glasso_covid_flight} the flight connectivity network against partial correlations estimated by \GLASSO. It appears that as the flight connectivity goes up, the variance in the partial correlations decreases slightly. However the dependence between the network and the partial correlations is much smaller than was observed between the geographical or Facebook networks and the partial correlations  in Figure 1. The fitted spike-and-slab distributions in the bottom of Figure \ref{fig:glasso_covid_flight} further demonstrate this. 

\begin{figure}[!ht]
\begin{center}
\includegraphics[width =0.49\linewidth]{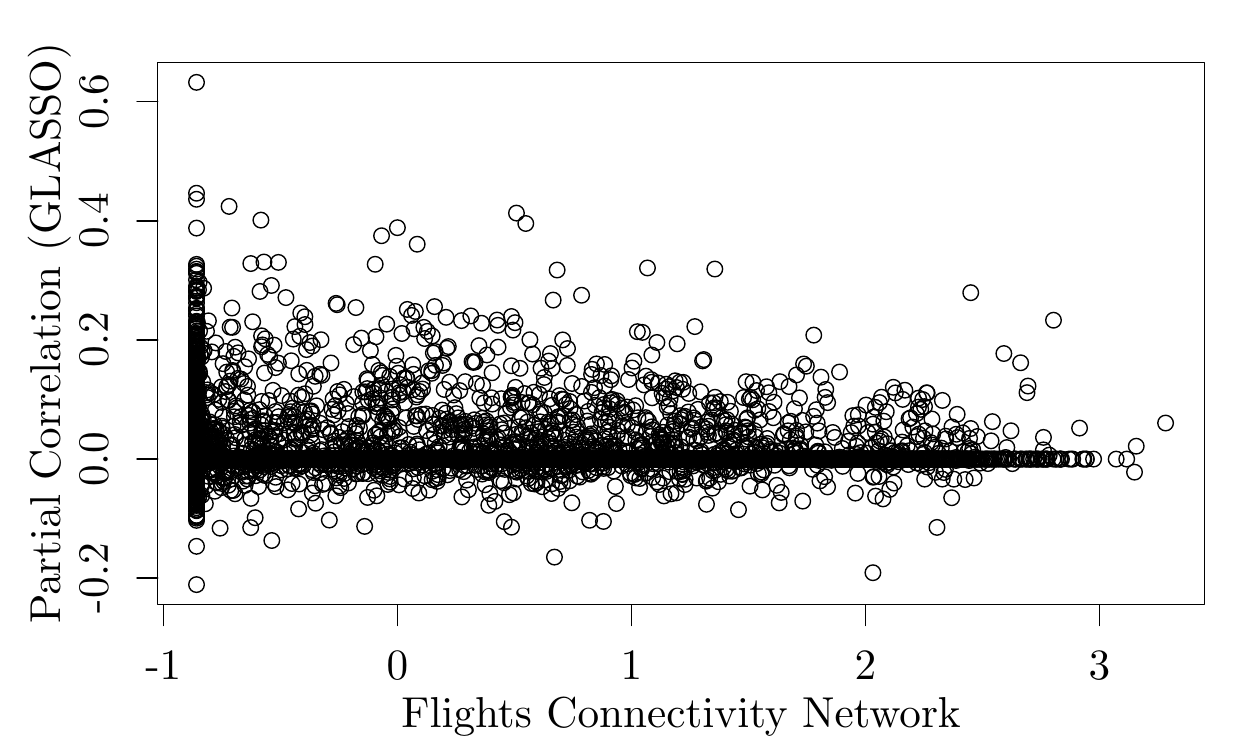}\\
\includegraphics[width =0.49\linewidth]{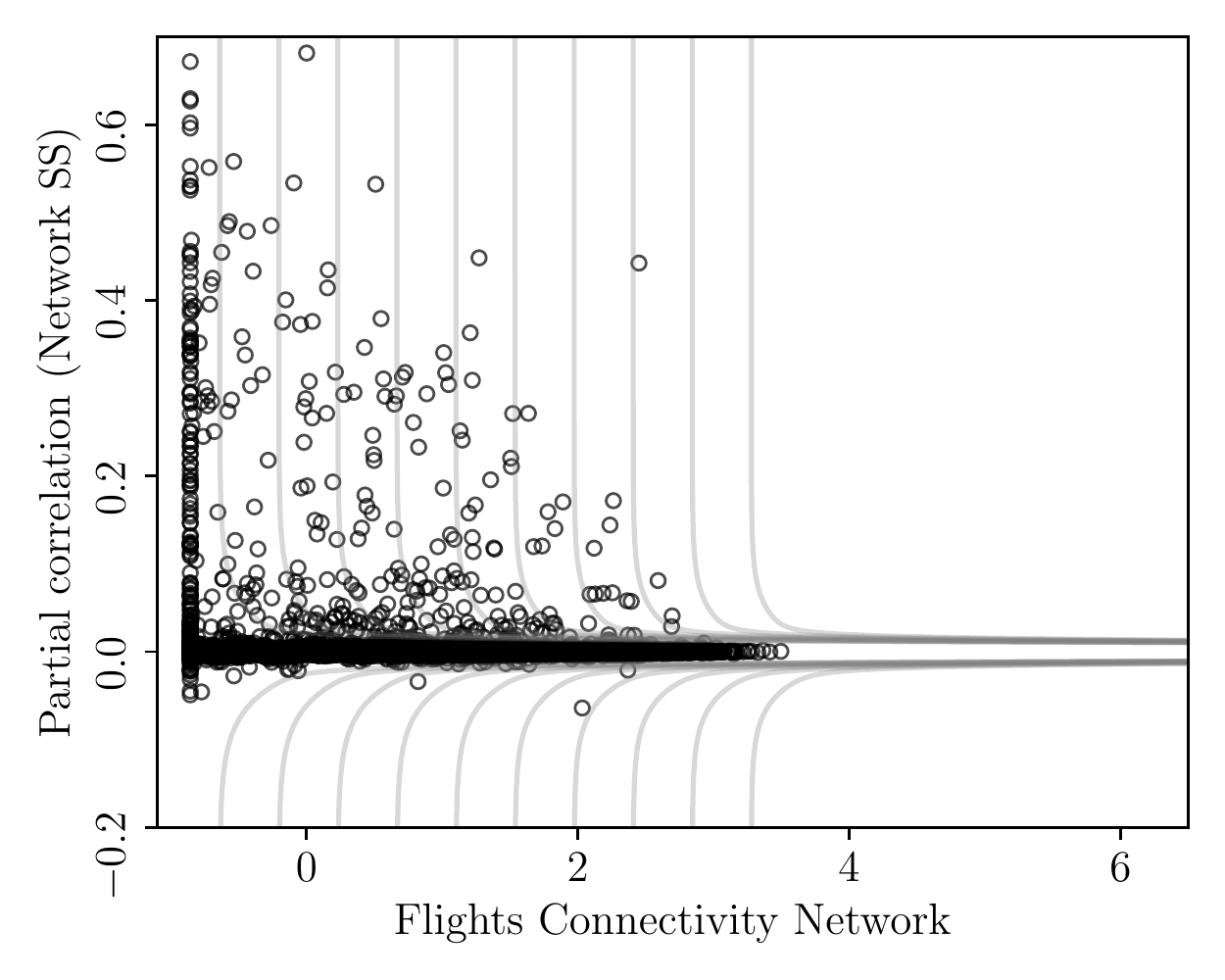}
%\includegraphics[width =0.49\linewidth]{plot/COVID_new/COVID_SS_prob_slab.pdf}
%This plot went in the appendix 
%\includegraphics[width =0.5\linewidth]{plot/COVID_new/partial corr vs facebook.pdf} 
%\includegraphics[width =0.5\linewidth]{plot/COVID_new/partial corr vs geodist.pdf}
%\includegraphics[width =0.49\linewidth]{plot/COVID_new/GLASSO_vs_A1_1_0.5_tikz-1.pdf}
%\includegraphics[width =0.49\linewidth]{plot/COVID_new/GLASSO_vs_A2_1_0.5_tikz-1.pdf}
\caption{Residual partial correlations in \COVID infections (adjusted for covariates) across counties vs flight connectivity network defined as $\log(1 + Flight)$. Top panel: partial correlations estimated with graphical LASSO, with penalization parameter set via \BIC. % \EBIC $\gamma_{\EBIC} = 0.5$. 
Bottom panel: fitted spike-and-slab distributions and fitted partial correlations estimated with network graphical spike-and-slab LASSO. %Bottom panel: prior slab probability as a function of both networks.
}
\label{fig:glasso_covid_flight}
\end{center}
\end{figure}

Table \ref{Tab:COVID_edge_network_SS} summarises the estimated graphical model under the network spike-and-slab model using a posterior slab probability threshold of $> 0.5$ and $> 0.95$. The number of edges estimated under both the 0.5 and 0.95 slab probability threshold is considerably smaller than the number of edges estimated under the network \GLASSO models. Under the 0.95 slab probability threshold, the estimated number of edges is more conservative. 

\begin{table}[]
\centering
\caption{\COVID data: Edge counts of the network spike-and-slab model when declaring an edge for posterior slab probability $> 0.5$ and $> 0.95$}
\label{Tab:COVID_edge_network_SS}
\begin{tabular}{ccccc}
\toprule{}     &     Edges ($> 0.5$) & Non-Edges ($> 0.5$)  &            Edges ($> 0.95$) & Non-Edges ($> 0.95$)   \\
\midrule Network SS  &             249 &                         54697 &                           102 & 54844 \\
\bottomrule
\end{tabular}
\end{table}

% \begin{table}[]
% \centering
% \caption{\COVID data: Edge counts of the network spike-and-slab model when declaring an edge for posterior slab probability $> 0.5$ and $> 0.95$
% \label{Tab:COVID_edge_network_SS}
% \begin{tabular}{ccccc}
% \toprule{}     &     Edges ($> 0.5$) & Non-Edges ($> 0.5$)  &            Edges ($> 0.95$) & Non-Edges ($> 0.95$)   \\
% \midrule Network SS  &             280 &                         4571 &                           68 & 4783 \\
% \bottomrule
% \end{tabular}
% \end{table}

Figure \ref{fig:ss_covid_probslab} shows how the estimated network hyperparameters of Table 3 affect the location of the slab and the probability of being in the slab marginally for each network when fixing the other two networks to their means. We see that while as both the geographical closeness and Facebook networks increase the location of the slab and the probability of being in the slab increases, the Facebook network has the larger effect. 

%This demonstrates the benefits of the more complex spike-and-slab framework, while the network \GLASSO estimated negative coefficients for both networks indicating positive correlation between the networks and the size of the partial correlations. The spike-and-slab disentangles this effect, showing that while the mean of those partial correlations that are non-zero increases with the geographical network, the probability of being non-zero does not. \jack{Move this text to the Stock market example }

\begin{figure}%[hbt!]
\begin{center}
\includegraphics[width =0.49\linewidth]{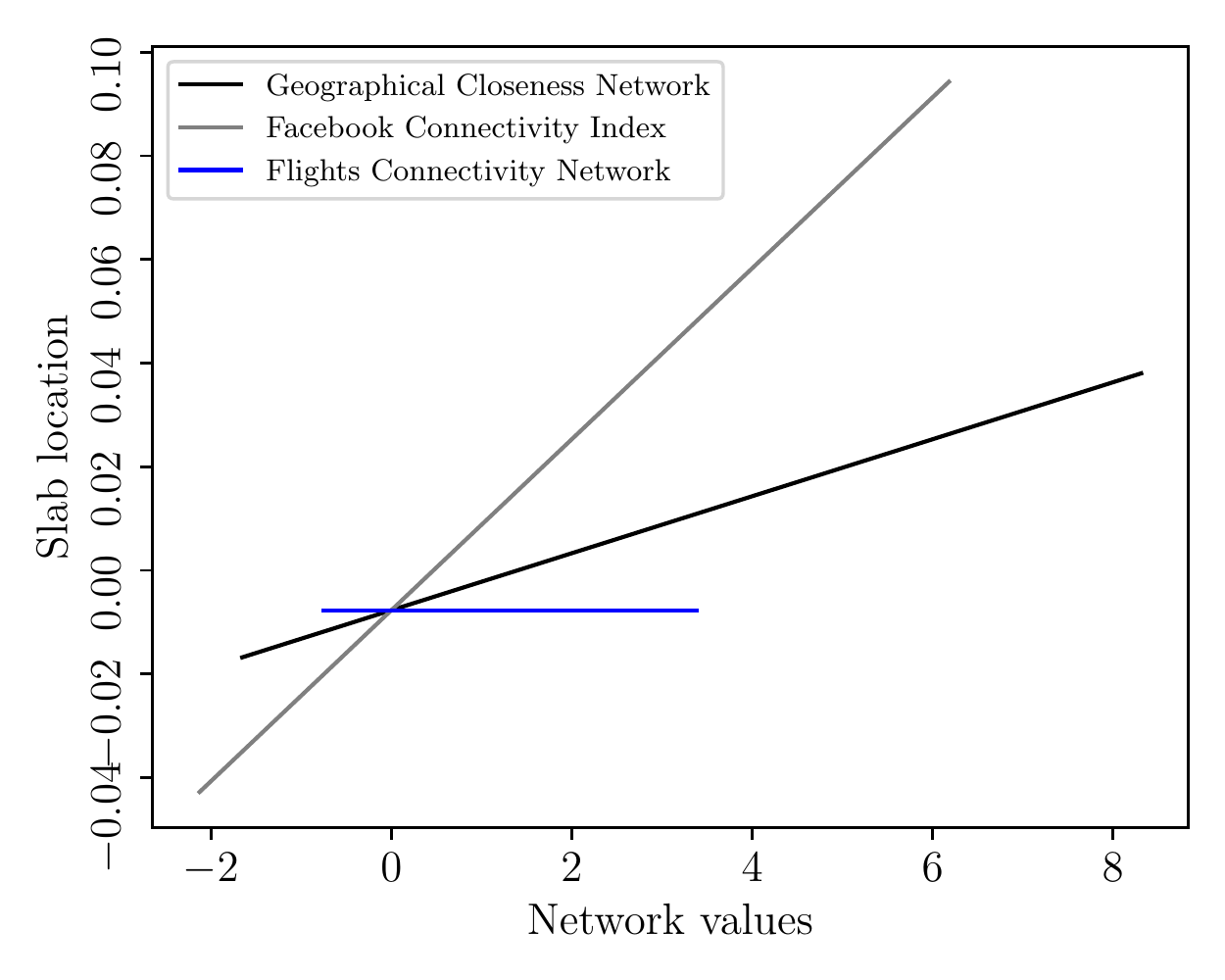}
\includegraphics[width =0.49\linewidth]{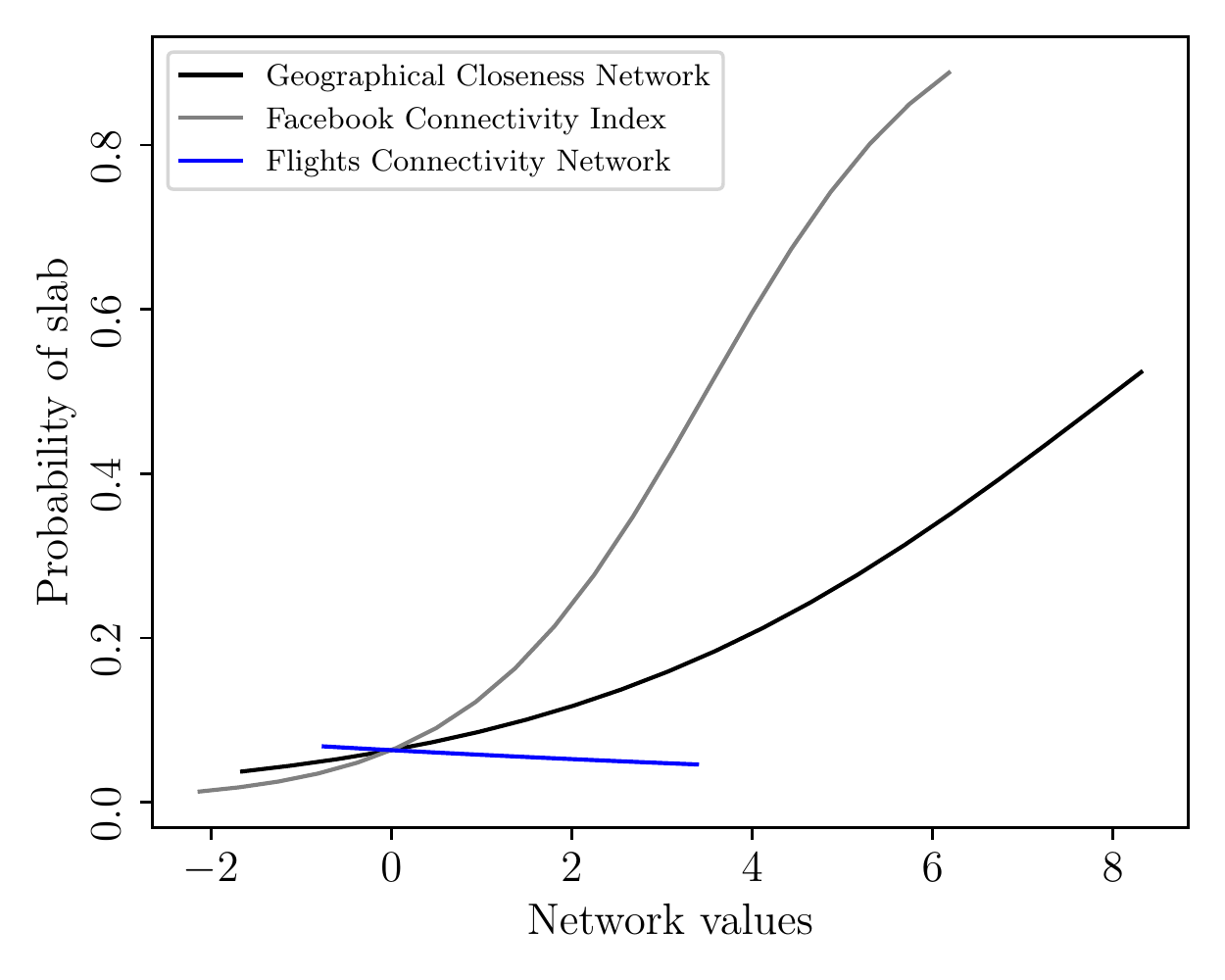}
\caption{\COVID data: Slab location (\textbf{left}) and slab probability (\textbf{right}) as a function of the three networks estimated by empirical Bayes. %We consider the slab here as the location of the partial correlations, the negation of the $\rho$ model parameters, which is why a negative value for $\eta_{0}$ (Table \ref{Tab:SS_COVID_CIs}) result in a positive effect.
}
\label{fig:ss_covid_probslab}
\end{center}
\end{figure}

\subsubsection{U.S. map plots}{\label{Sec:MapPlots}}

Figure \ref{fig:no_network_vs_both_networks} visualises the network given by non-zero elements of the \GLASSO estimated $\Theta$ with no network information (top) and the network \GLASSO estimate of $\Theta$ obtained when using both $A_1$ and $A_2$ (bottom), the model achieving the smallest \BIC, on top of a U.S. map. 
The network \GLASSO estimates a much sparser network, but we see there are still edges present between counties that are geographically close as well as those that are farther away. %To investigate the marginal effects of each network further, Figure \ref{fig:unique_edges_covid} plots only those edges that were identified uniquely under the Network \GLASSO model using either of the network models. This shows that the $A_1$ encourages counties that are close together to be connected, while $A_2$ allows for counties that are further away physically, but more connected via Facebook to be present 

%Figure \ref{fig:unique_edges_covid} top shows edges that were selected by our network \GLASSO but not by standard \GLASSO, when adding only the geographical network $A_1$. These edges correspond to counties that are close to each other.

%Figure \ref{fig:unique_edges_covid} top shows an analogous plot when only using the Facebook network $A_2$. Interestingly, here some of the edges are between geographically-distant counties in the west, north-east and south-east.

\begin{figure}%[hbt!]
\centering
\begin{subfigure}[b]{\linewidth}
    \includegraphics[trim= {0.0cm 0.5cm 0.0cm 0.5cm}, clip, width = 1\linewidth]{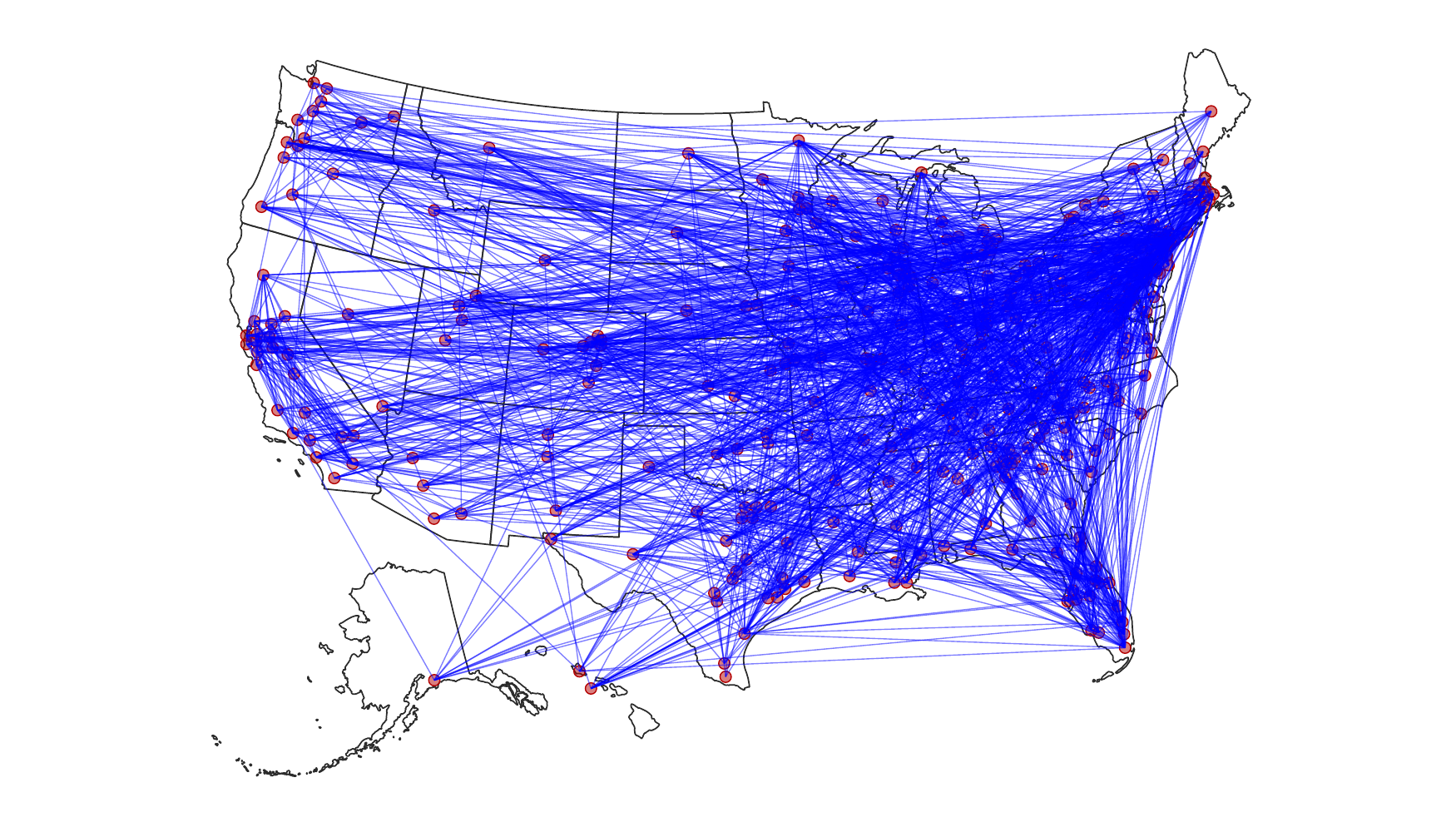}
    \caption{Edges identified in \GLASSO{} with no network}
\end{subfigure}
\vfill
\begin{subfigure}[b]{\linewidth}
    \includegraphics[trim= {0.0cm 0.5cm 0.0cm 0.5cm}, clip, width =\linewidth]{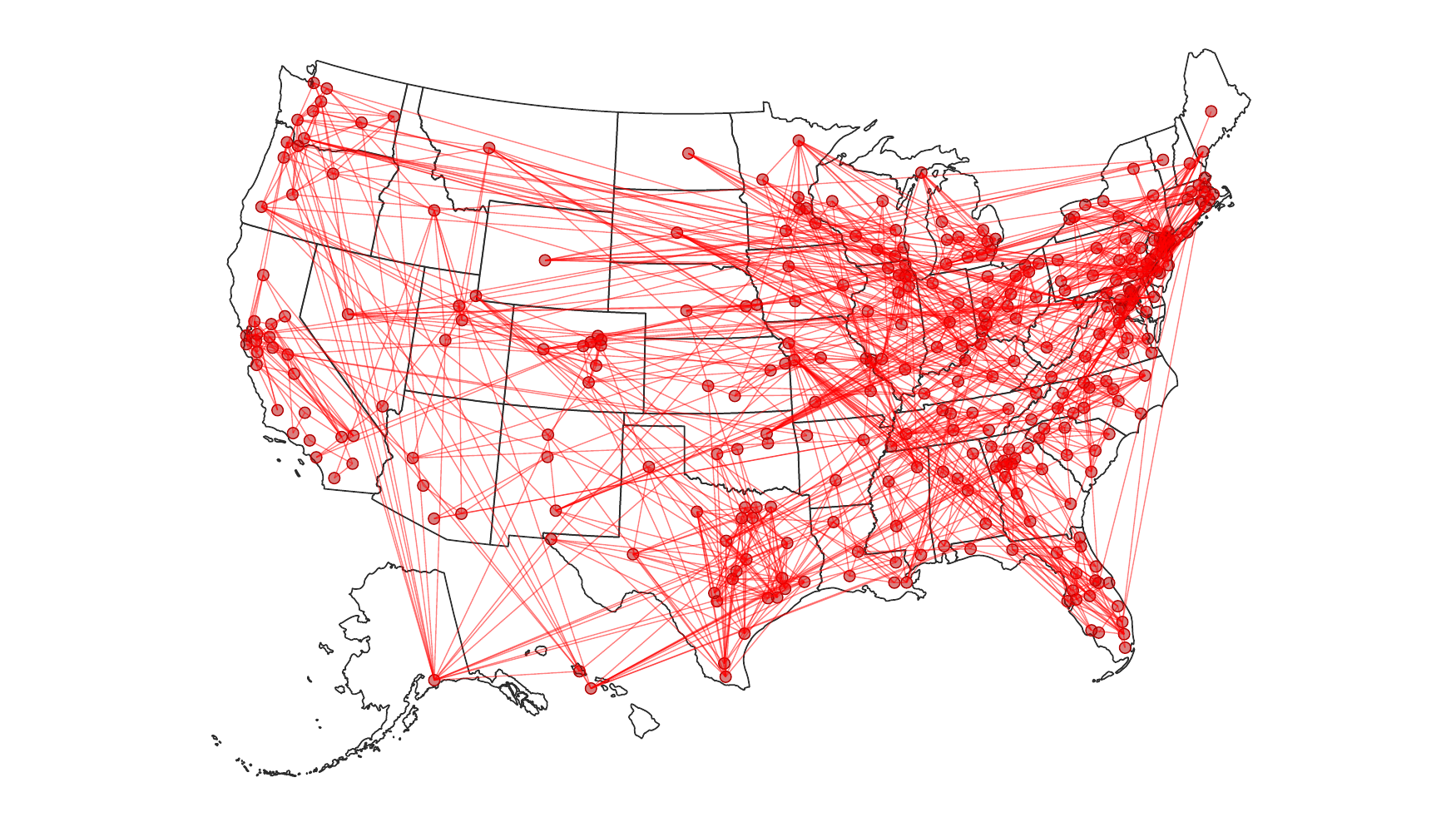}
    \caption{Edges identified in network \GLASSO{} with networks $A_1$ and $A_2$, the model achieveing the smallest \BIC}
\end{subfigure}
\caption{Edges identified by \GLASSO{} and network \GLASSO with the geographical closeness and Facebook networks. 
%\BIC - $\gamma_{\EBIC} = 0$}. 
%The coordinates of the county Honolulu (Hawaii) have been adjusted from (-164.44361, 23.87280) to (-158.2019740, 21.4613654) for presentation.
}
%trim={<left> <lower> <right> <upper>}
\label{fig:no_network_vs_both_networks}
\end{figure}

\subsection{Results using the \EBIC}{\label{Sec:EBIC_COVID}}

Similarly to the simulations, we also investigate the sensitivity of our \COVID data results by considering selecting hyperparameters using the \EBIC with $\gamma_{\EBIC} = 0.5$. Table \ref{tab:results_covid_0.5} presents these results. From the number of edges, we can see that using the \EBIC estimates sparser networks than under the \BIC, but the out-of-sample test set estimate suggests these estimates may be too sparse. Importantly, we see that the improvement of the network \GLASSO methods over standard \GLASSO is still apparent when using the \EBIC selection criteria.

%% OLD result p = 99
%\begin{table}[ht]
%\centering
%\caption{Comparison of four models for the \COVID data when using the \EBIC ($\gamma_{\EBIC} = 0.5$) to learn the network hyperparameters. $A_1$ and $A_2$: networks defined by $ 1/\log(Geodistance)$ and $\log(Facebook)$. \EBIC values account for the extra hyperparameters learned in the network \GLASSO model. Test set score obtained using 10-fold cross-validation}
%\caption{Four models for the \COVID data when using the \EBIC ($\gamma_{\EBIC} = 0.5$) to learn the network hyperparameters. $A_1$ and $A_2$: networks defined by $ 1/\log(Geodist)$ and $\log(Facebook)$. \EBIC values account for the extra hyper-parameters in the network \GLASSO models. 10-fold: 10-fold cross-validated log-likelihood
%\begin{tabular}{ccccccc}
%  \hline
%  Method & \EBIC & $\hat{\beta}_0$ & $\hat{\beta}_1$ & $\hat{\beta}_2$ & Edges & 10-fold\\ 
%   \hline
%   \GLASSO & 9556.655 & -0.789 &  &  &  175 & 59.67  \\ 
%   Network \GLASSO - $A_1$ & 7252.871 & 1.276 & -1.132 &  &  132 & 246.07 \\ 
%   Network \GLASSO - $A_2$ & \textbf{6795.846} & 3.605 &  & -1.947 &   91 & 262.9537 \\ 
%   Network \GLASSO - $A_1$ \& $A_2$ & 6809.573 & 3.556 & 0.278 & -2.278 &   97 & \textbf{263.12} \\ 
%   \hline
%\end{tabular}
%\label{tab:results_covid_0.5}
%\end{table}

\begin{table}[ht]
\centering
\caption{Eight models for the \COVID data when using the \EBIC ($\gamma_{\EBIC} = 0.5$) to learn the network hyperparameters. $A_1$, $A_2$ and $A_3$: networks defined by $ 1/\log(Geodist)$, $\log(Facebook)$ and $\log(1+Flight)$. \EBIC values account for the extra hyper-parameters in the network \GLASSO models. 10-fold: 10-fold cross-validated log-likelihood}
\begin{tabular}{cccccccc}
  \hline
  Method & \EBIC & $\hat{\beta}_0$ & $\hat{\beta}_1$ & $\hat{\beta}_2$ & $\hat{\beta}_3$ & Edges & 10-fold\\ 
   \hline
 \GLASSO & 32204.000 & -0.062 &  &  &  &    0 & 24.317 \\ 
 Network \GLASSO - $A_1$ & 27150.357 & 1.622 & -1.120 &  &  &  766 & 81.846 \\ 
 Network \GLASSO - $A_2$ & 24461.011 & 2.903 &  & -1.147 &  &  617 & 113.533 \\ 
 Network \GLASSO - $A_3$ & 32220.185 & 0.959 &  &  & -0.165 &  0 & 24.317 \\ 
 Network \GLASSO - $A_1$ \& $A_2$ & 24303.227 & 5.200 & -1.038 & -1.162 &  &  730 & 112.771 \\ 
 Network \GLASSO - $A_1$ \& $A_3$ & 26453.480 & 2.407 & -1.414 &  & 1.005 &  766 & 90.710 \\ 
 Network \GLASSO - $A_2$ \& $A_3$ & \textbf{23931.443} & 4.063 &  & -1.457 & -0.255 &  589 & \textbf{114.372} \\ 
 Network \GLASSO - $A_1$, $A_2$ \& $A_3$ & 24927.80 & 2.845 & -0.274 & -1.159 & 0.246 & 796 & 113.764 \\
   \hline
\end{tabular}
\label{tab:results_covid_0.5}
\end{table}

%\newpage

\section{Stock market data preparation}
\label{sec:stock_dataprocessing}

This section provides additional details for the analysis of the stock market excess returns data. 

\subsection{Data sources}

To undertake our analysis, we collected and combined the following datasets.

1. Stock price data\\
We extracted the daily closing stock price for $p = 366$ firms satisfying the following criteria: closing stock prices adjusted for stock splits and dividends were available in the COMPUSTAT  database for every trading day between 2 January 2019 to 31 December 2019 (leaving $n = 252$ time points), the stocks were associated to a member of the S\&P500 at the end of 2019, and we could retrieve their 10-K filings for at least one of the years 2014-2019. The data was downloaded from the Center for Research in Security Prices (CRSP) database accessed via Wharton Research Data Services (WRDS).\\

%2. CIK-TIC crosswalk data\\
%The stock price data use the TIC identifier while the risk data from which we derive our networks uses the CIK identifier.  We manually created a CIK-TIC crosswalk from the Compustat database accessed via WRDS. \\

2. S\&P 500 firms \\
The list of S\&P 500 firms was downloaded from \url{https://web.archive.org/web/20190912150512/https://en.wikipedia.org/wiki/List_of_S%26P_500_companies} which corresponds to the wikipedia page listing the S\&P500 retrieved on 12/09/2019, its last archived data in 2019.\\

3. Fama/French Three-Factor Model\\
We constructed excess returns using the Fama-French three-factor model \citep{FAMA19933}.  The three factors are the 1) overall market return, 2) a measure of firm size, and 3) a measure of book-to-market ratio.  The daily Fama/French factors were downloaded from \url{https://mba.tuck.dartmouth.edu/pages/faculty/ken.french/data_library.html}.  For each stock, we regress 2019 daily returns on the three factors (plus a constant) and extract the residual as the excess return.\\

4. Risk measures\\
Our network data measures the similarity of two companies' risk exposures stratified into Economic and Policy risks. The 10-K risk exposure data counts for each risk category the number of sentences within a company's 10-K filings that contained any member of a dictionary associated with that risk category.  We manually construct these using the dictionary terms listed in \citet{bakerPolicyNewsStock2019}.  From this data, we can construct a $p\times p$ network matrix for firms, where each entry $X_{ij}$ represents the degree of ``closeness" between firm $X_{i}$ and firm $X_{j}$.

%\subsection{Data processing}

%\begin{enumerate}
%    \item The price data from CRSP is arranged by TIC, a unique stock identifier, while the risk measures are arranged by CIK, a unique company identifier. Any two stocks associated with the same compnay had the same risk scores.
%    \item There are some negative values in the stock data. The negative signs are to ``indicate that it is a bid/ask average and not an actual closing price'' when the  ``closing price is not available'' \url{https://faq.library.princeton.edu/econ/faq/11159}. We, therefore, took the absolute values of the returns before the log return calculations. NO LONGER TRUE
%\end{enumerate}

\subsection{Model description}
Our final response variable is the log daily returns for $p = 366$ U.S. firms throughout 2019, resulting in $n = 251$ observations. We are, however, interested in the graphical model, $\mathcal{N}_p(0, \Theta^{-1})$, of the `excess returns', defined as the residuals of a linear model regressing the log-returns on the Fama-French factors.

%We download the firm-level stock returns at a daily frequency for 2019 from CRSP, and then extracted the frequency-specific factors from the Fama-French model. Over the period of interest, linear regressions are run on the Fama-French factors for selected 200 U.S. listed stocks. The residuals from these regressions are the excess returns. \GLASSO and \GOLAZO estimates were then implemented for the excess returns.

The `excess returns' for stock $j$ are therefore estimated, separately for each firm, using the following model
\begin{align}
r_{ij} - Rf_i  = b_{0j} + b_{1j} \times SMB_{i} + b_{2j} \times HML_{i} + b_{3j} \times(Rm-Rf)_{i} + \epsilon_{ij}
\end{align}
where
\begin{itemize}[itemindent=0pt]%,leftmargin=*]
    \item[(1)] $r_{ij}$ is the  log daily return of stock $j$ at time $i$ defined as $r_{ij} = \log p_{ij} - \log p_{i-1j}$, where $p_{ij}$ is the closing price of firm $j$ on day $i$
    \item[(2)] $Rf_i$ is the risk free rate at time $i$.
    \item[(3)] $SMB_{i}$ (Small Minus Big) is the average return on the three small portfolios minus the average return on the three big portfolios at time $i$. %\url{https://mba.tuck.dartmouth.edu/pages/faculty/ken.french/Data_Library/f-f_factors.html} %\jack{we either need more info or a reference here, ideally both} \li{Hi Jack, for covariates SMB, HML and (Rm-Rf), only one index per time point.}
    \item[(4)] $HML_{i}$ (High Minus Low) is the average return on the two value portfolios minus the average return on the two growth portfolios at time $i$. %\url{https://mba.tuck.dartmouth.edu/pages/faculty/ken.french/Data_Library/f-f_factors.html}%\jack{we either need more info or a reference here, ideally both}
    \item[(5)] $(Rm-Rf)_{i}$, the excess return on the market at time $i$, value-weighted return of all CRSP firms incorporated in the U.S. and listed on the NYSE, AMEX, or NASDAQ that have a CRSP share code of 10 or 11 at the beginning of $i$'s month, good shares and price data at the beginning of $i$'s month, and good return data for $i$ minus the one-month Treasury bill rate. %\url{https://mba.tuck.dartmouth.edu/pages/faculty/ken.french/Data_Library/f-f_factors.html} %\li{Yes, the $(Rm-Rf)_{i}$ data is daily. If you look at the data, you will see that the $R_m$ variable is daily data, while the $R_f$ variable is monthly data.}
    \item[(6)] Coefficients $b_{0j}$, $b_{1j}$, $b_{2j}$ and $b_{3j}$ are estimated for firm $j$ using ordinary least squares. 
\end{itemize}

\subsection{Checking model goodness-of-fit}{\label{ssec:gof_STOCK}}

Similarly to Section \ref{ssec:gof_COVID}, we produce diagnostic plots to confirm the validity of the linear-model and the Gaussianity and independence of its residuals. 

Figure \ref{Fig:STOCK_ACF} plots autocorrelation functions and partial autocorrelation functions, demonstrating that the observations can be considered independent and that there is no need to consider auto-regressive terms. Figure \ref{Fig:STOCK_residuals} plots the fitted values $\hat{y}_{ij}$ and each of the predictors against the residuals $\epsilon_{ij}$, demonstrating that the assumption that the covariates are linearly related to the response is satisfactory and that the residuals appear reasonably homoskedastic.

\begin{figure}%[hbt!]
\includegraphics[width =0.49\linewidth]{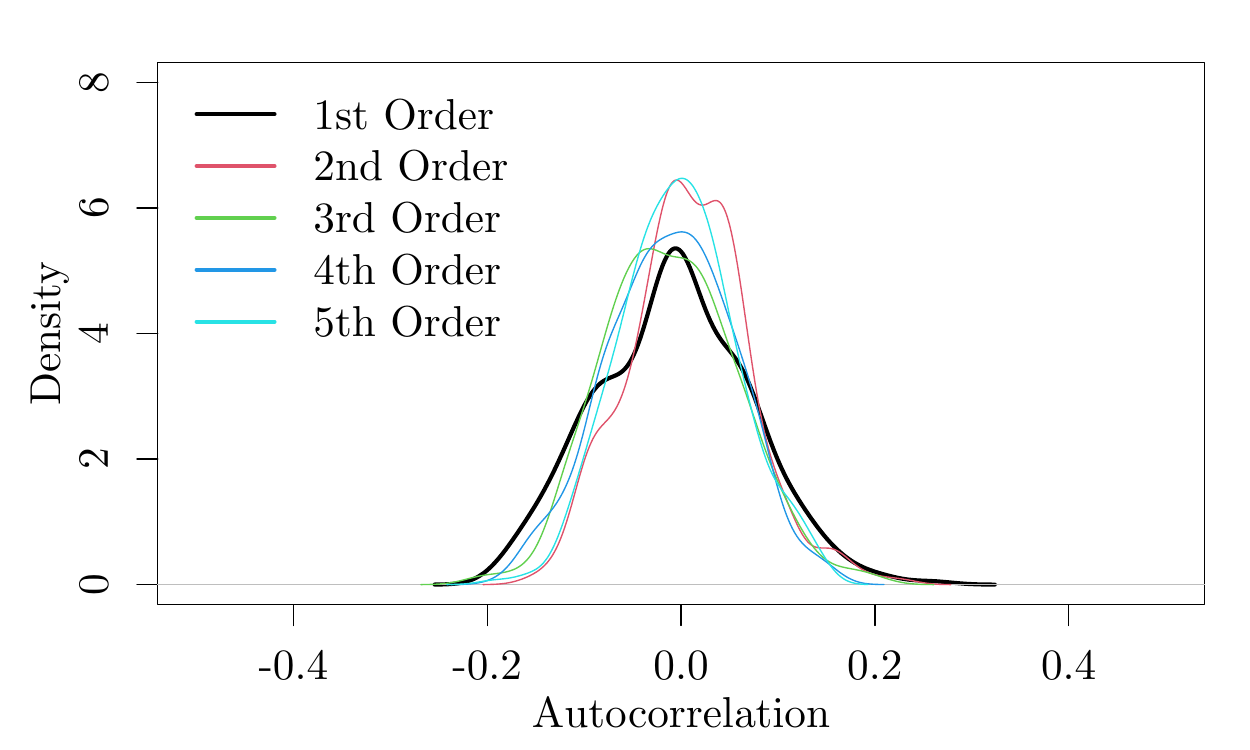} 
\includegraphics[width =0.49\linewidth]{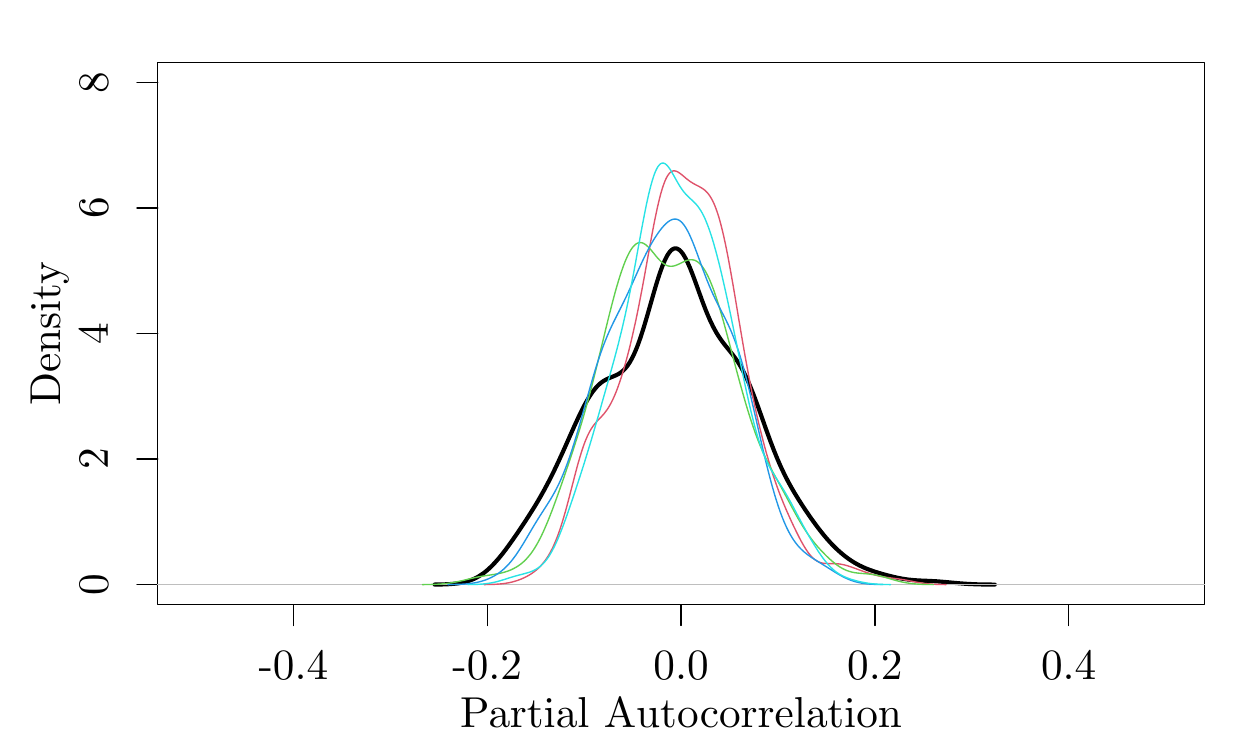}
\caption{Residual autocorrleation function and partial autocorrelation function for the stock market data.}
\label{Fig:STOCK_ACF}
\end{figure}

\begin{figure}%[hbt!]
\centering
\includegraphics[width =0.49\linewidth]{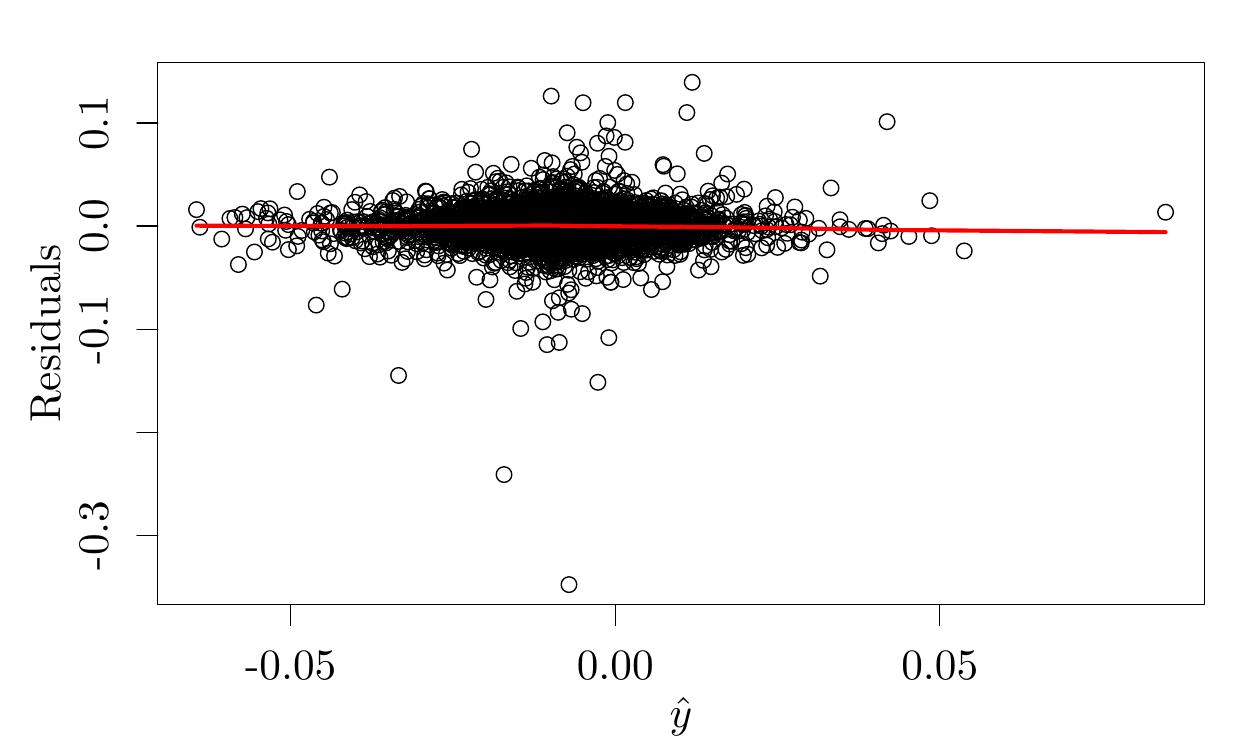} 
\includegraphics[width =0.49\linewidth]{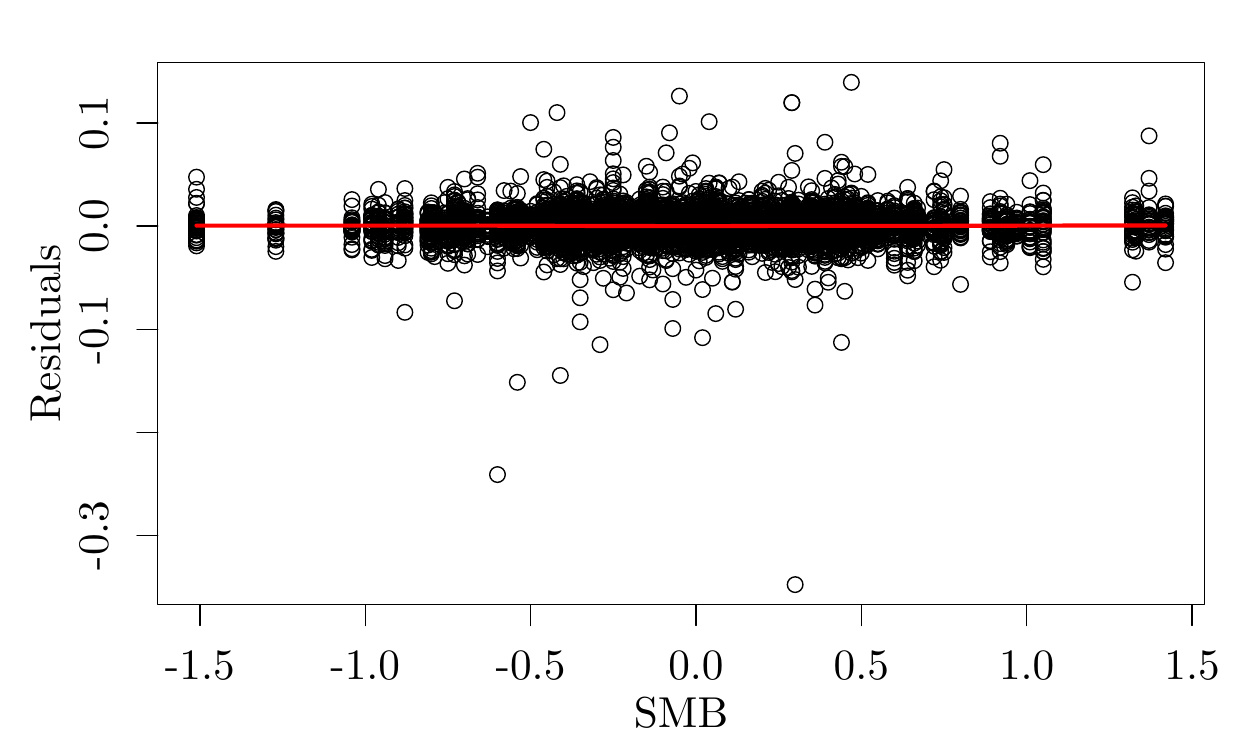}
\includegraphics[width =0.49\linewidth]{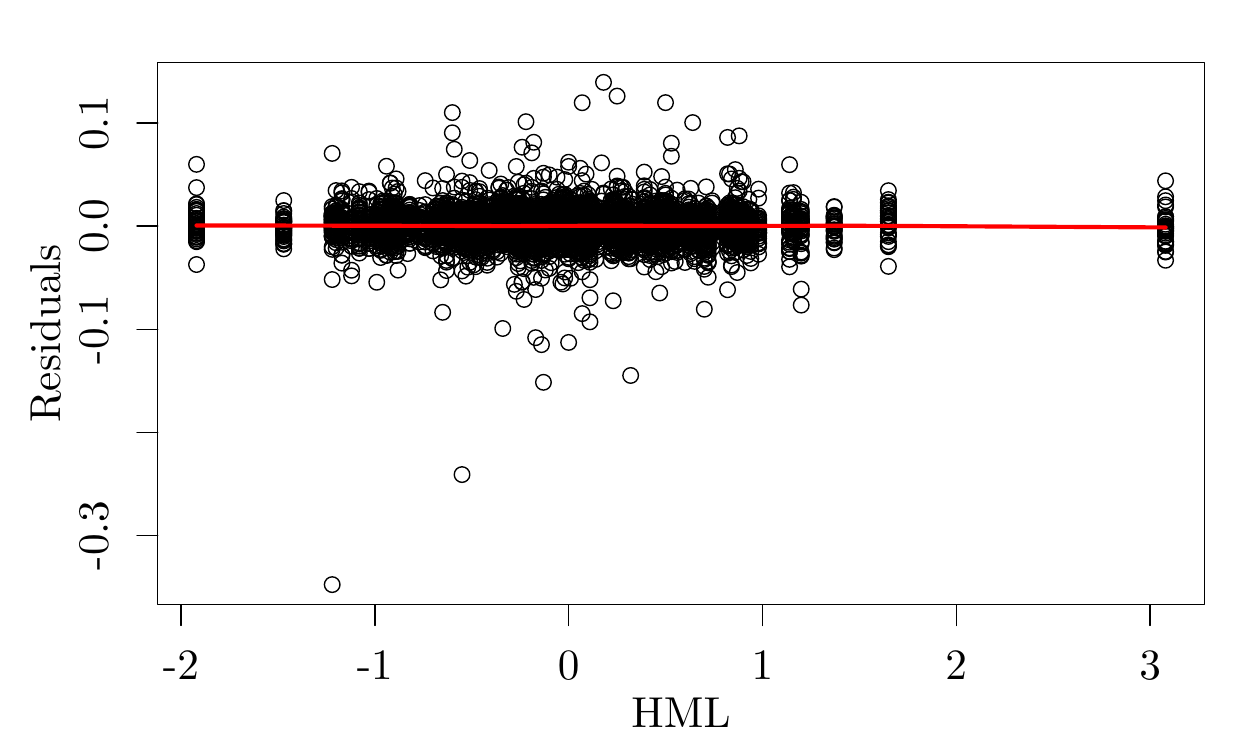}
\includegraphics[width =0.49\linewidth]{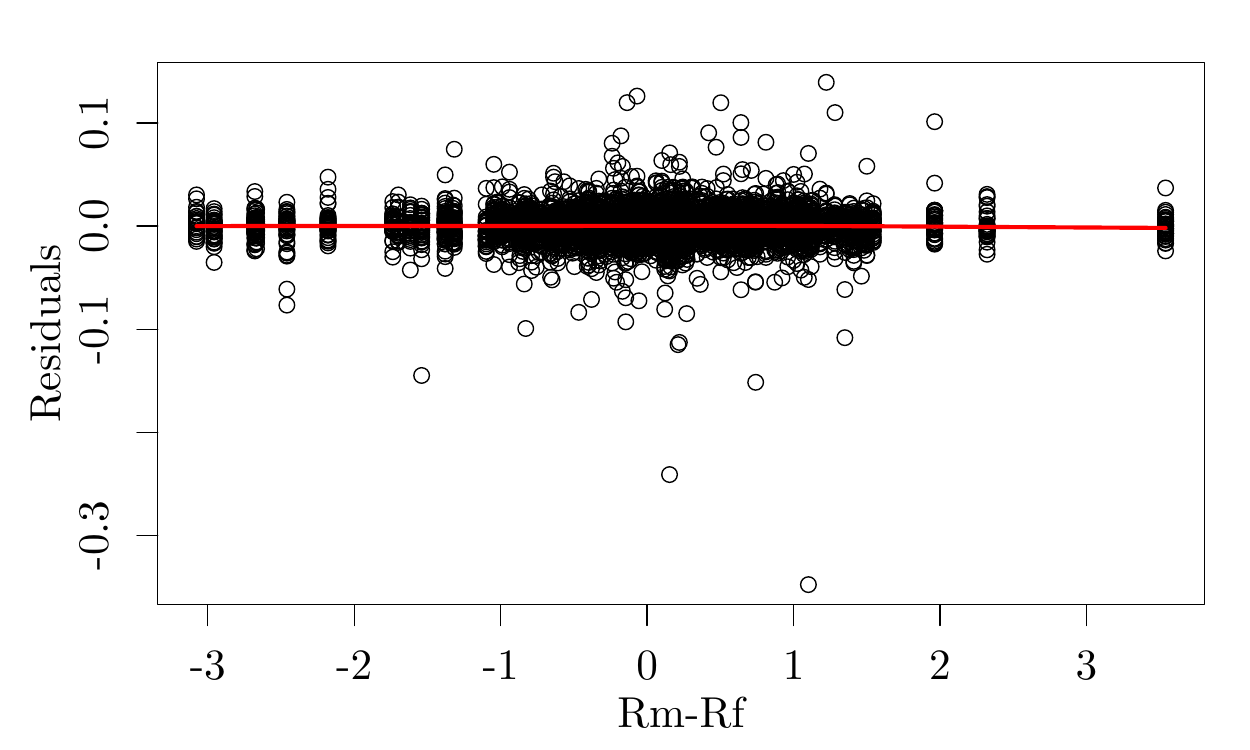}
\caption{Plots of the fitted values and each covariate against the residuals for the stock market data. The {\color{red}{\textbf{red}}} line corresponds to the LOWESS smooth.}
\label{Fig:STOCK_residuals}
\end{figure}

While the Gaussian assumption was tenable for the \COVID data, Figure \ref{Fig:STOCK_qq} shows that this is not the case for the stock market data. There is evidence of considerably heavier tails than Gaussianity. To address this issue we fit a non-paranormal model based on transforming the data into $f(\epsilon_i):=(f_1(\epsilon_{i1}),\ldots,f_p(\epsilon_{ip}))$, where $\hat{f}$ was estimated using the \textit{R} package \texttt{huge} \citep{zhao2012huge}. Figure \ref{Fig:STOCK_qq_transformed} shows a histogram and Q-Q-normal plot for $f(\epsilon_i)$, where the Gaussian assumption is more tenable.

\begin{figure}%[hbt!]
\centering
\includegraphics[width =0.49\linewidth]{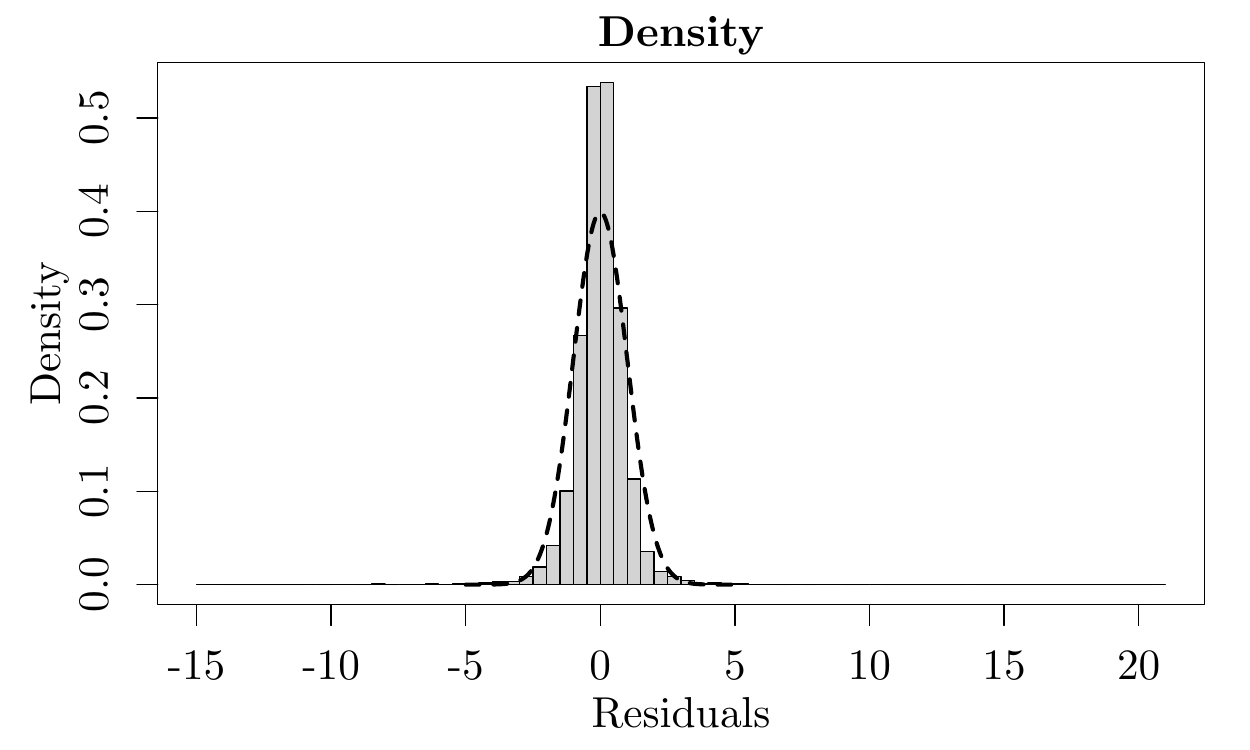} 
\includegraphics[width =0.49\linewidth]{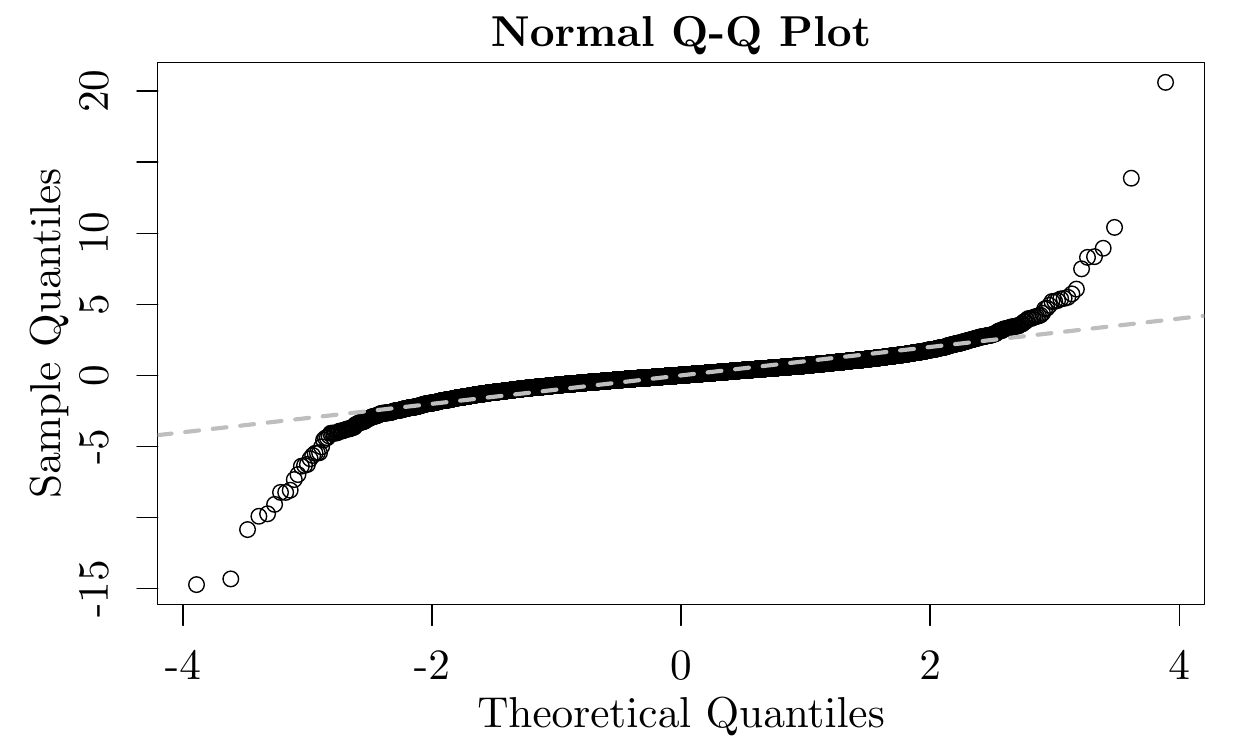}
\caption{Stock market data. \textbf{Left} Histogram of the standardised residuals compared with the standard Gaussian density. \textbf{Right} Q-Q Normal plot of the standardised residuals.}
\label{Fig:STOCK_qq}
\end{figure}

\begin{figure}%[hbt!]
\centering
\includegraphics[width =0.49\linewidth]{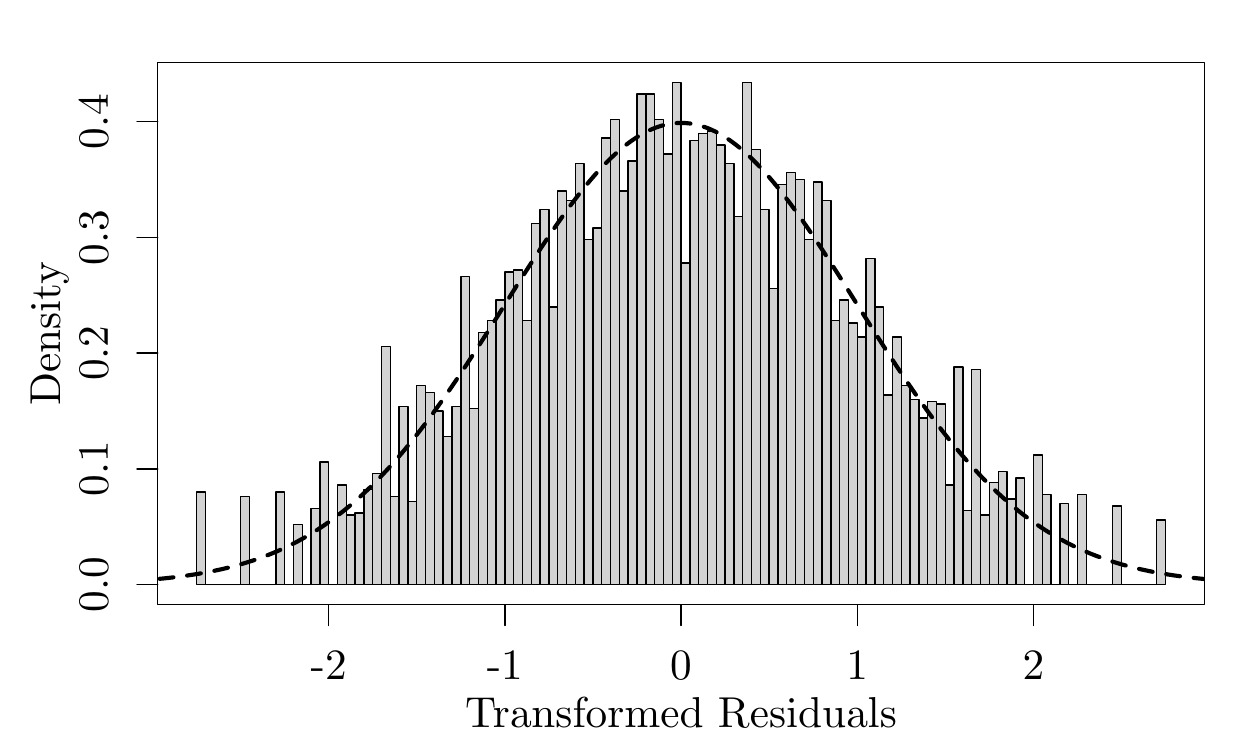} 
\includegraphics[width =0.49\linewidth]{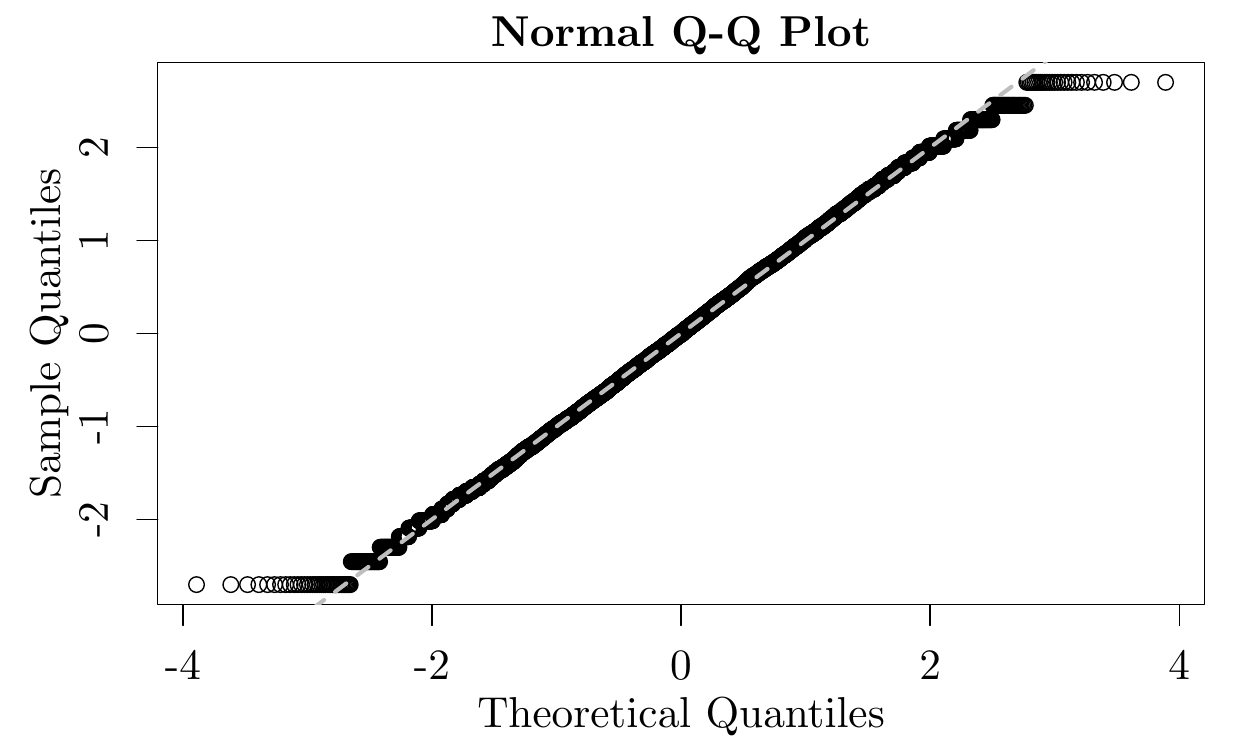}
\caption{Stock market data. \textbf{Left} Histogram of the transformed residuals compared with the standard Gaussian density. \textbf{Right} Q-Q Normal plot of the transformed residuals.}
\label{Fig:STOCK_qq_transformed}
\end{figure}

\subsection{The network predictors}
\label{ssec:NetworkLinearity_STOCK}

The price data from CRSP is arranged by TIC, a unique stock identifier, while the risk measures are arranged by CIK, a unique company identifier. Any two stocks associated with the same compnay had the same risk scores.

Based on the construction of \citet{bakerPolicyNewsStock2019}, we divided the 37 risk factors into two categories: the economic risks (containing 17 risks) and the policy risks (containing 20 risks) and standardised the sentence coutsn by the total number of sentences in the 10-K fillings. Then, for each risk type, we centered the $\log(1 + counts)$ and evaluated the Pearson's correlation between all pairs of companies to obtain two network matrices $E_{pears}$ and $P_{pears}$.

Figure \ref{Fig:GLASSO_vs_Stock2} demonstrates that for both networks there appears to be an increased chance of having positive partial correlation if the two firms have highly correlated risk factors. Figure \ref{fig:lincheck_STOCK} demonstrates that no further transformation of the networks is required to satisfy the network \GLASSO assumption of linearity.

%\jack{Edit the text here for the change in the plots - This plot is now in the paper}
%The fitted spike-and-slab priors illustrated at the bottom of Figure \ref{Fig:GLASSO_vs_Stock} show how the relationship between the networks and the graphical model is captured. The estimates parameterising these spike-and-slab distributions are available in Table \ref{Tab:SS_Stock_CIs}. We see that as the economics network increases, the slab location increases, the slab scale increases, and also the probability of being in the slab increases. Alternatively, when accounting for the economic network, an increase in the policy network corresponds to an increased probability of being in the slab, but not an increase in the slab location or scale.

\begin{figure}%[hbt!]
\centering
\includegraphics[width =0.49\linewidth]{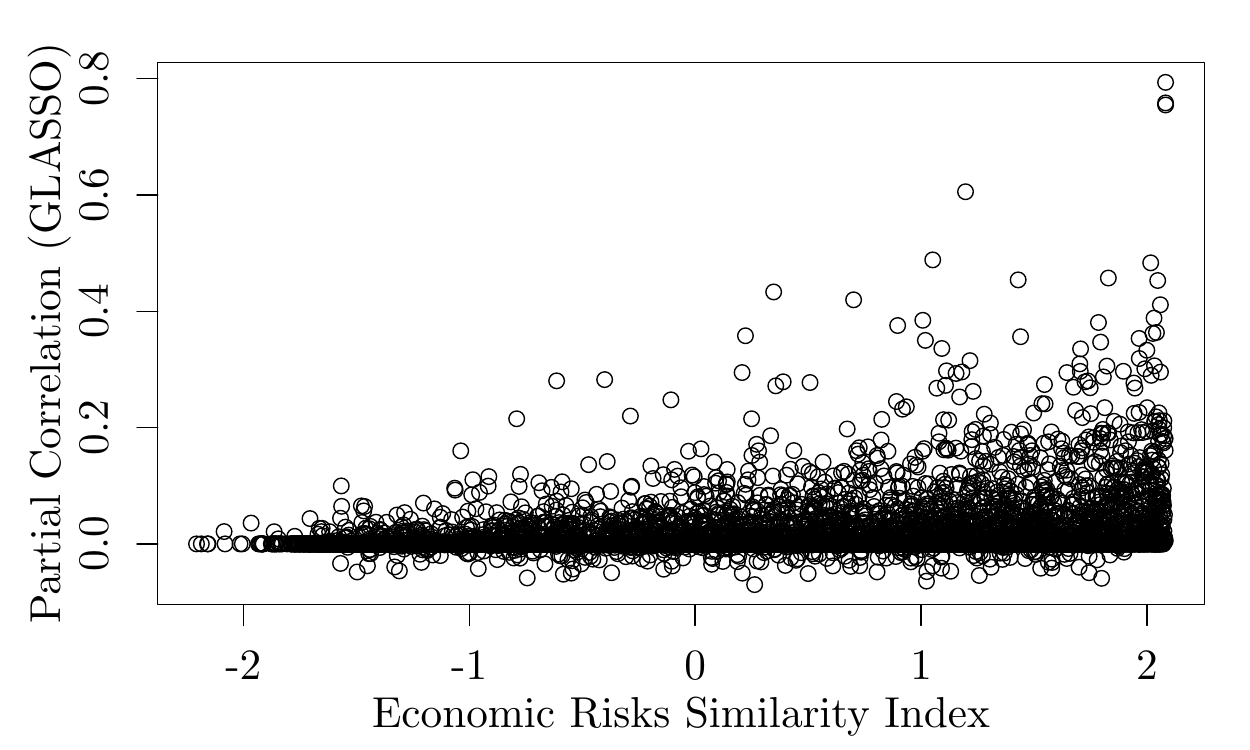}
\includegraphics[width =0.49\linewidth]{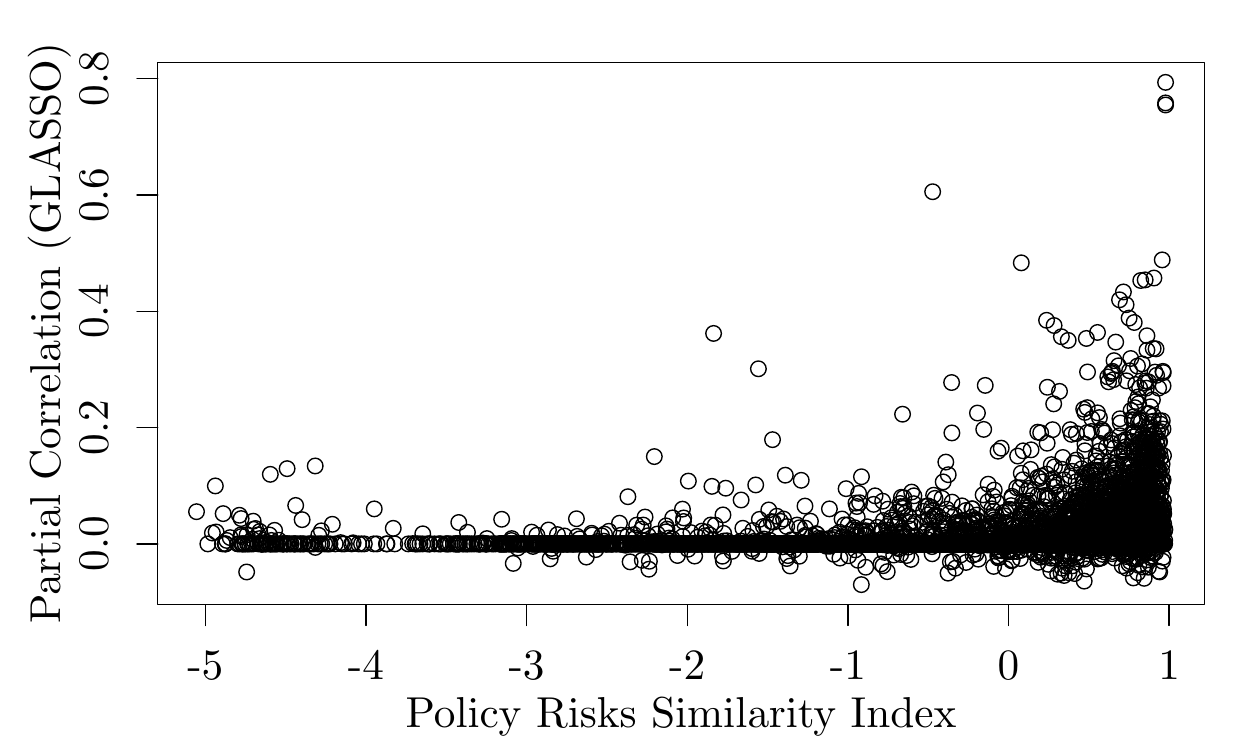}
\caption{Residual partial correlations of the stock market excess returns across firms vs Economy risk (left) and Policy risk (right). Partial correlations were estimated with \GLASSO, with penalization parameter set via \BIC.  % \EBIC $\gamma_{\EBIC} = 0.5$. 
%Bottom panel: fitted spike-and-slab distributions and fitted partial correlations estimated with network spike-and-slab model. %Bottom panel: prior slab probability as a function of both networks.
}
\label{Fig:GLASSO_vs_Stock2}
\end{figure}

\begin{figure}%[hbt!]
\centering
\includegraphics[width =0.49\linewidth]{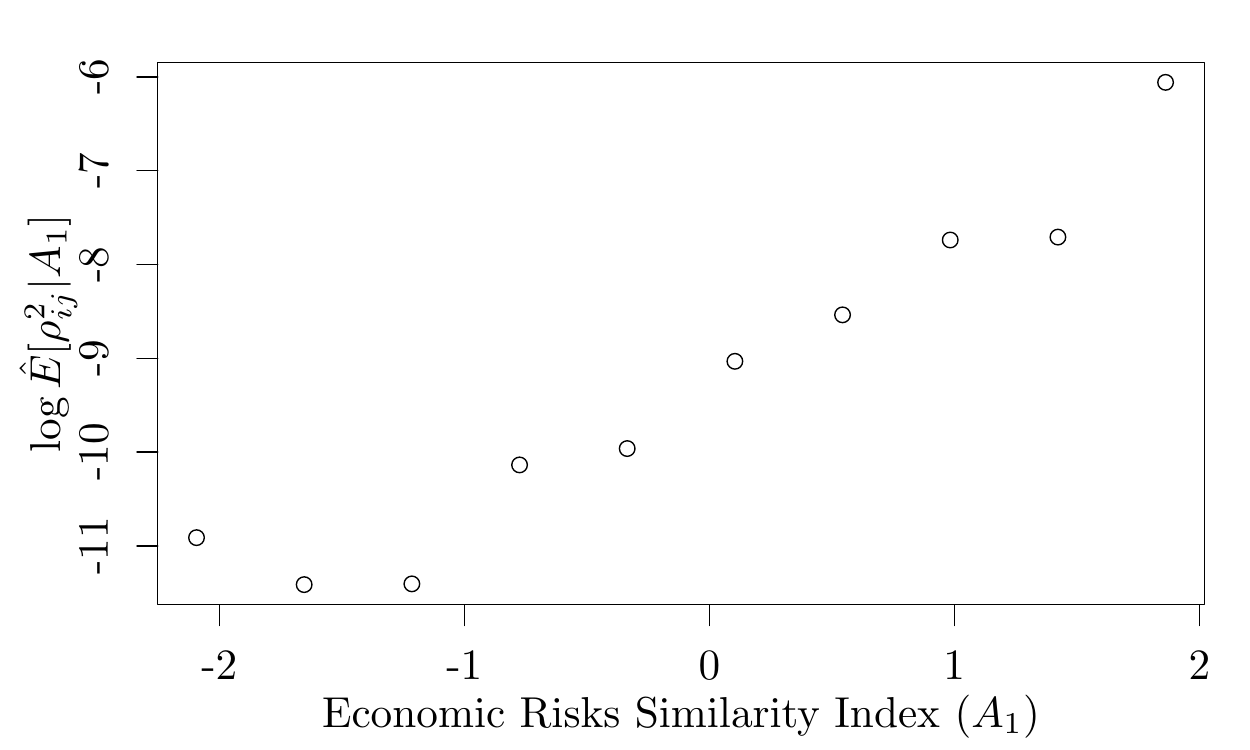}
\includegraphics[width =0.49\linewidth]{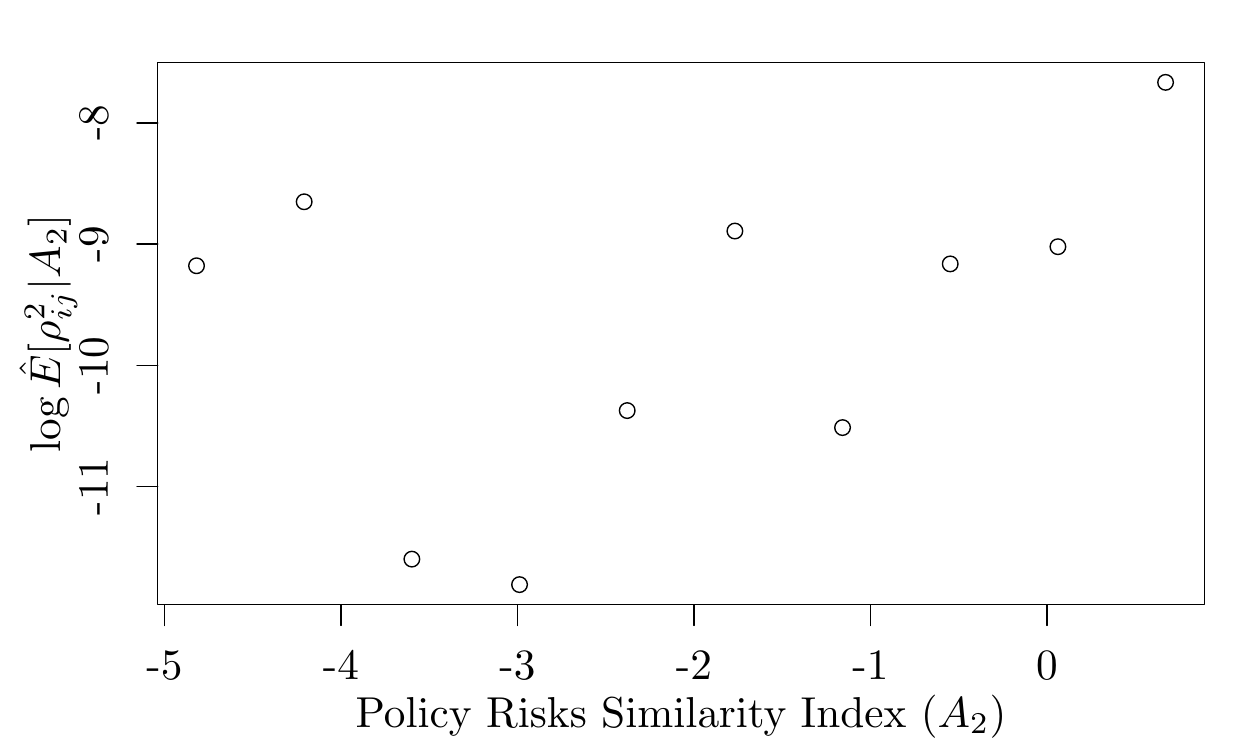}
%\includegraphics[width =0.49\linewidth]{plot/STOCK/GLASSO_vs_A1_2_0.5_tikz-1.pdf}
%\includegraphics[width =0.49\linewidth]{plot/STOCK/GLASSO_vs_A2_2_0.5_tikz-1.pdf}
%\caption{Assessing the linear relation between $\log E(\hat{\rho}_{jk}^2)$ the network matrices , where $\hat{\rho}_{jk}$ is the \GLASSO estimate. The scatters represent the mean $\log E(\rho_{jk}^2)$ for each of 10 equispaced bins defined by each network}
\caption{Assessing the linear relation between $\log \mathbb{E}[\hat{\rho}_{jk}^2| A]$ and the network matrices, where $\hat{\rho}_{jk}$ is the \GLASSO estimate. The points represent the log-mean values of $\hat{\rho}_{jk}^2$ within 10 equispaced bins defined for each network.
}
\label{fig:lincheck_STOCK}
\end{figure}

\subsection{Supplementary figures}

Table \ref{Tab:stock_edge_network_SS} summarises the estimated graphical model under the network spike-and-slab model using a posterior slab probability threshold of $> 0.5$ and $> 0.95$. The number of edges estimated under both slab probability threshold is smaller than the number of edges estimated under the network \GLASSO models. Under the 0.95 slab probability threshold, the estimated number of edges is more conservative.

\begin{table}[]
\centering
\caption{Stock market data: Edge counts of the network spike-and-slab model when declaring an edge for posterior slab probability $> 0.5$ and $> 0.95$}
\label{Tab:stock_edge_network_SS}
\begin{tabular}{ccccc}
\toprule{}     &     Edges ($> 0.5$) & Non-Edges ($> 0.5$)  &            Edges ($> 0.95$) & Non-Edges ($> 0.95$)   \\
\midrule Network SS  &             377 &                         66418 &                           189 & 66606 \\
\bottomrule
\end{tabular}
\end{table}

\subsection{Results using the \EBIC}{\label{Sec:EBIC_STOCK}}

Table \ref{tab:results_covid_0.5} presents results investigating the stability of our stock market data analysis to selecting hyperparameters using the \EBIC with $\gamma_{\EBIC} = 0.5$ rather than the \BIC. The \EBIC continues to estimate sparser networks than the \BIC, but to the detriment of the out-of-sample test set score. Importantly, we see that the improvement of the network \GLASSO methods over standard \GLASSO is still apparent when using the \EBIC selection criteria.

%OLD p = 200
%\begin{table}[ht]
%\centering
%%\caption{Comparison of four models for the stock market data when using the \EBIC ($\gamma_{\EBIC} = 0.5$) to learn the network hyperparameters. $A_1$ and $A_2$: networks defined by the Pearson's correlation between vectors of centered $\log(1 + count)$ for $ Economic$ and $Policy$ risk factor term counts. \EBIC values account for the extra hyperparameters learned in the network \GLASSO model. Test set score obtained using 10-fold cross-validation}
%\caption{Four models for the stock market data when using the \EBIC ($\gamma_{\EBIC} = 0.5$) to learn the network hyperparameters. $A_1$ is the Economic network, $A_2$ the Policy network. \EBIC values account for the extra hyper-parameters in the network \GLASSO models. 
%10-fold is the 10-fold cross-validation log-likelihood.}
%\begin{tabular}{ccccccc}
%  \hline
%  Method & \EBIC & $\hat{\beta}_0$ & $\hat{\beta}_1$ & $\hat{\beta}_2$ & Edges & 10-fold\\ 
%   \hline
%   \GLASSO & 50008.018 & -0.868 &  &  &   93 & -7020.84 \\ 
%   Network \GLASSO - $A_1$ & 49632.192 & -0.237 & -0.605 &  &  109 & \textbf{-6993.37} \\ 
%   Network \GLASSO - $A_2$ & 49723.101 & 1.632 &  & -1.868 &  100 & -7022.70 \\ 
%   Network \GLASSO - $A_1$ \& $A_2$ & \textbf{49587.239} & 0.889 & -0.361 & -1.028 &   78 & -7025.53 \\ 
%   \hline
%\end{tabular}
%\label{tab:results_stock_0.5}
%\end{table}

\begin{table}[ht]
\centering
\caption{Four models for the stock market data when using the \EBIC ($\gamma_{\EBIC} = 0.5$) to learn the network hyperparameters. $A_1$ is the Economic network, $A_2$ the Policy network. \EBIC values account for the extra hyper-parameters in the network \GLASSO models. 
10-fold is the 10-fold cross-validation log-likelihood. }
\begin{tabular}{ccccccc}
  \hline
  Method & \EBIC & $\hat{\beta}_0$ & $\hat{\beta}_1$ & $\hat{\beta}_2$ & Edges & 10-fold\\ 
   \hline
   \GLASSO & 88588.75 & -0.9106 &  &  &  616 & -494.781 \\ 
   Network \GLASSO - $A_1$ & 86140.75 & 3.350 & -2.604 &  &  572 & -493.508 \\ 
   Network \GLASSO - $A_2$ & 87675.87 & 0.531 &  & -2.081 &  732 & -494.090 \\ 
   Network \GLASSO - $A_1$ \& $A_2$ & \textbf{84150} & 10.289 & -4.020 & -5.718 &  468 & \textbf{-492.872} \\ 
   \hline
\end{tabular}
\label{tab:results_stock_0.5}
\end{table}

%\newpage
\section{\texttt{Stan} vs \texttt{NumPyro}}
\label{sec:stan_vs_numpyro}

%The large number of parameters to be estimated in Gaussian graphical models often makes Bayesian inference difficult. 
We estimated our network spike-and-slab models using the No-U-Turn Sampler (NUTS) \citep{hoffman2014no}, an extension of Hamiltonian Monte Carlo (HMC, \citealt{duane1987hybrid}) that automates the setting of the step-size in the Hamiltonian discretisation. Two probabilistic programming implementations of NUTS are \texttt{Stan} \citep{carpenter:2017} and \texttt{NumPyro} \citep{bingham:2019,phan:2019}. We provide implementations of our algorithm in both languages, but for our experiments, we found \texttt{NumPyro}'s ability to take advantage of parallel computing for automatic differentiation provided a considerable speed up. 

We illustrate this using one of our simulated examples from Section 4. We consider network matrix $A_{0.85}$, $n = 100$ and $p = 10$ and $p = 50$. We ran both \texttt{Stan} and \texttt{NumPyro} for 2000 warm-up iterations and 2000 sampling iterations. Table \ref{tab:stan_v_numpyro} compares the time taken to sample and the effective sample size (\ESS) of the resulting sample averaged across 10 repeat datasets. We present the \ESS as separately averaged across the $\rho$ model parameters and the network hyperparameter $\eta$. We see that both methods produce similar \ESS but that \texttt{NumPyro} does so over six times faster. 

We also take this opportunity to demonstrate how efficient the network \GLASSO is when implemented as a special case of the \GOLAZO algorithm \citep{lauritzen:2020}. For the same datasets considered above, we implement the network \GLASSO using $50\times 50$ grid search to estimate the network hyperparameters. We see that the \GOLAZO algorithm takes a fraction of the time to run as the Bayesian implementation even when using a rudimentary grid-search optimisation scheme.

%\begin{table}[ht]
%\centering
%\caption{Comparison of time taken for network \GLASSO implemented using the \GOLAZO algorithm and the network spike-and-slab sampling algorithms in \texttt{Stan} and \texttt{NumPyro}. \jack{This was before the sufficient statistics}}
%\begin{tabular}{cccc}
%  \hline
%   & Time (s) & ESS $\rho$'s & ESS $\eta$'s \\ 
%   \hline
%   \GOLAZO & 8.58 & - & - \\
%   \texttt{Stan} & 537.05 & 1089 & 555 \\ 
%   \texttt{NumPyro} & 47.15 & 1056 & 601 \\ 
%   \hline
%\end{tabular}
%\label{tab:stan_v_numpyro}
%\end{table}

\begin{table}[ht]
\centering
\caption{Comparison of time taken for network \GLASSO implemented using the \GOLAZO algorithm and the network spike-and-slab sampling algorithms in \texttt{Stan} and \texttt{NumPyro}.}
\begin{tabular}{cccc}
  \hline
   $p = 10$& Time (s) & \ESS $\rho$'s & \ESS $\eta$'s \\ 
   \hline
   \GOLAZO & 14.94 & - & - \\
   \texttt{Stan} & 184.93 & 855 & 373 \\ 
   \texttt{NumPyro} & 28.99 & 977 & 522 \\ 
   \hline
\end{tabular}
\vspace{2cm}
\begin{tabular}{cccc}
  \hline
   $p = 50$& Time (s) & \ESS $\rho$'s & \ESS $\eta$'s \\ 
   \hline
   \GOLAZO & 312.78 & - & - \\
   \texttt{Stan} & 7162.11 & 1663 & 268 \\ 
   \texttt{NumPyro} & 1133.057 & 1604 & 293 \\ 
   \hline
\end{tabular}
\label{tab:stan_v_numpyro}
\end{table}

Lastly, above we limited \texttt{NumPyro}'s access to only 6 cores on one machine. Using more cores, for example on a GPU, provides the potential for \texttt{NumPyro} to achieve even greater speed-ups for higher dimensional problems beyond the simple one considered here.

\end{document}